\documentclass[11pt,english]{article}

\usepackage[margin=1in]{geometry}
\usepackage{bbm}
\usepackage{graphicx,color}
\usepackage{amsmath,amssymb,amsthm}
\usepackage{enumerate}
\usepackage{enumitem}
\usepackage{fullpage}
\usepackage[ruled,linesnumbered,vlined]{algorithm2e}
\usepackage{mathrsfs}
\usepackage[noend]{algcompatible}
\usepackage[noblocks]{authblk}
\usepackage{varwidth}
\usepackage{mathrsfs}
\usepackage{babel}
\usepackage{accents}
\usepackage{mathtools}
\usepackage{tipa}
\usepackage[most]{tcolorbox}
\usepackage{wrapfig}
\usepackage{blindtext}
\usepackage{thm-restate}
\usepackage{longtable}
\usepackage[backref=page, pagebackref=true]{hyperref}
\usepackage[capitalize]{cleveref} 
\Crefname{enumtry}{try1}{}   
\newcommand{\rdrm}{\textsf{Word RAM~}}
\newcommand{\srt}{\mathsf{SORT}}
\newcommand{\ssa}{\mathsf{SSA}}

\DeclarePairedDelimiter{\nor2}{\lVert}{\rVert}

\newtheorem{question}{Question}
\newtheorem{goal}{Goal}

\def\denseformat{
\setlength{\textheight}{9in}
\setlength{\textwidth}{6.9in}
\setlength{\evensidemargin}{-0.2in}
\setlength{\oddsidemargin}{-0.2in}
\setlength{\headsep}{10pt}
\setlength{\topmargin}{-0.3in}
\setlength{\columnsep}{0.375in}
\setlength{\itemsep}{0pt}
}

\newtheorem{theorem}{Theorem}[section]

\newtheorem{definition}[theorem]{Definition}
\newtheorem{claim}[theorem]{Claim}
\newtheorem{lemma}[theorem]{Lemma}

\newtheorem{remark}[theorem]{Remark}

\newtheorem{observation}[theorem]{Observation}

\def\boldhead#1:{\par\vskip 7pt\noindent{\bf #1:}\hskip 10pt}
\def\ithead#1:{\par\vskip 7pt\noindent{\it #1:}\hskip 10pt}

\def\inline#1:{\par\vskip 7pt\noindent{\bf #1:}\hskip 10pt}
\def\midinline#1:{\par\noindent{\bf #1:}\hskip 10pt}
\def\dnsinline#1:{\par\vskip -7pt\noindent{\bf #1:}\hskip 10pt}
\def\ddnsinline#1:{\newline{\bf #1:}\hskip 10pt}
\def\largeinline#1:{\par\vskip 7pt\noindent{\large\bf #1:}\hskip 10pt}
\long\def\commhide #1\commhideend{}
\long\def\commfull #1\commend{#1}
\long\def\commabs #1\commenda{}
\long\def\commtim #1\commendt{#1}
\long\def\commb #1\commbend{}
\long\def\commedit #1\commeditend{} 

\long\def\commB #1\commBend{}       

\long\def\commex #1\commexend{}     

\long\def\commsiena #1\commsienaend{}  

\long\def\commBI #1\commBIend{}  

\long\def\CProof #1\CQED{}
\def\qed{\mbox{}\hfill $\Box$\\}
\def\blackslug{\hbox{\hskip 1pt \vrule width 4pt height 8pt
    depth 1.5pt \hskip 1pt}}
\def\QED{\quad\blackslug\lower 8.5pt\null\par}

\long\def\PPP#1{\noindent{\bf Proof:}{ #1}{\quad\blackslug\lower 8.5pt\null}}

\long\def\denspar #1\densend
{#1}

\newif\ifnotesw\noteswtrue
   {\ifnotesw\marginpar[\hfill\(\top\)]{\(\top\)}\fi}%
   {\ifnotesw\marginpar[\hfill\(\bot\)]{\(\bot\)}\fi}

\newcommand{\mnote}[1]%
    {\ifnotesw\marginpar%
        [{\scriptsize\it\begin{minipage}[t]{\marginparwidth}
        \raggedleft#1%
                        \end{minipage}}]%
        {\scriptsize\it\begin{minipage}[t]{\marginparwidth}
        \raggedright#1%
                        \end{minipage}}%
    \fi}

\def\MathF{\hbox{\rm I\kern-2pt F}}
\def\MathP{\hbox{\rm I\kern-2pt P}}
\def\MathR{\hbox{\rm I\kern-2pt R}}
\def\MathZ{\hbox{\sf Z\kern-4pt Z}}
\def\MathN{\hbox{\rm I\kern-2pt I\kern-3.1pt N}}
\def\MathC{\hbox{\rm \kern0.7pt\raise0.8pt\hbox{\footnotesize I}
\kern-4.2pt C}}
\def\MathQ{\hbox{\rm I\kern-6pt Q}}

\newsavebox{\ttop}\newsavebox{\bbot}

\def\eps{\epsilon}
\def\epsi{\varepsilon}

\newcommand{\mst}{\mathrm{MST}}

\newcommand{\pathg}{\mathsf{path~greedy}}
\newcommand{\greedy}{\mathsf{greedy}}
\newcommand{\msttilde}{\widetilde{\mst}}
\newcommand{\Ftilde}{\widetilde{F}}
\newcommand{\Ttilde}{\widetilde{T}}
\newcommand{\Ptilde}{\widetilde{P}}
\newcommand{\tm}{\mathsf{Time}}
\newcommand{\source}{\mathsf{source}}
\newcommand{\lt}{\mathsf{Light}}

\newcommand{\high}{\mathsf{high}}
\newcommand{\lowp}{\mathsf{low}^+}
\newcommand{\lowm}{\mathsf{low}^-}
\newcommand{\cone}{\mathsf{Cone}}
\newcommand{\prune}{\mathsf{pruned}}
\newcommand{\geom}{\mathsf{Geom}}
\newcommand{\gen}{\mathsf{Gen}}
\newcommand{\minor}{\mathsf{Minor}}
\newcommand{\internal}{\mathsf{intrnl}}
\newcommand{\prefix}{\mathsf{pref}}

\def\eps{\epsilon}

\DeclareMathAlphabet{\mathpzc}{OT1}{pzc}{m}{it}

\newcommand{\dm}{\mathsf{Dm}}
\newcommand{\adm}{\mathsf{Adm}}

\newcommand {\ignore} [1] {}

\denseformat

\usepackage{lineno}

\newcommand{\ma}{\mathcal{A}}
\newcommand{\mb}{\mathcal{B}}
\newcommand{\mc}{\mathcal{C}}
\newcommand{\md}{\mathcal{D}}

\newcommand{\mg}{\mathcal{G}}
\newcommand{\mv}{\mathcal{V}}
\newcommand{\me}{\mathcal{E}}

\renewcommand{\mp}{\mathcal{P}}

\newcommand{\mk}{\mathcal{K}}

\newcommand{\mx}{\mathcal{X}}
\newcommand{\my}{\mathcal{Y}}
\newcommand{\mz}{\mathcal{Z}}

\newcommand{\mi}{\mathcal{I}}

\newcommand{\mbe}{\mathbf{e}}

\DeclareMathOperator{\MST}{\mathrm{MST}}

\DeclareMathOperator{\defi}{\overset{\mathrm{def.}}{=} }

\date{}

\title{A Unified Framework of Light Spanners I: \\Fast (Yet Optimal) Constructions\thanks{This paper is the first of two papers that correspond together to our STOC 2023 paper, titled ``A Unified Framework for Light Spanners''. The STOC 2023 paper contains the formal statements of the main results, providing only sketched proofs or no proofs for most of the results. These two papers extend the STOC 2023 paper significantly, containing together the full details and proofs of all results. The preprint of the other paper can be found at \url{https://arxiv.org/abs/2111.13748}}}
\author{Hung Le}
\affil{University of Massachusetts Amherst}
\author{Shay Solomon}
\affil{Tel Aviv University}

\begin{document}
\pagenumbering{gobble}
\maketitle
\begin{abstract}
We present a {\em unified framework} for constructing light spanners in a variety of graph classes. Informally, the framework boils down to a {\em transformation} from sparse spanners to light spanners; since the state-of-the-art for sparse spanners is much more advanced than that for light spanners, such a transformation is powerful. Our framework is developed in two papers. The current paper is the first of the two --- it lays the {\em basis of the unified framework} and then applies it to design {\em fast} constructions with {\em optimal lightness} for several graph classes.  Our new constructions are significantly faster than the state-of-the-art for every graph class studied in this paper; the running times of our constructions are near-linear and usually optimal. 

Among various applications and implications of our framework, we highlight here the following (for simplicity assume $\eps> 0$ is fixed):
\begin{itemize}
\item In \emph{low-dimensional Euclidean spaces}, we present a construction of $(1+\eps)$-spanners for $n$-point sets with lightness and degree both bounded by constants,
running in $O(n\log n)$ time in the {\em algebraic computation tree (ACT)} (or {\em real-RAM}) model, which is the basic model used in Computational Geometry.  Our construction is optimal with respect to all the involved quality measures --- running time, lightness, and degree --- and it resolves a major problem in the area of geometric spanners, which was open for three decades.  

\item In \emph{general graphs}, we present a near-linear time algorithm for constructing light spanners of graphs with $n$ vertices and $m$ edges. Specifically, for any $k \ge 2$, we construct a $(2k-1)(1+\epsilon)$-spanner  with lightness $O(n^{1/k})$  in $O(m \alpha(m,n))$ time, where $\alpha(\cdot,\cdot)$ is the inverse-Ackermann function ; the lightness bound matches Erd\H{o}s' girth conjecture up to the $\eps$-dependency.
\end{itemize}

\paragraph{Remark} 
Our companion paper builds on the basis laid in this paper,  aiming to achieve optimality in a more refined sense, which takes into account a \emph{wider range of involved parameters}, most notably $\eps$, but also others such as the Euclidean dimension or the minor size (in minor-free graphs). 
\end{abstract}

\pagebreak

\tableofcontents

\pagenumbering{arabic}

\clearpage

\section{Introduction}\label{sec:intro}

For an edge-weighted graph $G = (V,E,w)$ and a {\em stretch parameter} $t \ge 1$, a subgraph $H = (V,E')$ of $G$ is called a \emph{$t$-spanner} if $d_H(u,v) \le t \cdot d_G(u,v)$, for every two vertices $u$ and $v$, where $d_G(u,v)$ and $d_H(u,v)$ are the distances between $u$ and $v$ in $G$ and $H$, respectively.
Graph spanners were introduced in two celebrated papers from 1989 \cite{PS89,PU89} for unweighted graphs, where it is shown that for any $n$-vertex graph $G = (V,E)$ and integer $k \ge 1$, there is an $O(k)$-spanner with $O(n^{1+ 1/k})$ edges.
We shall sometimes use a normalized notion of size, {\em sparsity}, which is the ratio of the size of the spanner to the size of a spanning tree, namely $n-1$.
Since then, graph spanners have been extensively studied, both for general weighted graphs and for restricted graph families, such as Euclidean spaces and minor-free graphs.  In fact, spanners for Euclidean spaces---{\em Euclidean spanners}---were studied implicitly already in the pioneering SoCG'86 paper of Chew~\cite{Chew86}, who showed that any finite point set in 2-dimensional Euclidean space admits a spanner of $O(n)$ edges and stretch $\sqrt{10}$, and later improved the stretch to 2~\cite{Chew89}.

As with the sparsity parameter, its weighted variant, lightness, has been extremely well-studied; the \emph{lightness} is the ratio of the weight of the spanner to $w(MST(G))$.  Seminal works on {\em light} spanners over the years provide spanners with optimal {\em lightness} in various graph classes, such as in general graphs~\cite{CW16}, Euclidean spanners \cite{das1994fast} and minor-free graphs~\cite{BLW17}. {\bf Despite the large body of work on light spanners,  the stretch-lightness tradeoff is not nearly as well-understood as the stretch-sparsity tradeoff}, and the intuitive reason behind that is clear: Lightness seems inherently more challenging to optimize than sparsity since different edges may contribute disproportionately to the overall lightness due to differences in their weights.  The three shortcomings of light spanners that emerge, when considering the large body of work in this area, are: (1) The techniques are ad hoc per graph class and thus can't be applied broadly (e.g., some require large stretch and are thus suitable to general graphs, while others are naturally suitable to stretch $1 + \eps$). 
(2) The running times of these constructions are usually far from optimal.
(3) These constructions are optimal in the standard and crude sense but not in a refined sense that takes into account a wider range of involved parameters.

We set out to address these shortcomings by presenting a {\em unified framework} of light spanners in a variety of graph classes. Informally, the framework boils down to a {\em transformation} from sparse spanners to light spanners; since the state-of-the-art for sparse spanners is much more advanced than that for light spanners, such a transformation is powerful.

Our framework is developed in two papers.  {\bf The current paper is the first of the two --- \bf it lays the {\em basis of the unified framework} and then applies it to design {\em fast} constructions with {\em optimal lightness}} for several graph classes. More specifically, this paper will address the first two shortcomings mentioned above, while the third shortcoming will be addressed by the second paper.  Our ultimate goal is to bridge the gap in the understanding between light and sparse spanners.  This gap is very prominent when considering the construction time.  To exemplify this statement, we next survey results on light spanners in several basic graph classes, focusing mostly on the construction time. Subsequently, we present our new constructions, all of which are derived as applications and implications of the unified framework developed in this work; our constructions are significantly faster than the state-of-the-art {\em for every examined graph class}; as will be shown, our running times are near-linear or linear and usually optimal.

\paragraph{Euclidean spanners in the algebraic computation tree (ACT) model} 
Spanners have had special success in geometric settings, especially in low-dimensional Euclidean spaces.  The reason Euclidean spanners have been extensively studied over the years --- in both theory and practice --- is that one can achieve stretch arbitrarily close to 1 together with constant sparsity and lightness (ignoring dependencies on $\eps$ and the dimension $d$).
In general metrics, on the other hand, a stretch better than 3 requires sparsity and lightness of $\Omega(n)$.   The {\em algebraic computation tree (ACT)} model is used extensively in computational geometry, and in the area of Euclidean spanners in particular; this model, introduced by Ben-Or~\cite{BenOr83}, is intimately related (and equivalent, if ignoring uniformity issues) to the real random access machine (real RAM) model. (The reader can refer to \cite{BenOr83} and Chapter 3 in the book \cite{NS07} for a detailed description of the ACT model;  we provided a brief description of this model in~\Cref{algapp}.)

In the ACT model, computing $(1+\epsilon)$-spanners for point sets in $\mathbb{R}^d$, $d = O(1)$, requires $\Omega(n\log n)$ time~\cite{CDS01,FP16}.\footnote{If one allows to use {\em indirect addressing}, then the lower bound of $\Omega(n \log n)$ no longer applies.}  There are various algorithms (see, e.g.,~\cite{CK95,KG92,Salowe91,Vaidya1991}) for computing $(1+\epsilon)$-spanners with constant sparsity, which achieve an optimal running time of $O(n\log n)$ in this model, for any fixed $\eps \in (0,1)$ and in any constant-dimensional Euclidean space. However, the lightness of the spanners produced by those algorithms is unbounded.
 
Starting in the late 80s, there has been a large body of work on light Euclidean spanners~\cite{LL89,CDNS92,DHN93,das1994fast,DNS95,das1996constructing,ADMSS95,RS98,GLN02,NS07,elkin2015optimal,LS19}. Light Euclidean spanners are not only important in their own right, but they also find applications in other contexts.
In particular, the breakthrough result of Rao and Smith \cite{RS98} gave an $O(n \log n)$-time approximation scheme for the Euclidean TSP, assuming that a $(1+\eps)$-spanner with constant lightness can be computed within time $O(n \log n)$. 
Also, Czumaj and Lingas \cite{czumaj2000fast} gave approximation schemes for Euclidean minimum-cost multi-connectivity problems under the same assumption.  
The assumption used in the results of \cite{RS98,czumaj2000fast} was made by relying on a spanner construction due to Arya et al.\ \cite{ADMSS95}, which was later shown to be flawed.  Gudmundsson, Levcopoulos, and Narasimhan (hereafter, GLN)~\cite{GLN02}, building on and improving over several previous works \cite{LL89,CDNS92,DHN93,das1994fast,DNS95,ADMSS95}, gave the first (correct) algorithm for constructing Euclidean $(1+\eps)$-spanners with constant lightness in $O(n \log n)$ running time, but their algorithm assumes {\em indirect addressing}.
A variation of the GLN algorithm, which applies to the ACT model, takes time $O(n\frac{\log^2n}{\log \log n})$; this is the state-of-the-art running time for constructing $(1+\epsilon)$-spanners with constant lightness in the ACT model, even in $\mathbb{R}^2$, and even allowing a super-constant lightness bound (of at most $o(\log n)$).   The question of whether one can compute such a spanner in optimal $O(n \log n)$ time in the ACT model was asked explicitly several times, including in the GLN paper~\cite{GLN02} and in the spanner book by Narasimhan and Smid \cite{NS07}.

 \begin{question} \label{ques:algebraic}
 	Can one   construct a Euclidean $(1+\epsilon)$-spanner with constant lightness
	within the optimal time of $O(n \log n)$ in the ACT model for any fixed $\eps \in (0,1)$?
 \end{question}

Constant lightness does not imply any sparsity bound. A stronger result would be to achieve a constant bound on both the lightness and   sparsity and even further, one could try to achieve a constant bound on the maximum degree too;
indeed,  there are $O(n \log n)$-time constructions of Euclidean spanners of bounded degree in the ACT model \cite{ADMSS95,NS07}. Euclidean spanners of bounded degree have found various applications.  In compact routing schemes low degree spanners give rise to routing tables of small size (see, e.g., ~\cite{CGMZ16,DBLP:conf/soda/GottliebR08,DBLP:conf/cocoon/BrankovicGR20}), and more generally, the (maximum) degree of the spanner determines the local memory constraints when using spanners also for other purposes, such as constructing network synchronizers and efficient broadcast protocols.  Moreover, in some applications, the degree of a vertex (or processor) represents its {\em load}, hence a low degree spanner guarantees that the load on all the processors in the network will be low.

 \begin{question} [Question 22 in \cite{NS07}] \label{ques:algebraic2}
 	Can one construct a Euclidean $(1+\epsilon)$-spanner with constant lightness and maximum degree (and thus constant sparsity) in optimal time $O(n \log n)$ in the ACT model, for any fixed $\eps \in (0,1)$? 
 \end{question}

\paragraph{General weighted graphs} 
The aforementioned results of \cite{PS89,PU89} for general graphs were strengthened in \cite{ADDJS93}, where it was shown that for every $n$-vertex \emph{weighted} graph $G = (V,E,w)$ and integer $k \ge 1$, there is a {\em greedy} algorithm for constructing a $(2k-1)$-spanner with $O(n^{1+1/k})$ edges, which is optimal under Erd\H{o}s' girth conjecture. Moreover, there is an $O(m)$-time algorithm for constructing $(2k-1)$-spanners in unweighted graphs with sparsity $O(n^{\frac{1}{k}})$~\cite{HZ96}. Therefore, not only is the stretch-sparsity tradeoff in general graphs optimal (up to Erd\H{o}s' girth conjecture), but one can achieve it in optimal time. For weighted graphs, one can construct $(2k-1)$-spanners with sparsity $O(k n^{\frac{1}{k}})$ within time $O(k m)$~\cite{baswana2007simple,roditty2005deterministic}.

Alth\"{o}fer et al.~\cite{ADDJS93}   showed that the lightness of the greedy spanner is $O(n/k)$. Chandra et al.~\cite{CDNS92} improved this lightness bound to $O(k \cdot n^{(1+\eps)/{(k-1)}} \cdot (1/\eps)^2)$, for any $\eps > 0$; another, somewhat stronger, form of this tradeoff from \cite{CDNS92}, is stretch $(2k-1)\cdot(1+\eps)$,
$O(n^{1/k})$ sparsity and $O(k \cdot n^{1/{k}} \cdot (1/\eps)^{2})$ lightness.
In a sequence of works from recent years \cite{ENS14,CW16,FS16},
it was shown that the lightness of the greedy spanner is $O(n^{1/k} (1/\eps)^{3+2/k})$ (this lightness bound is due to \cite{CW16}; the fact that this bound holds for the greedy spanner is due to \cite{FS16}). The best running time for the same lightness bound in prior work is super-quadratic in $n$: $O(n^{2 + 1/k + \epsilon'})$~\cite{ADFSW19} for any fixed constant $\epsilon' < 1$. 

\begin{question}\label{ques:gen}
 	Can one construct a $(2k-1)(1+\eps)$-spanner in general weighted graphs with lightness $O(n^{1/k})$, within near-linear time for any fixed $\eps \in (0,1)$?
 \end{question}

\paragraph{Unit disk graphs}
Given a set of $n$ points $P \subseteq \mathbb{R}^d$, a \emph{unit ball graph} for $P$, denoted by $U = U(P)$, 
is the \emph{geometric graph} with vertex set $P$, where there is an edge between two points $p \not= q \in P$ (with weight $\nor2{p,q}$) iff $\nor2{p,q}\leq 1$.\footnote{Throughout we use $\nor2{p,q}$ to denote the Euclidean distance between a pair $p,q$ of points in $\mathbb{R}^d$.} When $d = 2$, we call $U$ a \emph{unit disk graph (UDG)}; 
for convenience, we'll use the term {\em unit disk graph} also for $d > 2$.
(See~\Cref{sec:prelim} for a more detailed discussion on geometric graphs.)

There is a large body of work on spanners for UDGs; 
see~\cite{LCW02,LiCWW03,LW04,GaoGHZZ05,WangL07,PK08,PelegR10,FurerK12,biniaz2020plane}, and the references therein.
One conclusion that emerges from the previous work (see \cite{PelegR10} in particular) is that if one does not care about the running time, then constructing $(1+\eps)$-spanners for unit disk graphs is just as easy as constructing $(1+\eps)$-spanners for the entire Euclidean space. Moreover, the greedy $(1+\eps)$-spanner for the Euclidean space, after removing from it all edges of weight larger than 1, provides a $(1+\eps)$-spanner for the underlying unit disk graph. The greedy $(1+\eps)$-spanner in $\mathbb{R}^d$ has constant sparsity and lightness for constant $\eps$ and $d$, specifically, sparsity $\Theta(\eps^{-d+1})$ and lightness $O(\eps^{-d}\log(1/\eps))$, which is tight up to the $\log(1/\eps)$ factor (cf.\ \cite{LS19}).

The drawback of the greedy spanner is its running time: The state-of-the-art implementation in Euclidean low-dimensional spaces runs in $O(n^2 \log n)$ \cite{BCFMS10}. There is a much faster variant of the greedy algorithm, sometimes referred to as ``approximate-greedy'', with running time $O(n \log n)$ \cite{GLN02}. Alas, removing the edges of weight larger than 1 from the approximate-greedy $(1+\eps)$-spanner of the Euclidean space does not provide a $(1+\eps)$-spanner for the underlying unit disk graph; in fact, the stretch of the resulting spanner may be arbitrarily poor.
Instead of simply removing the edges of weight larger than 1 from the approximate-greedy spanner, one can replace them with appropriate replacement edges, as proposed in \cite{PelegR10}, but the running time of this process will be at least linear in the size of the unit disk graph, which is $\Omega(n^2)$ in the worst case. 

F{\"u}rer and Kasiviswanathan~\cite{FK06} showed that {\em sparse} $(1+\epsilon)$-spanners for UDGs can be built in near-linear time when $d = 2$, and in subquadratic time when $d$ is a constant of value at least $3$.

\begin{lemma}[Corollary 1 in~\cite{FurerK12}]\label{lm:UDG-sparse-sp} Given a set of $n$  points $P$ in $\mathbb{R}^d$, there is an algorithm that constructs  a $(1+\epsilon)$-spanner of the unit ball graph for $P$ with $O(n\epsilon^{1-d})$ edges. For $d = 2$, the running time  is  $O(n(\epsilon^{-2}\log n))$; for $d = 3$, the running time is $\tilde{O}(n^{4/3}\eps^{-3})$; and for $d\geq 4$, the running time is $O(n^{2-\frac{2}{(\lceil d/2 \rceil+1)} + \delta}\epsilon^{-d+1} + n\epsilon^{-d})$ for any constant $\delta > 0$.
\end{lemma}

Thus, there is a significant gap between the fastest constructions of sparse versus light spanners in UDGs.
In particular, no $o(n^2)$-time $(1+\eps)$-spanner construction for UDGs with a nontrivial lightness bound is known, even for $d = 2$. 
The question of closing this gap naturally arises.

 \begin{question} \label{ques:udg}
 	Can one construct within $o(n^2)$ time a $(1+\epsilon)$-spanner for UDGs with constant lightness for a fixed $\eps \in (0,1)$?  
	Is it possible to achieve a near-linear running time for $d = 2$?
 \end{question}

\paragraph{Minor-free graphs}
Alth\"{o}fer et al.~\cite{ADDJS93}   showed that the greedy $(1+\eps)$-spanner in planar graphs has lightness $O(1/\eps)$. Klein~\cite{Klein05} gave a fast construction of $(1+\eps)$-spanners with constant lightness (albeit with a worse dependence on $\eps$). It is known that the technique of~\cite{Klein05} can be extended to bounded genus graphs, provided that an embedding into a surface of the required genus is given as input; the time for computing such an embedding is linear in the graph size and exponential in the genus.

A natural goal would be to extend the results to minor-free graphs.\footnote{A graph $H$  is called a \emph{minor} of graph $G$ if $H$ can be obtained from $G$ by deleting edges and vertices and by contracting edges. A graph $G$ is said to be {\em $K_r$-minor-free}, if it excludes  $K_r$ as a minor for some fixed $r$, where $K_r$ is the complete graph on $r$ vertices. We shall omit the prefix $K_r$ in the term ``$K_r$-minor-free'',  when the value of $r$ is not important.}  
Borradaile, Le, and Wulff-Nilsen~\cite{BLW17} showed that the greedy $(1+\eps)$-spanners of $K_r$-minor-free graphs have lightness $\tilde{O}_{r,\epsilon}(\frac{r}{\epsilon^3})$, where the notation $\tilde{O}_{r,\epsilon}(.)$ hides polylog factors of $r$ and $\frac{1}{\epsilon}$. However,  the fastest implementation of the greedy spanner requires quadratic time~\cite{ADDJS93}, even in graphs with $O(n)$ edges; more generally, the running time of the greedy algorithm from~\cite{ADDJS93} on a graph with $m = \tilde{O}_r(nr)$ edges is $\tilde{O}_r(n^2 r^2)$.  Moreover, the same situation occurs even in sub-classes of minor-free graphs, particularly bounded treewidth graphs.

 \begin{question} \label{ques:minor}
 	Can one construct in linear or near-linear time a $(1+\epsilon)$-spanner for minor-free graphs with constant lightness?  
 \end{question}

\subsection{Research Agenda: From Sparse to Light Spanners}

Thus far, we exemplified the statement that the stretch-lightness tradeoff is not as well-understood as the stretch-sparsity tradeoff when considering the construction time.  Even when ignoring the running time, there are significant gaps between these tradeoffs
when considering {\em fine-grained dependencies}, i.e., when considering these tradeoffs in a wider range of involved parameters, most notably $\eps$, but also other parameters, such as the dimension (in Euclidean spaces) or the minor size (in minor-free graphs). This statement is not to underestimate in any way the exciting line of work on light spanners but rather to call for attention to the important research agenda of narrowing this gap and ideally closing it.

\paragraph{Fast constructions}
All questions above, from \Cref{ques:algebraic} to 
\Cref{ques:minor}, ask the same thing: Can one achieve {\em fast constructions} of light spanners that {\em match} the corresponding results for sparse spanners? 

\begin{goal} \label{g1}
Achieve {\em fast constructions} of light spanners that {\em match} the corresponding constructions of sparse spanners. 
In particular, achieve nearly linear-time constructions of spanners with optimal lightness for basic graph families, such as the ones covered in the aforementioned questions. 
\end{goal}

\paragraph{Fine-grained optimality}  A fine-grained optimization of the stretch-lightness tradeoff, which takes into account the exact dependencies on $\eps$ and the other involved parameters, is a highly challenging goal. For planar graphs, the aforementioned result~\cite{ADDJS93} on the greedy $(1+\eps)$-spanner with lightness $O(1/\eps)$ provides an optimal dependence on $\eps$ in the lightness bound, due to a matching lower bound.
For constant-dimensional Euclidean spaces, an optimal tradeoff of stretch $1+\eps$ versus lightness $\Theta(\eps^{-d})$ was achieved recently by the authors \cite{LS19}.   
Can one achieve such fine-grained optimality for other well-studied graph families, such as general graphs and minor-free graphs?
\begin{goal} \label{g2}
Achieve fine-grained optimality for light spanners in basic graph families. 
\end{goal}

\paragraph{Unification}
Some of the papers on light spanners employ inherently different techniques than others, e.g., the technique of \cite{CW16} requires large stretch while others are naturally suitable to stretch $1+\eps$.
Since the techniques in this area are ad hoc per graph class, they can't be applied broadly.
A unified framework for light spanners would be of both theoretical and practical merit.
\begin{goal} \label{g3}
Achieve a unified framework of light spanners.
\end{goal}

Establishing a thorough understanding of light spanners by meeting (some of) the above goals is not only of theoretical interest but is also of practical importance due to the wide applicability of spanners.  Perhaps the most prominent applications of light spanners are to efficient broadcast protocols in the message-passing model of distributed computing \cite{ABP90,ABP91}, network synchronization and computing global functions \cite{Awerbuch85,PU89,ABP90,ABP91,Peleg00}, and the TSP \cite{Klein05,Klein06,RS98,GLN02,BLW17,Gottlieb15}. There are many more applications, such as data gathering and dissemination tasks in overlay networks \cite{BKRCV02,VWFME03,KV01},  VLSI circuit design \cite{CKRSW91,CKRSW292,CKRSW92,SCRS01}, wireless and sensor networks \cite{RW04,BDS04,SS10},  routing \cite{WCT02,PU89,PU89b,TZ01},  and computing almost shortest paths \cite{Cohen98,RZ11,Elkin05,EZ06,FKMSZ05}, and distance oracles and labels \cite{Peleg00Prox,TZ01b,RTZ05}. 

\subsection{Our Contribution}

Our work aims at meeting the above goals (\Cref{g1}---\Cref{g3}) by presenting a unified framework for optimal constructions of light spanners in a variety of graph classes. Basically, we strive to translate results --- in a unified manner --- from sparse spanners to light spanners without significant loss in any of the parameters. One of our results is particularly surprising --- Theorem~\ref{thm:general-fast}, for general graphs --- {\bf since the new bounds for light spanners outperform the best-known bounds for sparse spanners}.

As mentioned, the current paper lays the {\em basis of the framework} and applies it to design {\em fast} constructions with {\em optimal lightness} for several graph classes, thereby resolving all aforementioned questions.  Our companion paper builds on the basis laid in this paper,  aiming to achieve fine-grained optimality.  

Next, we elaborate on the applications and implications of our framework and put it into context with previous work. For simplicity, we shall assume here that $\eps \in (0,1)$ is fixed;  the exact dependencies of $\eps$ will be explicated in subsequent sections of this paper.

\paragraph{Euclidean spanners in the ACT model} 
We present a spanner construction that achieves constant lightness and degree within  optimal time of $O(n \log n)$ in the ACT model; this
 proves the following theorem, which affirmatively resolves~\Cref{ques:algebraic2}, and thus~\Cref{ques:algebraic}, which was open for   three decades.

\begin{theorem}\label{thm:ACT} For any set $P$ of $n$ points in $\mathbb{R}^d$, any $d = O(1)$ and any  fixed $\eps \in (0,1)$, one can construct in the ACT model   a $(1+\epsilon)$-spanner for $P$ with constant degree and lightness  within optimal time $O(n \log n)$.
\end{theorem}

\paragraph{General graphs}

For general graphs, we provide a nearly linear-time spanner construction with a nearly optimal lightness in the worst-case sense, assuming Erd\H{o}s' girth conjecture, thus answering~\Cref{ques:gen}. 

\begin{theorem}\label{thm:general-fast}
	For any edge-weighted graph $G(V,E)$, a stretch parameter $k \geq 2$ and an arbitrary small fixed $\epsilon \in (0,1)$, there is a deterministic algorithm that constructs a $(2k-1)(1+\epsilon)$-spanner of $G$ with lightness $O(n^{1/k})$ in $O(m\alpha(m,n))$ time, where $\alpha(\cdot,\cdot)$ is the inverse-Ackermann function. 
\end{theorem}

We remark that $\alpha(m,n) = O(1)$ when $m = \Omega(n\log^{*}n)$; in fact, $\alpha(m,n) = O(1)$ even when $m = \Omega(n\log^{*(c)}n)$ for any constant $c$, where $\log^{*(\ell)}(.)$ denotes the iterated log-star function with $\ell$ stars.  Thus the running time in Theorem~\ref{thm:general-fast} is linear in $m$ in almost the entire regime of graph densities, i.e., except for very sparse graphs.
The previous state-of-the-art running time for the same lightness bound is super-quadratic in $n$, namely $O(n^{2 + 1/k + \epsilon'})$, for any constant $\epsilon' < 1$~\cite{ADFSW19}.

Surprisingly, the result of~\Cref{thm:general-fast} outperforms the analog result for sparse spanners in weighted graphs: for stretch $2k-1$, the only spanner construction with sparsity $O(n^{1/k})$ is the greedy spanner, whose running time is $O(mn^{1+\frac{1}{k}})$. Other results~\cite{ADFSW19,EN18} with stretch $(2k-1)(1+\eps)$ have (nearly) linear running time, but the sparsity is $O(n^{1/k}\log(k))$, which is worse than our lightness bound by a factor of $\log(k)$.  

Informally, the reason we can achieve light spanners that outperform the state-of-the-art sparse spanners stems from the fact that our framework 
essentially reduces the problem of constructing light spanners in weighted graphs to that of constructing sparse spanners in unweighted graphs.
(And in unweighted graphs, one can construct a $(2k-1)$-spanner with $O(n^{1/k})$ sparsity in $O(m)$ time~\cite{HZ96}.)

\paragraph{Unit disk graphs}
We prove the following theorem, which resolves~\Cref{ques:udg}.
\begin{theorem}\label{thm:UDG-fast} For any set $P$ of $n$  points in $\mathbb{R}^d$,  any $d = O(1)$ and any  fixed $\eps \in (0,1)$, one can construct a $(1+\epsilon)$-spanner of the UDG for $P$ with constant sparsity and lightness. 
For $d = 2$, the construction running time  is  $O(n \log n)$; for $d = 3$, the running time is $\tilde{O}(n^{4/3})$; and for $d\geq 4$, the running time is $O(n^{2-\frac{2}{(\lceil d/2 \rceil+1)} + \delta})$ for any constant $\delta > 0$. 
\end{theorem}

\paragraph{Minor-free graphs} 
We prove the following theorem, which resolves~\Cref{ques:minor}. 

\begin{theorem}\label{thm:minor-free-fast}
	For any $K_r$-minor-free graph $G$  and any fixed $\eps \in (0,1)$, one can construct a $(1+\epsilon)$-spanner of $G$ with lightness  $O(r\sqrt{\log r})$ in $O(nr\sqrt{\log r})$ time.  
\end{theorem}

\subsection{Subsequent Work}

In a subsequent and consequent follow-up to this work,  the same authors~\cite{LS22} used our framework here to present a fast construction of spanners with {near-optimal} {\em sparsity} and \emph{lightness} for general graphs \cite{LS22}. We also adapted and simplified our construction here to construct a \emph{sparse} spanner (with unbounded lightness) in $O(m\alpha(m,n) + \srt(m))$ time in the pointer-machine model, where $\srt(m)$ is the time to sort $m$ integers. Even in a stronger \rdrm model, the best-known algorithm for sorting $m$ integers takes $O(m\sqrt{\log \log m})$~\cite{HT02} expected time. Thus, the running time of the sparse spanner algorithm is still inferior to our running time in \Cref{thm:general-fast}. In the  \rdrm model, a linear time algorithm for constructing a sparse spanner was presented; we do not consider this model in our work here.

\subsection{A Unified Framework: Technical and Conceptual Highlights}\label{subsec:framework-intro}

In this section, we give a high-level overview of our framework for constructing light spanners with stretch $t(1+\eps)$, for some parameter $t$ that depends on the examined graph class; e.g., for Euclidean spaces $t  = 1+\eps$, while for general graphs $t = 2k-1$. We have ignored thus far the dependencies of $\eps$ by assuming it is fixed, but in what follows, we shall explicate the exact dependencies of $\eps$ on the running time and lightness bounds. Although the $\eps$-dependencies are not a central part of this paper, they are central to our companion paper on light spanners that achieve fine-grained optimality, and thus they are central to our framework at large.  We shall construct spanners with stretch $t(1+O(\eps))$
and   assume w.l.o.g.\ that $\eps$ is sufficiently smaller than $1$; a stretch of $t(1+\eps)$, for any $\eps \in (0,1)$,   can be achieved by scaling.

Let $L$ be a positive parameter, and let $H_{< L}$ be a $t(1+\gamma\eps)$-spanner for all edges in $G = (V,E,w)$ 
of weight $< L$, for some constant $\gamma \geq 1$.  That is, $V(H_{< L}) = V$ and for any edge $(u,v)\in E$ with $w(u,v)<  L$:
\begin{equation}\label{eq:Stretch-HL}
	d_{H_{< L}}(u,v) \leq t(1+\gamma \eps)w(u,v) .
\end{equation}

Note that by the triangle inequality, $H_{< L}$ is also a $t(1+\gamma\eps)$-spanner for every pair of vertices of distance $< L$. Our framework relies on the notion of a {\em cluster graph}, defined as follows.

\begin{definition}[$(L,\eps,\beta)$-Cluster Graph]\label{def:ClusterGraph-Param} An edge-weighted graph $\mg= (\mv,\me,\omega)$ is called an \emph{$(L,\eps,\beta)$-cluster graph} with respect to spanner $H_{< L}$, for positive parameters $L, \eps, \beta$, if it satisfies the following conditions:
	\begin{enumerate}
		\item Each node $\varphi_C \in \mv$ corresponds to a subset of vertices $C \in V$, called a \emph{cluster},
		 in the original graph $G$. For any pair $\varphi_{C_1}, \varphi_{C_2}$ of distinct nodes 
		in $\mv$, we have $C_1\cap C_2 = \emptyset$.
		\item Each edge $(\varphi_{C_1},\varphi_{C_2})\in \me$ corresponds to an edge $(u,v)\in E$, such that $u \in C_1$ and $v\in C_2$. Furthermore, $\omega(\varphi_{C_1},\varphi_{C_2}) = w(u,v)$.
		\item $L \leq \omega(\varphi_{C_1},\varphi_{C_2}) < (1+\eps)L$, for every edge $(\varphi_{C_1},\varphi_{C_2})\in \me$.
		\item $\dm(H_{< L}[C]) \leq \beta \eps L$, for any cluster $C$ corresponding to a node $\varphi_C \in \mv$.  
	\end{enumerate} 
Here $\dm(X)$ denotes the {\em diameter} of a graph $X$, i.e., the maximum pairwise distance in $X$. 
\end{definition} 

Condition (1) asserts that clusters corresponding to nodes of $\mg$ are vertex-disjoint. Furthermore, Condition (4) asserts that they induce subgraphs of low diameter in $H_{< L}$. In particular, if $\beta$ is constant, then the diameter of clusters is roughly $\eps$ times the weight of edges in the cluster graph.  

In our framework, we use the cluster graph to compute a subset of edges in $G$ of weights in $[L,(1+\eps) L)$ to add to the spanner $H_{< L}$, so as to obtain a spanner, denoted by $H_{<(1+\eps)L}$, for all edges in $G$ of weight less than  $(1+\eps)L$. As a result, we extend the set of edges whose endpoints' distances are preserved (to within the required stretch bound) by the spanner.  By repeating the same construction for edges of higher and higher weights, we eventually obtain a spanner that preserves all pairwise distances in $G$.

To facilitate the transformation of edges of $\mg$ to edges of $G$,  we assume access to a function $\source(\cdot)$ that supports the following operations in $O(1)$ time: (a) given a node $\varphi_C$, $\source(\varphi_C)$ returns a vertex $r(C)$ in cluster $C$, called the {\em representative} of $C$, (b) given an edge $(\varphi_{C_1}, \varphi_{C_2})$ in $\me$,  $\source(\varphi_{C_1}, \varphi_{C_2})$ returns the corresponding edge $(u,v)$ of  $(\varphi_{C_1}, \varphi_{C_2})$, which we refer to as the {\em source edge} of $(\varphi_{C_1}, \varphi_{C_2})$, where $u \in C_1$ and $v\in C_2$. We note that $u$ (resp., $v$) need not be $r(C_1)$ (resp., $r(C_2)$) and that for an edge  $(\varphi_{C_1}, \varphi_{C_2})$ in $\me$ there could be multiple edges $(u,v)\in G$ such that $u\in C_1$ and $v\in C_2$; our algorithm will choose one (often the smallest weight) as the source of $(\varphi_{C_1}, \varphi_{C_2})$. 
Constructing the function $\source(\cdot)$ efficiently is straightforward; the details are in \Cref{sec:framework}. 
Our framework assumes the existence of the following algorithm, hereafter the \emph{sparse spanner algorithm ($\ssa$)}, which computes a subset of edges in $\mg$, whose source edges are added to $H_{< L}$.

\begin{tcolorbox}
	\hypertarget{SPHigh}{}
	$\ssa$: Given an $(L,\eps,\beta)$-cluster graph $\mg(\mv,\me,\omega)$ and function $\source(\cdot)$ as defined above, 
		the $\ssa$ outputs a subset of edges  $\me^{\prune}\subseteq \me$ such that: 
	\begin{enumerate}[noitemsep]
	 	\item \textbf{(Sparsity)~} \hypertarget{Sparsity}{}  $|\me^{\prune}| \leq \chi |\mv|$ for some parameter $\chi > 0$  (which we would like to minimize). 
		\item \textbf{(Stretch)~} \hypertarget{Stretch}{} For each edge $(\varphi_{C_u},\varphi_{C_v})\in \me$, $d_{H_{<(1+\eps)L}}(u,v)\leq t(1+s_{\ssa}(\beta)\eps)w(u,v)$  where $(u,v) = 
		\source(\varphi_{C_u}, \varphi_{C_v})$ and $s_{\ssa}(\beta)$ is some constant that depends on $\beta$  only, and $H_{< (1+\eps)L}$  is the graph obtained by adding the source edges of $\me^{\prune}$ to $H_{< L}$. 
	\end{enumerate}
	Let  $\tm_{\ssa} = O((m'+n')\tau(m',n'))$ be the running time of the $\ssa$, where $\tau$ is a monotone non-decreasing function,   $n' = |\mv|$ and $m' = |\me|$.
\end{tcolorbox}

The final lightness of the spanner we construct will depend on parameter $\chi$ in the $\ssa$, and therefore, $\chi$ should be as small as possible. 

Intuitively, the $\ssa$ can be viewed as an algorithm that constructs a \emph{sparse spanner for an unweighted graph}, as edges of $\mg$ have the same weights up to a factor of $(1+\eps)$ and the only requirement from the edge set $\me^{\prune}$ returned by the $\ssa$, besides achieving small stretch, is that it would be of small size.   Importantly, while the interface to the $\ssa$ remains the same across all graphs,  its exact implementation may change from one graph class to another;  
informally, for each graph class, the $\ssa$ is akin to the state-of-the-art unweighted spanner construction for that class, and this part of the framework is pretty simple.  The highly nontrivial part of the framework is given by the following theorem, which provides a {\em black-box transformation} from the $\ssa$ to an efficient {\em meta-algorithm} for constructing light spanners. We note that this transformation remains the same across all graphs.

\begin{restatable}{theorem}{Framework}
	\label{lm:framework} Let $L,\eps, t, \gamma, \beta$ be non-negative parameters where $\gamma, \beta \geq 1$ only take on constant values, and $0 < \eps \ll 1$. 
	Let $\mathcal{F}$ be an arbitrary graph class.
	If, for any graph $G$ in $\mathcal{F}$, the $\ssa$ can take any $(L,\eps,\beta)$-cluster graph $\mg(\mv,\me,\omega)$ corresponding to $G$  as input and return as output a subset of edges $\me^{\prune}\subseteq \me$ satisfying the aforementioned two properties of (\hyperlink{Sparsity}{Sparsity}) and (\hyperlink{Stretch}{Stretch}),
	then for any graph in $\mathcal{F}$ we can construct a spanner with stretch $t(1+(s_{\ssa}(O(1))+O(1))\eps)$, lightness $O((\chi \eps^{-3} + \eps^{-4})\log(1/\eps))$, and in time $O(m\eps^{-1}(\alpha(m,n) + \tau(m,n) + \eps^{-1})\log(1/\eps))$.  
\end{restatable}

We note that $\gamma$ in \Cref{lm:framework} is the stretch parameter in \Cref{eq:Stretch-HL}, which is encoded via the definition of $\ssa$. We remark the following regarding \Cref{lm:framework}.

\begin{remark}\label{remark:ACTIntro} (1) If the \hyperlink{SPHigh}{$\ssa$} can be implemented in the ACT model, then the construction of light spanners provided by \Cref{lm:framework} can also be implemented in the ACT model in the stated running time. (2)  Parameters $\gamma, \beta$ only take on constant values, and $\eps$ is bounded inversely by $\gamma$ and $\beta$. In all constructions in \Cref{sec:applications}, $\eps \leq \min\{1/\gamma,1/(6\beta+6)\}$.
\end{remark}

In the implementations of \hyperlink{SPHigh}{$\ssa$} for Euclidean spaces and UDGs, we need to gurantee that $H_{< L}$ preserves distances smaller than $L$ within a factor of $t(1+\gamma \eps)$. However, we do not need this gurantee for general graphs and minor-free graphs; all we need is Item (4) in \Cref{def:ClusterGraph-Param}.

The transformation provided by \Cref{lm:framework}, from sparsity in almost unweighted graphs (as captured by the $\ssa$)  to lightness, has a constant loss on lightness (for constant $\eps$) and a small running time overhead.  In \Cref{sec:applications}, we provide simple implementations of the $\ssa$ for several classes of graphs in time $O(m+n)$, for a constant $\eps$; \Cref{lm:framework} thus directly yields a running time of $O((m+n)\alpha(m,n))$. For minor-free graphs, with an additional effort, we remove the factor $\alpha(m,n)$ from the running time. For Euclidean spaces and UDGs, we apply the transformation not on the input space but rather on a sparse spanner, with $O(n)$ edges, hence the running time $O((m+n)\alpha(m,n))$ of the transformation is not the bottleneck, as it is dominated by the time $\Theta(n \log n)$ needed for building Euclidean spanners.

Despite the clean conceptual message behind \Cref{lm:framework} --- in providing a transformation from sparse to light spanners --- its proof is technical and highly intricate. This should not be surprising, as our goal is to have a single framework that can be applied to basically any graph class. The applicability of our framework goes far beyond the specific graph classes considered in the current paper, which merely aim at capturing several very different conceptual and technical hurdles, e.g.,  complete vs.\ non-complete graphs, geometric vs.\ non-geometric graphs, stretch $1+\eps$ vs.\ large stretch, etc. The heart of our framework is captured by \Cref{lm:framework}; we give a brief overview of the proof in \Cref{subsec:proof-overview} below. In our companion paper, we build on this framework to achieve fine-grained optimality for light spanners.

We next argue that our approach is inherently different than previous ones. To this end, we highlight one concrete result --- on Euclidean spanners in the ACT model --- which breaks a longstanding barrier in the area of geometric spanners by using an inherently {\em non-geometric} approach. All the previous algorithms for light Euclidean spanners were achieved via the greedy and approximate-greedy spanner constructions. The greedy algorithm is non-geometric but slow, whereas the approximate-greedy algorithm is geometric and can be implemented much more efficiently. The analysis of the lightness in both algorithms is done via the so-called {\em leapfrog property} \cite{DHN93,das1994fast,DNS95,das1996constructing,GLN02,NS07}, which is a geometric property. The fast spanner construction of GLN \cite{GLN02} implements the approximate-greedy algorithm by constructing a hierarchy of clusters with $O(\frac{\log n }{\log \log n})$ levels and, for each level, Dijkstra's algorithm is used for the construction of clusters for the next level. The GLN construction incurs an additional $O(n\log n)$ factor for each level to run  Dijkstra's algorithm in the ACT model, which ultimately leads to a running time of $O( n \frac{\log^2 n}{\log \log n})$. By employing indirect addressing and exploiting geometric properties, GLN designed an implementation of Dijkstra's algorithm with a running time of $O(n)$ per level after a preprocessing time of $O(n\log n)$. The resulting algorithm with {\em indirect addressing} takes time $O(n\log n)$. Our approach is inherently different, and in particular, we do not need to run	Dijkstra's algorithm or any other single-source shortest or (approximately shortest) path algorithm. The key to our efficiency is careful usage of the new notion of augmented diameter and its interplay with the potential function argument and the hierarchical partition that we use. We stress again that our approach is non-geometric, and the only potential usage of geometry is in the sparse spanner construction that we apply. (Indeed, the sparse spanner construction that we chose to apply is geometric, but this is not a must.)

\subsection{Overview of the Proof of \texorpdfstring{\Cref{lm:framework}}{Framework}}\label{subsec:proof-overview}
  Our starting point of the proof of \Cref{lm:framework}  is a basic hierarchical partition, which dates back to the early 90s \cite{ALGP89,CDNS92}, and was used by most if not all of the works on light spanners (see, e.g.,~\cite{ES16,ENS14,CW16,BLW17,BLW19,LS19}).  Each level $i\geq 0$ of the partition is associated with (i) a set of clusters of diameter $\eps L_i$ where $L_i = (1/\eps)^i$ and (ii) a set of edges of $G$ of weight in the range $[L_i/(1+\eps), L_i)$, called \emph{level-$i$ edges}. Observe that the length of level-$i$ edges is $\Omega(1/\eps)$ times longer than the diameter of the clusters at level $i$. The spanner construction is carried out level by level: first constructing clusters for level 0,  ``taking care'' of edges associated with level 0 (by adding edges to the spanner that preserve the distances between the endpoints of these edges), then moving on to level 1, and later to the next level, and so on. 
  
  One subtle issue is that using a single hierarchy of partitions could not cover all the edges of $G$, since in one hierarchy, level-$(i+1)$ edges are $\Omega(1/\eps)$ longer than level-$i$ edges. This issue can be resolved by using $O(\log (1/\eps))$ hierarchies~\cite{BLW17,BLW19,Le20} and running the same algorithm $O(\log 1/\eps)$ times, each time with a different hierarchy; doing so adds only an $O(\log 1/\eps)$ factor overhead to the final running time and lightness. Our construction here has to be slightly more delicate: we run the (same) algorithm on level 0 of {\em all} hierarchies (from lower values of $L_0$ to higher values), and only then on level 1 of {\em all} hierarchies, and so on. This is important because when we consider edges at level $i$ in a given hierarchy, we rely on the assumption that all edges of length less than $L_i/(1+\eps)$ are already preserved in the spanner constructed so far, including edges not associated with any level of the current hierarchy. We note that in other spanner constructions such a coordination between different hierarchies is not needed.      
 
  Let us focus on the edges associated with level $i$ of some hierarchy. To preserve (the distances between the endpoints of) the level-$i$ edges, a simple idea is to construct a $(L_i,\eps, O(1))$-cluster graph  $\mg_i$ as in \Cref{def:ClusterGraph-Param}: the edge set contains all level-$i$ edges and the vertex set corresponds to the level-$i$ clusters containing the endpoints of the edges. Then, one can simply apply $\ssa$ to $\mg_i$ to get a subset of level-$i$ edges to add to the current spanner. The problem with this naive suggestion is that the total weight of the final spanner would be\footnote{To get this lightness bound, one has to apply standard techniques in a nontrivial way.} $O(\chi \eps^{-1} \log(1/\eps)\log n )$ instead of $O((\chi \eps^{-3} + \eps^{-4})\log(1/\eps))$. That is, one has to pay a factor of $\log n$ in the lightness since the total lightness added at every level $i$ could be $\Omega(\chi \eps^{-1})$, and there are $\Omega(\log n)$ levels and $\log(1/\eps))$ different hierarchies.

  To remove the $\log(n)$ factor in the lightness, one has to take into account the {\em dependency between edges added to the spanner at different levels}. In the geometric setting, as mentioned in \Cref{subsec:framework-intro},  the leap-frog property~\cite{DHN93,das1994fast,DNS95,das1996constructing} captures and handles this dependency in a nontrivial way. For general graphs, the seminal work of Chechik and Wullf-Nilsen~\cite{CW16} introduced a different technique for handling the dependency between different levels, which uses a {\em potential function argument}.  Roughly speaking, the potential $\Phi_i$ of level $i$ is the total diameter of all clusters at level $i$;  the potential of the $0$-th level is $\Phi_0 \leq w(\MST)$. Next, they constructed a spanner in such a way that the total weight of edges added at level $i$ is, loosely speaking, about $O(n^{1/k}/\eps^{2+1/k})(\Phi_i - \Phi_{i+1})$, where $k$ is the stretch parameter. Then, by taking the sum over all levels, the total weight is bounded by  $O(n^{1/k}/\eps^{2+1/k}) \sum_{i}(\Phi_i - \Phi_{i+1}) \leq O(n^{1/k}/\eps^{2+1/k})\Phi_0 = O(n^{1/k}/\eps^{2+1/k})w(\MST)$, leading to the lightness bound of $O(n^{1/k}/\eps^{2+1/k})$. Here, we over-simplified the ideas of~\cite{CW16} in three places: (i) it is not always possible to bound the total weight of edges added at level $i$ by $\Phi_i - \Phi_{i+1}$ (there are cases that have to be handled differently); (ii) they only handle edges of weight from $1$ up to $g^k$ for some large constant $g$ and a post-processing step is needed to handle edges of weight larger than $g^k$, leading to another factor of $1/\eps$ in the final lightness bound; and (iii) their hierarchical partition is different from the hierarchical partition that we have described so far. In particular, it is not clear how one could implement the construction of
  ~\cite{CW16}
  in subquadratic time,  as it requires a certain type of dynamic approximate distance oracle. Such an oracle was provided in a recent work~\cite{ADFSW19}, but it is not strong enough to break the quadratic time barrier (for near-optimal lightness). In this work, however, we are aiming at near-linear time.

Borradaile, Le, and Wulff-Nilsen ~\cite{BLW17} introduced a credit argument, which was an adaptation of the potential function argument of  Chechik and Wullf-Nilsen~\cite{CW16}, to show that the lightness of the greedy $(1+\eps)$-spanner of minor-free graphs is $O(1/\eps^3)$, removing the $\log(n)$ factor from the lightness bound in an earlier paper~\cite{GS02}.
While the potential function argument of Chechik and Wullf-Nilsen~\cite{CW16} is suitable for a stretch of at least 3,
the credit argument of \cite{BLW17} is more natural for the regime of stretch $1+\eps$, and was used by followup works~\cite{BLW19,Le20,LS19} to construct light spanners in the same stretch regime. It is unclear how to implement any of these algorithms in subqudratic time. For example, even in the basic setting of point sets in the Euclidean space $\mathbb{R}^d$, the result of~\cite{LS19} shows that the greedy $(1+\eps)$-spanner has lightness $O(\eps^{-d}\log(1/\eps)$, but the fastest implementation of greedy spanners takes $O(n^2\log(n))$ time~\cite{BCFMS10}.

 In this work we adapt the potential function argument of Chechik and Wullf-Nilsen~\cite{CW16} to the hierarchy of partitions that we set up as described above. We introduce the notion of \emph{augmented diameter} of a cluster and define the potential of a level of the hierarchy to be the sum of the augmented diameters of all the clusters at that level. The formal definition of augmented diameter appears in Section~\ref{sec:prelim}, but at a high level, the idea is to consider weights on both nodes and edges in a cluster, where the node weights are determined by the {\em potential values} of clusters computed (via simple recursion) in previous levels of the hierarchy. The main advantage of augmented diameter over the standard notion of diameter is that it can be computed efficiently, while the computation of diameter is much more costly. Informally, the augmented diameter can be computed efficiently since (i) we can upper bound the hop-diameter of clusters, and (ii) the clusters at each level are computed on top of some underlying tree;
roughly speaking, that means that all the distance computations are carried out on top of subtrees of bounded hop-diameter (or depth), hence the source of efficiency. 

One conceptual idea that guides our cluster and spanner construction is the {\em local view of the potential}.  In our context, it means that each cluster at level $i+1$ is constructed from clusters at level $i$ so as to maximize the \emph{local potential change}, which is basically the difference between the total potential of children clusters and the potential of the parent cluster. This local view is implicit in the cluster construction of Borradaile, Le, and Wulff-Nilsen~\cite{BLW17}; here, we made it explicit via the notion of (corrected) local potential change (\Cref{eq:LocalPotential} and \Cref{def:ClusterGraphNew}). We note that in the work~\cite{BLW17},  clusters are not used in the construction of the spanner, and hence efficiency is irrelevant. Instead, they use the cluster hierarchy to \emph{analyze the greedy algorithm}. On the other hand, our main focus here is on achieving a  \emph{(near-)linear time} construction, and we provide an efficient construction of the clustering algorithm of  Borradaile, Le, and Wulff-Nilsen~\cite{BLW17}.   Basically,  using the augmented diameter, we could bound the size of subgraphs arising during the course of our algorithm and compute the augmented diameters of clusters efficiently.  

The clusters for level $i+1$ that we construct can be partitioned into two sets: one set contains clusters that have a large (corrected) local potential change, called \emph{abundant clusters}, and the other set contains clusters that have 0 local potential change, called \emph{depleted clusters}. This induces a partition of level-$i$ clusters into two sets: the \emph{abundant set}, which consists of the level-$i$ clusters that are in the abundant level-$(i+1)$ clusters, and the \emph{depleted set}, which consists of those in depleted level-$(i+1)$ clusters\footnote{We only use the terminology of abundant and depleted clusters in the introduction; the actual construction is more delicate and requires a more nuanced terminology.}. We then apply $\ssa$ on the cluster graph (at level $i$) induced by the abundant set. Since the clusters are abundant, we can bound the set of edges added by the $\ssa$  by their (corrected) local potential change. For level-$i$ edges between level-$i$ clusters in the depleted set, we simply add them to the spanner. In this case, we cannot bound the spanner edges by the local potential change (because it could be 0 for depleted clusters).  The observation is that the total weight of these edges over all levels is small and, therefore, we can take care of this case by a simple tweak (the sequence  $\{a_i\}_{i \in \mathbb{N}^+}$ in \Cref{lm:framework-technical}).

One interesting aspect of our construction is that even if the running time at each level could be $\Omega(m)$ (modulo the running time of the $\ssa$), our overall running time overhead is $O(m\alpha(m,n))$ instead of $O(m \log(n))$, where $O(\log n)$ is the height of the hierarchy. That is, our framework can exploit the dependency between levels to optimize the running time. In our follow-up work~\cite{LS22}, we adapted the framework for lightness in this paper, specifically exploiting the dependency between different levels, to construct a spanner for general graphs with near-optimal \emph{sparsity} in linear time (in the RAM model).

  In summary, we propose a unified framework that reduces the problem of \emph{efficiently} constructing a \emph{light} spanner to the conjunction of two problems: (1) efficiently constructing a hierarchy of clusters with several \emph{carefully chosen properties}, and (2) efficiently constructing a {\em sparse spanner}; these two problems are intimately related in the sense 
that the ``carefully chosen properties'' of the clusters are set so that we are able to apply the sparse spanner construction efficiently.


\begin{table}[!ht]
\caption{Notation introduced in \Cref{sec:intro}.}
	\label{table:glossray}
\renewcommand{\arraystretch}{1.3}
\resizebox{\textwidth}{!}{%
\begin{tabular}{|l|l|}
  	\hline
	\textbf{Notation} & \textbf{Meaning} \\ \hline
	$t,\eps$ & Stretch parameters, $t\geq 1, \eps \ll 1$.\\ \hline 
	$\nor2{p,q}$ & Euclidean distance between two points $p,q\in \mathbb{R}^d$.\\ \hline 
	$\alpha(m,n)$ &  The inverse Ackermann function.\\ \hline 
	$H_{<L}$ & $t(1+\gamma \eps)$-spanner for edges of weights less than $L$.\\ \hline 
	$\gamma$ & Stretch parameter in  $H_{< L}$;  $\gamma\geq 1$.\\ \hline 
	$L,\beta$ & Parameters in $(L,\eps,\beta)$-cluster graph (\Cref{def:ClusterGraph-Param}). \\ \hline 
$\mg = (\mv,\me,\omega)$ &  The $(L,\eps,\beta)$-cluster graph;  $L \leq \omega(\varphi_{C_1},\varphi_{C_2}) < (1+\eps)L$  $\forall (\varphi_{C_1},\varphi_{C_2})\in \me$. \\ \hline 
	$\varphi_C$ &  The node in $\mg$ corresponding to a cluster $C$.\\ \hline 
	$\ssa$ & The sparse spanner \hyperlink{SPHigh}{algorithm}.\\ \hline 
	$\chi$ & The  \hyperlink{Sparsity}{sparsity} parameter of $\ssa$.\\ \hline
	$s_{\ssa}(\cdot)$ & The \hyperlink{Stretch}{stretch} function of $\ssa$.\\ \hline  
	$\tau(\cdot, \cdot)$ & The function in the running time of $\ssa$.\\ \hline 
	$\source(\varphi_{C})$ & This returns the representative $r(C)$ in cluster $C$.\\ \hline 
	$\source(\varphi_{C_1}, \varphi_{C_2})$ & This returns the corresponding edge $(u,v)$ of  $(\varphi_{C_1}, \varphi_{C_2}) \in \me$. \\ \hline 
\end{tabular}
}
\renewcommand{\arraystretch}{1}
\end{table}

\section{Preliminaries}\label{sec:prelim}

Let $G$ be an arbitrary edge-weighted graph. We denote by $V(G)$ and $E(G)$ the vertex set and edge set of $G$, respectively. We denote by $w: E(G)\rightarrow \mathbb{R}^+$ the weight function on the edge set.  Sometimes we write $G = (V,E)$ to clearly explicate the vertex set and edge set of $G$, and  $G =(V,E,w)$ to further indicate the weight function $w$ associated with $G$. 
We use $\mst(G)$ to denote a minimum spanning tree of $G$; when the graph is clear from context, we simply use $\mst$  as a shorthand for $\mst(G)$.

For a subgraph $H$ of $G$, we use $w(H) \defi \sum_{e\in E(H)}w(e)$ to denote the total edge weight of $H$.   The {\em distance} between two vertices $p,q$ in $G$, denoted by $d_G(p,q)$, is the minimum weight of a path between them in $G$. The diameter of $G$,
denoted by $\dm(G)$, is the maximum pairwise distance in $G$. A \emph{diameter path} of $G$ is a shortest (i.e., of minimum weight) path in $G$ realizing the diameter of $G$, that is,  it is a shortest path between some pair $u,v$ of vertices in $G$ such that $\dm(G) = d_G(u,v)$.

Sometimes we shall consider graphs with weights on both \emph{edges and vertices}. We define the \emph{augmented weight} of a path to be the total weight of all edges and vertices along the path. The \emph{augmented distance} between two vertices in $G$ is defined as the minimum augmented weight of a path between them in $G$. 
Likewise, the \emph{augmented diameter} of $G$, denoted by $\adm(G)$,
 is the maximum pairwise augmented distance in $G$; 
since we will focus on non-negative weights, the augmented distance and augmented diameter
are no smaller than the (ordinary notions of) distance and diameter. 
An \emph{augmented diameter path} of $G$ is a path of minimum augmented weight realizing the augmented diameter of $G$.

Given a subset of vertices $X\subseteq V(G)$, we denote by $G[X]$ the subgraph of $G$ \emph{induced by $X$}: $G[X]$ has $V(G[X]) = X$ and $E(G[X]) = \{(u,v)\in E(G) ~\vert~   u,v \in X\}$. Let $F\subseteq E(G)$ be a subset of edges of $G$. We denote by $G[F]$ the subgraph of $G$ with $V(G[F]) = V(G)$ and $E(G[F]) = F$.

Let $S$ be a \emph{spanning} subgraph of $G$; weights of edges in $S$ are inherited from $G$. The \emph{stretch} of $S$  is  given by $\max_{x,y \in V(G)} \frac{d_S(x,y)}{d_G(x,y)}$, and the maximum is attained by some edge $(x,y)$ of $G$. Throughout we will use the following known observation, e.g., Lemma 1 in~\cite{ADDJS93},
which implies that the stretch of $S$ is equal to $\frac{d_S(u,v)}{w(u,v)}$ for some edge $(u,v) \in E(G)$. 
\begin{observation} \label{stretch:ob}
 $\max_{x,y \in V(G)} \frac{d_S(x,y)}{d_G(x,y)} =  \max_{(x,y) \in E(G)} \frac{d_S(x,y)}{d_G(x,y)}$.
\end{observation}
 We say that $S$ is a \emph{$t$-spanner} of $G$ if the stretch of $S$ is at most $t$. There is a simple greedy algorithm,  called $\pathg$ (or $\greedy$ for short), to find a $t$-spanner of a graph $G$: Examine the edges $e = (x,y)$ in  $G$ in nondecreasing order of weights, and add to the spanner edge $(x,y)$ iff the distance between $x$ and $y$ in the {\em current} spanner is larger than $t\cdot w(x,y)$.
 
  We say that a subgraph $H$ of $G$ is a $t$-spanner for a \emph{subset of edges $X\subseteq E$} if $\max_{(u,v) \in X} \frac{d_H(u,v)}{d_G(u,v)} \leq t$.

In the context of minor-free graphs, we denote by $G/e$ the graph obtained from $G$ by contracting $e$, where $e$ is an edge in $G$. If $G$ has weights on edges, then every edge in $G/e$ inherits its weight from $G$.

In addition to general and minor-free graphs, this paper studies {\em geometric graphs}.  Let $P$ be a set of $n$ points in $\mathbb{R}^d$. We denote by $\nor2{p,q}$ the Euclidean distance between two points $p,q\in \mathbb{R}^d$.   A  \emph{geometric graph} $G$ for $P$ is a graph where the vertex set corresponds to the point set, i.e., $V(G) = P$, and the edge weights are the Euclidean distances, i.e.,  $w(u,v) = \nor2{u,v}$ for every edge $(u,v)$ in $G$. Note that $G$ need not be a complete graph. If $G$ is a complete graph, i.e., $G = (P, \binom{P}{2},\nor2{\cdot})$, then $G$ is equivalent to the {\em Euclidean space} induced by the point set $P$.
For geometric graphs, we use the term \emph{vertex} and \emph{point} interchangeably.

We use $[n]$ and $[0,n]$ to denote the sets $\{1,2,\ldots,n\}$ and $\{0,1,\ldots,n\}$, respectively.

\section{Applications of the Unified Framework}\label{sec:applications}
In this section, we implement the \hyperlink{SPHigh}{$\ssa$} for each of the graph classes.   By plugging the $\ssa$ on top of the general transformation, as provided by \Cref{lm:framework}, we shall prove all theorems stated in \Cref{sec:intro}. We assume that $\eps \ll 1$, and this is without loss of generality since we can remove this assumption by scaling $\eps \leftarrow \eps'/c $ for \emph{any} $\eps' \in (0,1)$ and sufficiently large constant $c$. The scaling will incur a constant loss on lightness and running time, as the dependency on $1/\eps$ is polynomial in all constructions below. We refer readers to \Cref{table:glossray} for a summary of the notation introduced in \Cref{sec:intro}.

\subsection{Euclidean Spanners and UDG Spanners}\label{subsec:EuclideanUDG}

In this section, we prove the following theorem.

\begin{restatable}{theorem}{ACTUDGClustering}
	\label{thm:ACTUDGSpanner}
	Let $G = (P,E,w)$ be a $(1+\eps)$-spanner  either for a set of $n$ points $P$  or for  the unit ball graph $U$ of $P$ in $\mathbb{R}^d$ with $m$ edges. There is an algorithm that can  compute a $(1 + O(\eps))$-spanner $H$ of $G$ in the ACT model with lightness  $O((\eps^{-(d+2)} + \eps^{-4})\log(1/\eps))$  in time  $O(m\eps^{-1}(\alpha(m,n) + \eps^{1-d})\log(1/\eps))$.
\end{restatable}

We now show that \Cref{thm:ACTUDGSpanner} implies \Cref{thm:ACT} and \Cref{thm:UDG-fast}.

\begin{proof}[Proofs of \Cref{thm:ACT} and \Cref{thm:UDG-fast}]
	
It is known that a Euclidean $(1+\eps)$-spanner for a set of $n$ points $P$ in $\mathbb{R}^d$ with degree $O(\epsilon^{1-d})$ can be constructed in $O(n\log n)$ time in the ACT model (cf.\ Theorems 10.1.3 and 10.1.10 in \cite{NS07}). 
Furthermore, when $m = O(n\eps^{1-d})$, we have that: $$\alpha(m,n) ~=~ \alpha(nO(\eps^{-d}), n) ~=~ O(\alpha(n) + \log(\eps^{-d})) ~=~ O(\alpha(n) + d\log(1/\eps)).$$
Thus,  Theorem~\ref{thm:ACT} follows from Theorem~\ref{thm:ACTUDGSpanner}.

By \Cref{lm:UDG-sparse-sp}, we can construct sparse $(1+\epsilon)$-spanners for unit ball graphs with $m = O(n\eps^{1-d})$ edges in $O(n(\epsilon^{-2}\log n)$ time when $d = 2$, $\tilde{O}(n^{4/3}\eps^{-3})$ time when $d = 3$, and  $O(n^{2-\frac{2}{(\lceil d/2 \rceil+1)} + \delta}\epsilon^{-d+1} + n\epsilon^{-d})$ time for any constant $\delta > 0$ when $d\geq 4$.  Thus, Theorem~\ref{thm:UDG-fast} follows from Theorem~\ref{thm:ACTUDGSpanner}.
\end{proof}

\begin{figure}[!ht]
	\begin{center}
		\includegraphics[width=0.6\textwidth]{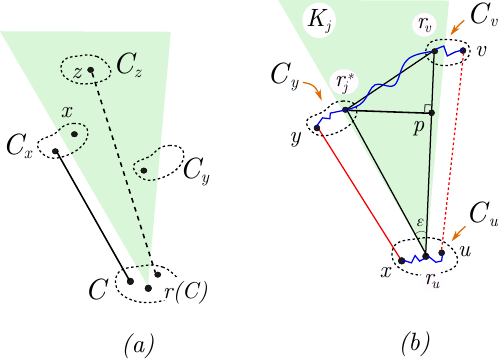}
	\end{center}
	\caption{ (a) Three clusters having representatives in the same cone with apex $r(C)$. Two clusters $C_x$ and $C_z$ are neighbors of $C$ in $\mathcal{G}$. Our algorithm  will add the edge between $C$ and $C_x$ to $\me^{\prune}$ since $x$ is closer to $r(C)$ than $z$. Cluster $C_y$ is a non-neighbor whose representative is closest to $r(C)$, but we do not add any edge between $C$ and $C_y$.  (b) Illustration for the stretch bound proof of \Cref{lm:ACT-UDG}. Black dashed curves represent three clusters $C_u,C_v, C_y$. 
	The solid red edge $(x,y)$ corresponds to an edge added to $\me^{\prune}$, while the dashed red edge  $(u,v)$ is not added. The green shaded region represents cone $Q_j$ of angle $\epsilon$ with the apex at $r_u$.}
	\label{fig:ACT-stretch}
\end{figure}

By \Cref{lm:framework}, in order to prove \Cref{thm:ACTUDGSpanner},  it suffices to implement the \hyperlink{SPHigh}{$\ssa$} for Euclidean and UDG spanners.
Next, we give a detailed geometric implementation of the \hyperlink{SPHigh}{$\ssa$}, hereafter $\ssa_{\geom}$; note that the stretch parameter $t$ in the geometric setting is $1+\eps$.  The idea is to use a Yao-graph like construction: For each node $\varphi_C \in \mv$, we construct a collection of cones of angle $\eps$ around the representative $r(C) = \source(\varphi_{C})$ of the cluster $C$ corresponding to $\varphi_C$. Recall that we have access to a $\source$ function that returns the representative of each cluster in $O(1)$ time.  Then for each cone, we look at all the representatives of the neighbors (in $\mg$) of $C$ that fall into that cone and pick to $\me^{\prune}$ the edge that connects $r(C)$ to the representative that is closest to it.  It could be that a \emph{non-neighbor} cluster of $C$ has a representative closer to $C$, but we do not add any edge between the two clusters. This is a difference between our algorithm and the Yao-graph algorithm. See \Cref{fig:ACT-stretch}(a).

\begin{tcolorbox}
	\hypertarget{SPHEuclidean}{}
	\textbf{$\ssa_{\geom}$ (Euclidean and UDG):} The input is a $(L,\eps,\beta)$-cluster graph $\mg(\mv,\me,\omega)$ that corresponds to a Euclidean or UDG spanner. The output is $\me^{\prune}$; initially, $\me^{\prune} = \emptyset$.
	\begin{quote}
		 For each node $\varphi_{C_u} \in \mv$, do the following:
		\begin{itemize}
		\item Let $\mathcal{N}(\varphi_{C_u})$ be the set of neighbors of $\varphi_{C_u}$ in $\mg$.  We construct a collection of $\tau = O(\epsilon^{1-d})$  cones $\cone(C_u) = \{Q_1, Q_2,\ldots, Q_{\tau}\}$ that partition $\mathbb{R}^d$, each of angle $\eps$ and with apex at $r(C_u)$, the representative of  $C_u$. It is   known (see, e.g. Lemma 5.2.8 in~\cite{NS07}) that we can construct $\cone(C_u)$ in time $O(\epsilon^{1-d})$ in the ACT model. 
		
		\item For each $j\in [\tau]$:
		\begin{itemize}
		\item Let $R_j = \{r(C'): \varphi_{C'} \in \mathcal{N}(\varphi_{C_u}) \wedge (r(C') \in Q_j)\}$ 
		be the set of representatives that belong to the cone $Q_j \in \cone(C_u)$.  
		 Let $r_j^{*} = \arg\min_{r\in R_j} \nor2{r(C_u),r}$ be the representative in $R_j$ that is closest to $r(C_u)$. 
		\item Let $\varphi_{C_v}$ be the node of $\mg$ whose cluster $C_v$  has $r^*_j$ as the representative. By the definition of $R_j$, $(\varphi_{C_u},\varphi_{C_v})$ is an edge in $\me$. Add $(\varphi_{C_u},\varphi_{C_v})$ to $\me^{\prune}$.  
					\end{itemize}
			/* We add at most one edge to $\me^{\prune}$ incident on $\varphi_{C_u}$ for each of the $\tau$ cones.	*/	
		\end{itemize}
	\end{quote}
\end{tcolorbox}

We next analyze the running time of $\ssa_{\geom}$, and also show that it satisfies the two properties of (\hyperlink{Sparsity}{Sparsity}) and (\hyperlink{Stretch}{Stretch}) 
required by the abstract \hyperlink{SPHigh}{$\ssa$}; these properties are described in \Cref{subsec:framework-intro}. Recall that $H_{<(1+\eps)L}$ is the graph obtained by adding the source edges of $\me^{\prune}$ to $H_{< L}$, which is the spanner for all edges in $G$ of weight $< L$. Note that the stretch of $H_{< L}$ is $t(1+\gamma \eps)$ for $t = 1+\eps$, where $\gamma$ is a constant. Furthermore, as mentioned, we assume w.l.o.g.\ that $\eps$ is sufficiently smaller than $1$.

\begin{lemma}\label{lm:ACT-UDG} \hyperlink{SPHEuclidean}{$\ssa_{\geom}$} can be implemented in $O((|\mv| + |\me|)\eps^{1-d})$ time in the ACT model. Furthermore, 1. (Sparsity) $|\me^{\prune}| = O(\eps^{1-d})|\mv|$, and 2. (Stretch) 
For each edge $(\varphi_{C_u},\varphi_{C_v})\in \me$, $d_{H_{<(1+\eps)L}}(u,v)\leq t(1+s_{\ssa_{\geom}}(\beta)\eps)w(u,v)$, where $(u,v) = 
		 \source(\varphi_{C_u}, \varphi_{C_v})$, $s_{\ssa_{\geom}}(\beta) = 19\beta + 14$ and $\eps \leq \min\{\frac{1}{\gamma},\frac{1}{8\beta + 6}\}$.
\end{lemma}
\begin{proof} We first analyze the running time. We observe that since we can construct $\cone(C_u)$ for a single node $\varphi_{C_u}$ in $O(\epsilon^{1-d})$ time in the ACT model,  the running time to construct all sets of cones $\{\cone(C_u)\}_{\varphi_{C_u} \in \mv}$ is $O(|\mv|\epsilon^{1-d})$. Now consider a specific node $\varphi_{C_u}$. For each neighbor $\varphi_{C'} \in \mathcal{N}(\varphi_{C_u})$ of $\varphi_{C_u}$, finding the cone $Q_j \in \cone(C_u)$ such that $r(C') \in Q_j$ takes $O(\tau) = O(\epsilon^{1-d})$ time. 
	Thus, $\{R_j\}_{j=1}^{\tau}$ can be constructed in $O(|\mathcal{N}(\varphi_{C_u})|\eps^{1-d})$ time. Finding the set of representatives $\{r^*_j\}_{j=1}^{\tau}$ takes $O(|\mathcal{N}(\varphi_{C_u})|)$ time by calling function $\source(\cdot)$. Thus, the total running time to implement \hyperlink{SPHEuclidean}{Algorithm $\ssa_{\geom}$} is:
	\begin{equation*}
		O(|\mv|\epsilon^{1-d}) +  \sum_{\varphi_{C_u} \in \mv} O(|\mathcal{N}(\varphi_{C_u})|\eps^{1-d}) = O((|\mv| + |\me|)\eps^{1-d})~,
	\end{equation*}
	as claimed.	 
	
	 By the construction of the algorithm, for each node $\varphi_{C}\in \mv$, we add at most $\tau = O(\eps^{1-d})$ incident edges in $\me$ to $\me^{\prune}$; this implies Item 1. 
	
It remains to prove Item 2: For each edge $(\varphi_{C_u},\varphi_{C_v})\in \me$, the stretch  in $H_{<(1+\eps)L}$ of the corresponding edge $(u,v) = \source(\varphi_{C_u},\varphi_{C_v})$  is at most $(1+s_{\ssa_{\geom}}(\beta)\eps)$ with $s_{\ssa_{\geom}}(\beta)= 2(19\beta + 14)$. Let $r_u \defi r(C_u)$ and $r_v \defi r(C_v)$ be the representatives of $C_u$ and $C_v$, respectively. Let $Q_j$ be the cone in $\cone(C_u)$ such that $r_v \in Q_j$ for some $j \in [\tau]$ (we are using the notation in \hyperlink{SPHEuclidean}{$\ssa_{\geom}$}).  If $r_v = r^*_j$, then $(u,v) \in H_{<(1+\eps)L}$ by the construction in \hyperlink{SPHEuclidean}{$\ssa_{\geom}$}, and so the stretch is $1$. Otherwise, let $C_y$ be the cluster that contains the representative $r^*_j$. By the construction in \hyperlink{SPHEuclidean}{$\ssa_{\geom}$}, there is an edge $(x,y)\in H_{<(1+\eps)L}$ where $x\in C_u$ and $y\in C_y$.  (See Figure~\ref{fig:ACT-stretch}.) By property 4 of $\mg$ in \Cref{def:ClusterGraph-Param},  $\max\{\dm(H_{<(1+\eps)L}[C_u]), \dm(H_{<(1+\eps)L}[C_v]), \dm(H_{<(1+\eps)L}[C_y])\} \leq \beta \eps L$. Note that edges in $\me$ have weights in $[L,(1+\eps)L)$ by property 3 in \Cref{def:ClusterGraph-Param}. By the triangle inequality:
	\begin{equation}\label{eq:ACT-triangle}
		\begin{split}
			\nor2{r_u,r_v} &\leq \nor2{u,v} + 2\beta\eps L \leq (1+(1+2\beta)\epsilon)L\\
			\nor2{r_u,r^*_j}  &\leq \nor2{x,y} + 2\beta \eps L \leq (1+(1+2\beta)\epsilon)L\\
			\nor2{u,v} &\leq \nor2{r_u,r_v} + 2\beta\epsilon L\\
			\nor2{x,y}  &\leq \nor2{r_u,r^*_j} + 2\beta\epsilon L
		\end{split}
	\end{equation}
	
	Furthermore, since $L \leq  \nor2{u,v},  \nor2{x,y}\leq  (1+\eps)L$, it follows that:
	\begin{equation}\label{eq:ACT-uvxy}
		\begin{split}
			\nor2{u,v} &\leq (1+\eps) \nor2{x,y} \\
			\nor2{x,y} &\leq (1+\eps) \nor2{u,v} 
		\end{split}
	\end{equation}

	\begin{claim}\label{clm:ACT-rvrj} $\nor2{r_v,r^*_j} \leq (8\beta+6)\eps L $. 
	\end{claim}
	\begin{proof} Recall that $\nor2{r_u,r^*_j} \leq \nor2{r_u,r_v}$. Let $p$ be the projection of $r_j^*$ onto the segment $r_ur_v$ (see \Cref{fig:ACT-stretch}). Since $\angle r_vr_ur^*_j \leq \epsilon$, $\nor2{r_j^*,p} \leq \sin(\eps) \nor2{r_u,r^*_j} \leq \sin(\eps) \nor2{r_u,r_v} ~\leq~ \eps(1+(1+2\beta)\eps)L$. We have:
		\begin{equation}\label{eq:ACT-rvrj}
			\begin{split}
				\nor2{r_v,r^*_j} &\leq  \nor2{p,r^*_j} + \nor2{r_v,p}    = \nor2{p,r^*_j} +  \nor2{r_u,r_v} - \nor2{p,r_u} \\
				&\leq \nor2{p,r^*_j} +   \nor2{r_u,r_v} - (\nor2{r_u,r_j^*} - \nor2{r_j^*,p})\\
				&\leq (\nor2{r_u,r_v} - \nor2{r_u,r_j^*}) + 2\eps(1+(1+2\beta)\eps)L
			\end{split}
		\end{equation}

		We now bound $(\nor2{r_u,r_v} - \nor2{r_u,r_j^*})$. By \Cref{eq:ACT-triangle} and \Cref{eq:ACT-uvxy}, it holds that:
		\begin{equation} \label{addedeq}
			\begin{split}
				\nor2{r_u,r_v} - \nor2{r_u,r_j^*}  &\leq \nor2{u,v} + 2\beta\epsilon L  - (\nor2{x,y}- 2\beta\epsilon L)\\ 
				&= \nor2{u,v}- \nor2{x,y}  + 4\beta\epsilon L \leq \eps \nor2{x,y} + 4\beta\epsilon L  \leq (4\beta+1+\eps)\eps L
			\end{split}
		\end{equation}
		
		Plugging \Cref{addedeq} into \Cref{eq:ACT-rvrj}, we get:
		\begin{equation*}
			\begin{split}
			\nor2{r_v,r_j^*} &\leq (4\beta+1+\eps)\eps L +  2\eps(1+(1+2\beta)\eps)L \\
			&\leq  (4\beta+2)\eps L + 2\eps L + 2(1+2\beta)\eps L \qquad \mbox{(since }\eps \leq 1)\\
			&\leq (8\beta +6)\eps L~,
			\end{split}
		\end{equation*}
	as claimed.	This completes the proof of \Cref{clm:ACT-rvrj}. \qed
	\end{proof}
	
Next, we continue with the proof of \Cref{lm:ACT-UDG}.	By \Cref{clm:ACT-rvrj}, $\nor2{r_v,r^*_j}  < L$ when $\eps < 1/(8\beta+ 6)$. If the input graph is a UDG, then $\me  \not=\emptyset$ only if $L \leq 1$. Thus,  $\nor2{r_v,r^*_j}\leq 1$ and hence, there is an edge $(r_v,r^*_j)$ of length $\nor2{r_v,r^*_j}$ in the input UDG. (This is the only place, other than starting our construction with a $(1+\eps)$-spanner for the input UDG, where we exploit the fact that the input graph is a UDG.)
	
	Since $\nor2{r_v,r^*_j} < L$, the distance between $r_v$ and $r^*_j$ is preserved up to a factor of $(1+\gamma \eps)$ in $H_{< L}$.	That is, $d_{H_{<(1+\eps)L}}(r_v,r^*_j) \leq (1+\gamma \eps)\nor2{r_v,r^*_j}$.

	Note that $r_u,r_v,r_j^*$ are in the input point set $P$ by the definition of representatives.  By the triangle inequality, it follows that:
	\begin{equation}\label{eq:ACT-e1}
		\begin{split}
			d_{H_{<(1+\eps)L}}(u,v) &\leq d_{H_{<(1+\eps)L}}(u, x) +  \nor2{x,y} +  d_{H_{<(1+\eps)L}}(y,r_j^*)+ d_{H_{<(1+\eps)L}}(r_j^*, r_v)\\
            & \hspace{2cm} +  d_{H_{<(1+\eps)L}}(r_v, v)\\
			&\leq \beta\epsilon L +  \nor2{x,y} + \beta\eps L + (1+\gamma \eps)\nor2{r_v,r^*_j}  + \beta\eps L  \\
			&\leq  \nor2{x,y} +3\beta\eps L + \underbrace{(1+\gamma \eps)}_{\leq~ 2 \text{ since } \eps ~\leq~ 1/\gamma}(8\beta+6)\eps L \qquad \mbox{(by \Cref{clm:ACT-rvrj})}\\
			&\leq   \nor2{x,y} + (19\beta + 12)\eps L 
		\end{split}
	\end{equation}
	By \Cref{eq:ACT-uvxy},   $\nor2{x,y}~\leq~ (1+\eps)\nor2{u,v}\leq \nor2{u,v} + (1+\eps)\eps L \leq \nor2{u,v} + 2\eps L $. Thus, by Equation~\eqref{eq:ACT-e1}:
	$$	d_{H_{<(1+\eps)L}}(u,v) \leq \nor2{u,v} + (19\beta + 14)\eps L \stackrel{\nor2{u,v}\geq L}{\leq} (1+(19\beta + 14)\eps)\nor2{u,v}.$$
	That is, the stretch of  $(u,v)$  in $H_{<(1+\eps)L}$ is at most $1+s_{\ssa_{\geom}}(\beta)\eps$ with $s_{\ssa_{\geom}}(\beta)= 19\beta + 14$, as required.
\end{proof}

\begin{remark}\label{remark:ACT-Faster} \hyperlink{SPHEuclidean}{$\ssa_{\geom}$} can be implemented slightly faster, within time $O(|\mv|\eps^{1-d} + |\me|\log(1/\eps))$, by using a data structure that allows us to search for the cone that a representative belongs to in $O(\log(1/\eps))$ time. Such a data structure is described in Theorem 5.3.2 in the book by Narasimhan and Smid~\cite{NS07}.
\end{remark}

We are now ready to prove \Cref{thm:ACTUDGSpanner}.

\begin{proof}[Proof of \Cref{thm:ACTUDGSpanner}] We use $\ssa_{\geom}$ in place of \hyperlink{SPHigh}{$\ssa$} in \Cref{lm:framework} to construct the light spanner. By \Cref{lm:ACT-UDG}, we have $s_{\ssa}(\beta) = 2(19\beta + 14)$, $\chi = O(\eps^{1-d})$ and $\tau(m',n') = O(\eps^{1-d})$. Thus, by plugging in the values of $\chi$ and $\tau$, we obtain the lightness and the running time as required by \Cref{thm:ACTUDGSpanner}. The stretch of the spanner is:
	\begin{equation*}
		(1+\eps)(1 + (s_{\ssa}(O(1)) + O(1))\eps)  = (1 + O(\eps))~,
	\end{equation*}
when $\eps \leq 1$.
\end{proof}

\subsection{General Graphs}\label{subsec:general}

In this section, we prove \Cref{thm:general-fast} by giving a detailed implementation of \hyperlink{SPHigh}{$\ssa$} for general graphs, hereafter $\ssa_{\gen}$. Here we have $t = 2k-1$ for an integer parameter $k \ge 2$. We will use as a black box the linear-time construction of sparse spanners in general unweighted graphs by Halperin and Zwick~\cite{HZ96}. 

\begin{theorem}[Halperin-Zwick~\cite{HZ96}]\label{thm:unweighted-2k-1} Given an unweighted $n$-vertex graph $G$ with $m$ edges, a $(2k-1)$-spanner of $G$ with $O(n^{1+\frac{1}{k}})$ edges can be constructed deterministically in $O(m + n)$ time, for any $k \ge 2$.
\end{theorem}

\begin{tcolorbox}
	\hypertarget{SPHGeneral}{}
	\textbf{$\ssa_{\gen}$ (General Graphs):} The input is a $(L,\eps,\beta)$-cluster graph $\mg(\mv,\me,\omega)$. The output is $\me^{\prune}$; initially, $\me^{\prune} = \emptyset$.
	\begin{quote}
		We construct a new \emph{unweighted} graph $J = (V_J,E_J)$ as follows. For each node $\varphi \in \mv$, we add a vertex $v_{\varphi}$ to $V_J$. For each edge $(\varphi_1,\varphi_2) \in \me$, we add an edge $(v_{\varphi_1}, v_{\varphi_2})$ to $E_J$. 
		
		Next, we run  Halperin-Zwick's algorithm (\Cref{thm:unweighted-2k-1}) on $J$ to construct a $(2k-1)$-spanner $S_J$ for $J$. Then for each edge $(v_{\varphi_1},v_{\varphi_2})$ in $E(S_J)$, we add the corresponding edge $(\varphi_1,\varphi_2)$ to $\me^{\prune}$. 
	\end{quote}
\end{tcolorbox}

We next analyze the running time of $\ssa_{\gen}$, and also show that it satisfies the two properties of (\hyperlink{Sparsity}{Sparsity}) and (\hyperlink{Stretch}{Stretch}) 
required by the abstract \hyperlink{SPHigh}{$\ssa$}; these properties are described in \Cref{subsec:framework-intro}.

\begin{lemma}\label{lm:App-Gen} \hyperlink{SPHGeneral}{$\ssa_{\gen}$} can be implemented in $O(|\mv| + |\me|)$ time. Furthermore, 1. (Sparsity) $\me^{\prune} = O(n^{1/k})|\mv|$, and 2. (Stretch) For each edge $(\varphi_{C_u},\varphi_{C_v})\in \me$, $d_{H_{<(1+\eps)L}}(u,v)\leq (2k-1)(1+s_{\ssa_{\gen}}(\beta)\eps)w(u,v)$, where $(u,v) = 
		 \source(\varphi_{C_u}, \varphi_{C_v})$, $s_{\ssa_{\gen}}(\beta) = (2\beta+1)$ and $\eps \leq 1$. 
\end{lemma}
\begin{proof} The running time of $\ssa_{\gen}$ follows directly from \Cref{thm:unweighted-2k-1}. Also, by \Cref{thm:unweighted-2k-1}, $|\me^{\prune}|= O(|\mv|^{1+1/k}) = O(n^{1/k}|\mv|)$; this implies Item 1. 
		
	It remains to prove Item 2: 
	For each edge $(\varphi_{C_u},\varphi_{C_v})\in \me$, the stretch in $H_{<(1+\eps)L}$ (constructed as described in \hyperlink{SPHigh}{$\ssa$})
	of the corresponding edge $(u,v) = 		 \source(\varphi_{C_u}, \varphi_{C_v})$ is
at most  $(2k-1)(1+ (2\beta+1)\eps) w(u,v)$.	Recall that $H_{<(1+\eps)L}$ is the graph obtained by adding the source edges of $\me^{\prune}$ to $H_{< L}$.
	
	Let $(u_{1},v_1)$ be the edge in $E_J$ that corresponds to the edge $(\varphi_{C_u},\varphi_{C_v})$. By \Cref{thm:unweighted-2k-1}, there is a path $P$ between $u_1$ and $v_1$ in $S_J$ such that $P$ contains at most $2k-1$ edges. We write $P = (u_1=x_0, (x_0,x_1), x_1, (x_1,x_2), \ldots, x_{p} = v_1)$ as an alternating sequence of vertices and edges. Let $\mp = ( \varphi_0, (\varphi_0,\varphi_1), \varphi_1, (\varphi_1,\varphi_2), \ldots, \varphi_{p})$ be a path of $\mg$, written as an alternating sequence of vertices and edges, that is obtained from $P$ where $\varphi_j$ corresponds to $x_j$, $1\leq j\leq p$. Note that $\varphi_1 = \varphi_{C_u}$ and $\varphi_{p} = \varphi_{C_v}$.
	
	\begin{figure}[!ht]
		\begin{center}
			\includegraphics[width=0.9\textwidth]{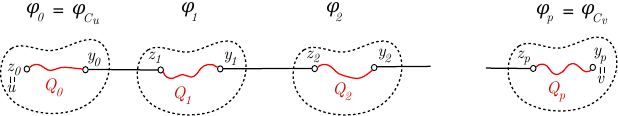}
		\end{center}
		\caption{A path from $u$ to $v$.}
		\label{fig:path}
	\end{figure}
	Let $\{y_i\}_{i=0}^p$ and $\{z_i\}_{i=0}^p$ be two sequences of vertices of $G$ such that (a) $z_0 = u$ and $y_p = v$, and (b) $(y_{i-1}, z_i)$ is the edge in $G$ corresponding to edge $(\varphi_{i-1},\varphi_i)$ in $\mathcal{P}$, for $1\leq i\leq p$. Let $Q_i$, $0\leq i \leq p$, be a shortest path in $H_{< L}[C_i]$ between $z_i$ and $y_i$, where $C_i$ is the cluster corresponding to $\varphi_i$. See \Cref{fig:path} for an illustration. Observe that $w(Q_i) \leq \beta \eps L$ by property 4 in \Cref{def:ClusterGraph-Param}. Let $P' = Q_0\circ (y_0,z_1)\circ \ldots\circ Q_p$ be a (possibly non-simple) path from $u$ to $v$ in $H_{<(1+\eps)L}$; here $\circ$ is the path concatenation operator. Hence: 
	
		\begin{equation}
			\begin{split}
				w(P') &\leq (2k-1)(1+\eps)L + (2k)\beta\eps L \leq (2k-1)(1+\eps +  2\beta\eps)L\\
				 &\leq (2k-1)(1+(2\beta + 1)\eps)w(u,v) \qquad \mbox{(since $w(u,v)\geq L$)}\\
			\end{split}
		\end{equation}
	 
Thus, the stretch of edge $(u,v)$ is at most $(2k-1)(1+(2\beta + 1)\eps)$, as required. 
\end{proof}

 We are now ready to prove \Cref{thm:general-fast}.  
 \begin{proof}[Proof of \Cref{thm:general-fast}]   We use algorithm $\ssa_{\gen}$ in place of the abstract \hyperlink{SPHigh}{$\ssa$} in \Cref{lm:framework} to construct the light spanner. By \Cref{lm:App-Gen}, we have $s_{\ssa}(\beta) = (2\beta + 1)$, $\chi = O(n^{1/k})$ and $\tau(m',n') = O(1)$. Thus, by plugging in the values of $\chi$ and $\tau$, we obtain the lightness and the running time as required by \Cref{thm:general-fast}. The stretch of the spanner is:
 	\begin{equation*}
 		(2k-1)(1 + (s_{\ssa}(O(1)) + O(1))\eps)  = (2k-1)(1 + O(\eps))~.
 	\end{equation*}
  By scaling, we get the required stretch of $(2k-1)(1+\eps)$.
 \end{proof}

 \subsection{Minor-free Graphs}\label{subsec:minor-free}

Here we prove a weaker version of \Cref{thm:minor-free-fast}, where the running time is $O(nr\sqrt{r}\alpha(nr\sqrt{r},n))$. In \Cref{sec:minor-linear} we show how to achieve a linear running time, via an adaptation of our framework (described in detail in \Cref{sec:framework}) to minor-free graphs.

The implementation of the abstract algorithm \hyperlink{SPHigh}{$\ssa$} for minor-free graphs, hereafter $\ssa_{\minor}$, simply outputs the edge set $\me$. Note that the stretch in this case is $t = 1+\eps$.

 \begin{tcolorbox}
 	\hypertarget{SPHMinor}{}
 	\textbf{$\ssa_{\minor}$ (Minor-free Graphs):} The input is a  $(L,\eps,\beta)$-cluster graph $\mg(\mv,\me,\omega)$. The output is $\me^{\prune}$.
 	\begin{quote}
 		The algorithm returns $\me^{\prune} = \me$.
 	\end{quote}
 \end{tcolorbox}
 
We next analyze the running time of $\ssa_{\minor}$, and also show that it satisfies the two properties of (\hyperlink{Sparsity}{Sparsity}) and (\hyperlink{Stretch}{Stretch})
required by the abstract \hyperlink{SPHigh}{$\ssa$}. To this end, we use the following result:
 
 \begin{lemma}[Kostochka~\cite{Kostochka82} and Thomason~\cite{Thomason84}]\label{lm:minor-sparsity} Any $K_r$-minor-free graph with $n$ vertices has $O(r\sqrt{\log r}n)$ edges. 
 \end{lemma}
 
 \begin{lemma}\label{lm:App-Minor}\hypertarget{SPHGeneral}{$\ssa_{\minor}$} can be implemented in $O((|\mv| + |\me|))$ time. 
Furthermore, 1. (Sparsity) $\me^{\prune} = O(r\sqrt{\log r})|\mv|$, and 2. (Stretch) For each edge $(\varphi_{C_u},\varphi_{C_v})\in \me$, $d_{H_{<(1+\eps)L}}(u,v)\leq (1+\eps)(1+s_{\ssa_{\minor}}(\beta)\eps)w(u,v)$, where $(u,v) = 
\source(\varphi_{C_u}, \varphi_{C_v})$, $s_{\ssa_{\minor}}(\beta) = 0$ and $\eps \leq 1$. 
 \end{lemma}
 \begin{proof} The running time of $\ssa_{\minor}$ follows trivially from the construction. Noting that $\mg$ is a minor of the input graph $G$, $\mg$ is $K_r$-minor-free. Thus, $|\me| = O(r\sqrt{\log r})|\mv|$ by \Cref{lm:minor-sparsity}; this implies Item 1.  Since we take every edge of $\me$ to $\me^{\prune}$, the stretch is $1$ and hence $s_{\ssa_{\minor}}(\beta) = 0$, yielding Item 2. 
 \end{proof}
 
  We are now ready to prove a weaker version of \Cref{thm:minor-free-fast} for minor-free graphs, where the running time is $O(nr\sqrt{r}\alpha(nr\sqrt{r},n))$.
 \begin{proof}  We use algorithm $\ssa_{\minor}$ in place of the abstract \hyperlink{SPHigh}{$\ssa$} in \Cref{lm:framework} to construct the light spanner. By \Cref{lm:App-Minor}, we have $s_{\ssa}(\beta) = 0$, $\chi =  O(r\sqrt{\log r})$ and $\tau(m',n') = O(1)$. Thus, by plugging in the values of $\chi$ and $\tau$, we obtain the lightness claimed in \Cref{thm:minor-free-fast} and a running time of $O(nr\sqrt{r}\alpha(nr\sqrt{r},n))$, for a constant $\eps$. The stretch of the spanner is:
 	\begin{equation*}
 		(1+\eps)(1 + (s_{\ssa}(O(1))+ O(1))\eps)  = (1 + O(\eps))
 	\end{equation*}
 By scaling, we get a stretch of $(1+\eps)$.
 \end{proof}

\section{Unified Framework: Proof of \texorpdfstring{\Cref{lm:framework}}{Framework}}\label{sec:framework}
 
In \Cref{subsec:framework-details}, we outline a technical framework that we use to design a fast algorithm for constructing light spanners.  In the companion paper, we build on this framework to construct light spanners with fine-grained optimality. In \Crefrange{subsec:DesignPotential}{subsec:ConstructHi}, we provide a detailed implementation of the framework outlined in \Cref{subsec:framework-details} with a specific goal of obtaining a fast construction. In particular, \Cref{lm:framework} will be proved in \Cref{subsec:ConstructHi}. We refer readers to \Cref{subsec:proof-overview} for the overview of the proof.  We will introduce more notation in this section, summarized in \Cref{table:notation}. 

\subsection{The Framework}\label{subsec:framework-details}

Let $\MST$ be a minimum spanning tree of the input $n$-vertex $m$-edge graph $G = (V,E,w)$. Let $T_{\MST}$ be the running time needed to construct $\MST$. By scaling, we shall assume w.l.o.g.\ that the minimum edge weight is $1$. Let $\bar{w} = \frac{w(\mst)}{m}$. We remove from $G$ all edges of weight larger than $w(\MST)$; such edges do not belong to any shortest path, hence 
removing them does not affect the distances between vertices in $G$.  We define two sets of edges, $E_{light}$ and $E_{heavy}$, as follows:
\begin{equation}\label{Eprimedef}
	E_{light} = \{e\in E: w(e) \leq \frac{\bar{w}}{\eps}\} \qquad \& \qquad E_{heavy} = E\setminus E_{light}
\end{equation}

It is possible that $\frac{\bar{w}}{\eps} < 1$; in this case, $E_{light} = \emptyset$. 
The next observation is implied by the definition of $\bar{w}$.

\begin{observation}\label{obs:Eprime-weight} $w(E_{light}) \leq  \frac{w(\mst)}{\epsilon}$.
\end{observation}

Recall that the parameter $\eps$ is in the stretch $t(1+\eps)$ in \Cref{lm:framework}. It controls the stretch blow-up in \Cref{lm:framework}, and ultimately, the stretch of the final spanner. There is an inherent trade-off between the stretch blow-up (a factor of $1+\eps$) and the blow-up of the other parameters, including running time and lightness, by at least a factor of $1/\eps$. 

By \Cref{obs:Eprime-weight}, we can safely add $E_{light}$ to our final spanner while paying only an additive $+\frac{1}{\epsilon}$ term to the lightness bound. Hence, by~\Cref{stretch:ob},
in the spanner construction that follows, it suffices to focus on the stretch for edges in $E_{heavy}$.
 Next, we partition the edge set $E_{heavy}$ into subsets of edges, such that for any two edges $e, e'$ in the same subset, their weights are either \emph{almost the same} (up to a factor of $1+\psi$)
or they are \emph{far apart} (by at least a factor of $\frac{1}{\eps(1+\psi)}$), where $\psi$ is a parameter to be optimized later. In the applications of our framework in this paper, we choose $\psi = \eps$; in our companion paper, we choose $\psi = 1/250$. 

\begin{definition}[Partitioning $E_{heavy}$]\label{def:refineEdprime} Let $\psi$ be any parameter
in the range  $(0,1]$. Let $\mu_{\psi} = \lceil \log_{1+\psi}\frac{1}{\eps}\rceil$. We partition  $E_{heavy}$ into subsets $\{E^{\sigma}\}_{\sigma \in [\mu_\psi] }$ such that $ E^{\sigma} = \cup_{i\in \mathbb{N}^+} E^{\sigma}_i$ where:
	\begin{equation}\label{eq:Esigmaixdef}
		E^{\sigma}_i = \left\{e : \frac{L_i}{1+\psi} \leq w(e) < L_i \right\} \mbox{ with } L_i = L_{0}/\eps^i, L_0 = (1+\psi)^{\sigma}\bar{w}~. 
	\end{equation}
\end{definition}

By definition, we have $L_i = L_{i-1}/\eps$ for each $i \ge 1$.

Readers may notice that if $\log_{1+\psi}\frac{1}{\eps}$ is not an integer, by the definition of $E^{\sigma}$, it could be that $E^{\mu_{\psi}}\cap E^{1} \not= \emptyset$, in which case $\{E^{\sigma}\}_{\sigma \in [\mu_\psi] }$ is not really a partition of $E_{heavy}$. This can be fixed by defining $E^{\mu_\psi}$ as the set of edges that are not in $\cup_{1\leq \sigma \leq \mu_{\psi}-1} E^{\sigma}$. We henceforth assume that  $\{E^{\sigma}\}_{\sigma \in [\mu_\psi] }$  is a partition of $E_{heavy}$.

The following lemma shows that it suffices to focus on the stretch of edges in $E^{\sigma}$, for an arbitrary $\sigma \in [\mu_{\psi}]$.

\begin{lemma}\label{lm:ReductionToEsigma} If for every $\sigma \in [\mu_{\psi}]$ and $k\geq 1$, we can construct a $k$-spanner $H^{\sigma}\subseteq G$ for $E^{\sigma}$ with lightness at most $\lt_{H^{\sigma}}$ (w.r.t. $\mst$)  in time $\tm_{H^{\sigma}}(m,n)$ (where $\lt_{H^{\sigma}}$ and $\tm_{H^{\sigma}}(m,n)$ do not depend on $\sigma$), then we can construct a $k$-spanner for $G$ with lightness $O\left(\frac{\lt_{H^{\sigma}}\log(1/\eps)}{\psi} + \frac{1}{\eps}\right)$ in time $O\left(\frac{\tm_{H^{\sigma}}(m,n) \log(1/\eps)}{\psi} + T_{\mst}\right)$.
\end{lemma}
\begin{proof}
Let $H$ be a graph with $V(H) = V(G)$ and $E(H) = E_{light} \cup \left( \cup_{\sigma \in [\mu_{\psi}]}H^{\sigma}\right)$. The fact that $H$ is a $k$-spanner of $G$ follows directly from~\Cref{stretch:ob}.

To bound the lightness of $H$, note that $\mu_{\psi} = O(\frac{\log(1/\eps)}{\log(1+\psi)})$.  
Since $\log(1+x)\geq x$  when $x\in (0,1]$:
\begin{equation*}
	\begin{split}
	    w\left( \cup_{\sigma \in [\mu_{\psi}]}H^{\sigma}\right) ~&\leq~ (\lt_{H^{\sigma}}\cdot \mu_{\psi})w(\mst)  ~\\&=~  O\left(\frac{\lt_{H^{\sigma}}\log(1/\eps)}{\log(1+\psi)}\right)w(\mst) ~\\ &=~ O\left(\frac{\lt_{H^{\sigma}}\log(1/\eps)}{\psi}\right)w(\mst)~.
	\end{split}
\end{equation*}
The lightness bound of $H$ now follows from \Cref{obs:Eprime-weight}.

To bound the running time, note that the time needed to construct $E_{light}$ is $T_{\mst} + O(m) = O(T_{\mst})$. Since we remove edges of weight at least $w(\mst)$ from $G$ and every edge in $E_{heavy}$ has a weight at least $\frac{\bar{w}}{\eps} = \frac{w(\mst)}{\eps m}$, the number of sets that each $E^{\sigma}$ is partitioned to is $O(\log_{1/((1+\psi)\eps)}(\eps m)) = O(\log(m))$ for any $\eps \leq 1/2$. Thus, the partition of $E_{heavy}$ can be trivially constructed in $O(m)$ time. It follows  that the running time to construct $H$ is:
\begin{eqnarray*}
	\mu_{\psi} \tm_{H^{\sigma}}(n,m) + O(T_{\mst})  + O(m) ~=~ O\left(\frac{\tm_{H^{\sigma}}(n,m) \log(1/\eps)}{\log(1+\psi)} + 
	T_{\mst} \right) \\ ~=~ O\left(\frac{\tm_{H^{\sigma}}(n,m) \log(1/\eps)}{\psi} + T_{\mst}\right), 
\end{eqnarray*}
as desired.
\end{proof}

 We shall henceforth focus on constructing a spanner for $E^{\sigma}$, for an arbitrarily fixed $\sigma \in [\mu_{\psi}]$. In what follows, we present a clustering framework for constructing a spanner $H^{\sigma}$ for $E^{\sigma}$  with \emph{stretch $t(1+\epsilon)$}. We will assume that $\epsilon$ is sufficiently smaller than $1$.

 We call edges in $E^{\sigma}_i$ in \Cref{eq:Esigmaixdef} \emph{level $i$ edges}. Our construction considers edges in  $E^{\sigma}$ by level, starting from level $1$. The order between edges within the same level considered by our algorithm is not important.   An important subtlety of our construction is that when we consider level $i$ edges, we will assume that all the edges of length strictly smaller than $L_i$, \emph{including those in $ E\setminus E^{\sigma}$}, are already preserved with stretch $t(1+\rho\eps)$ in the current spanner $H_{< L_i}$. We can inductively assume this because we will alternate between different values of $\sigma$ in our construction. More precisely, we consider edges at level $1$ of $E^{\sigma}$ \emph{for all $\sigma\in [1, \mu_{\psi}]$} by the order from smaller $\sigma$ to larger $\sigma$, then edges at level $i$ for all $\sigma$, and then edges at level $3$, and so on.  A different way to look at this is that the edges at the same level of all $\sigma$ will be considered by increasing  $\sigma$, while edges at different levels will be considered by increasing levels. This subtlety is needed since in \hyperlink{SPHigh}{the $\ssa$}, we assume a good stretch spanner for all edges of length strictly smaller than $L$, which will be $L_i$ when we consider edges at level $i$ of $E^{\sigma}$.

\paragraph{Subdividing $\mst$} We subdivide each edge $e \in E(\mst)$ of weight more than $\bar{w}$ into $\lceil \frac{w(e)}{\bar{w}} \rceil$ edges of weight (of at most $\bar{w}$ and at least $\bar{w}/2$ each) that sums to $w(e)$. (New edges do not have to have equal weights.)  Let $\widetilde{\mst}$ be the resulting subdivided $\mst$.
We refer to vertices that are subdividing the $\mst$  edges as \emph{virtual vertices}. Let $\tilde V$ be the set of vertices in $V$ and virtual vertices; we call $\tilde{V}$  the {\em extended set} of vertices. Let $\tilde G = (\tilde V,\tilde E)$ be the graph that consists of the edges in $\widetilde{\mst}$ and $E^{\sigma}$. 
\begin{observation}\label{obs:Edge-Gtilde}
	$|\tilde{E}|  = O(m)$.
\end{observation}
\begin{proof}
It suffices to show that
$|E(\widetilde{\mst})| = O(m)$. Indeed, since $w(\widetilde{\mst}) = w(\mst)$ and each edge of $\widetilde{\mst}$ has weight at least $\bar{w}/2$, we have $|E(\widetilde{\mst})| \le 2m$.	%
\end{proof}

The $t(1+\eps)$-spanner that we construct for $E^{\sigma}$ is a subgraph of $\tilde{G}$ containing all edges of $\widetilde{\mst}$; we can enforce this assumption by adding the edges of $\widetilde{\mst}$ to the spanner. By replacing the edges of $\widetilde{\mst}$ by those of $\mst$, we can transform any subgraph of $\tilde{G}$ that contains the entire tree $\widetilde{\mst}$ to a subgraph of $G$ that contains the entire tree $\mst$. We denote  by $\tilde H^{\sigma}$ the $t(1+\eps)$-spanner of $E^{\sigma}$ in $\tilde{G}$;
by abusing the notation, we will write $H^{\sigma}$ rather than $\tilde H^{\sigma}$ in the sequel,
under the understanding that in the end we transform $H^{\sigma}$ to a subgraph of $G$.

Recall that $E^{\sigma} = \cup_{i\in \mathbb{N}^+} E^{\sigma}_i$ where $E^{\sigma}_i$ is the set of edges defined in \Cref{eq:Esigmaixdef}. We refer to   edges in $E^{\sigma}_i$ as   \emph{level-$i$ edges}. We say that a level $i$ is empty if the set $E^{\sigma}_i$  of level-$i$ edges is empty;
in the sequel, we shall only consider the nonempty levels. 

\begin{claim}\label{clm:numLevels} The number of (nonempty) levels is $O(\log m)$.
\end{claim}
\begin{proof}
Note that every edge of $E^{\sigma}$ has weight at least $\frac{\bar{w}}{\epsilon}$ and at most $w(\mst) = m\bar{w}$. Furthermore, the weight of any edge in $E^\sigma_{i+1}$ is at least $\frac{1}{(1+\psi)\epsilon} $ times the weight of any edge in 	$E^\sigma_{i}$ for any $i \in \mathbb{N}^+$. Note that $\frac{1}{(1+\psi)\epsilon} \geq \frac{1}{2\eps}$ since $\psi \leq 1$. Therefore, the number of levels is $O\left(\frac{\log(m \eps)}{\log(1/(2\eps))} \right) = O(\log(m))$ for any $\eps \leq 1/2$.
\end{proof}

Our construction crucially relies on a \emph{hierarchy of clusters}. A {\em cluster} in a graph is simply a subset of vertices in the graph.  Nonetheless, as will become clear soon, we care also about {\em edges} connecting vertices in the cluster and of the properties that these edges possess. Our hierarchy of clusters, denoted by $\mathcal{H} = \{\mathcal{C}_1,\mathcal{C}_2, \ldots \}$ satisfies the following properties:  
	\begin{itemize}  [noitemsep] 

		\item \textbf{(P1)~} 	\hypertarget{P1}{} For any $i\geq 1$, each $\mathcal{C}_i$ is a partition of $\tilde{V}$. When $i$ is large enough, $\mathcal{C}_i$ contains a single set $\tilde{V}$ and $\mathcal{C}_{i+1} = \emptyset$.
		\item \textbf{(P2)~} \hypertarget{P2}{} $\mathcal{C}_i$ is an \emph{$\Omega(\frac{1}{\eps})$-refinement} of $\mathcal{C}_{i+1}$, i.e., every cluster $C\in \mathcal{C}_{i+1}$ is obtained as the union of $\Omega(\frac{1}{\epsilon})$ clusters in $\mathcal{C}_i$ for $i\geq 1$.
		\item \textbf{(P3)~} \hypertarget{P3}{} For each cluster $C\in \mathcal{C}_i$, we have $\dm(H^{\sigma}[C]) \leq g L_{i-1}$, for a sufficiently large constant $g$ to be determined later. (Recall that $L_i$ is defined in \Cref{eq:Esigmaixdef}.) 
	\end{itemize}

\begin{remark}\label{remark:diameter} (1) We construct $H^{\sigma}$ along with the cluster hierarchy. Suppose that at some step $s$ of the algorithm, we construct a level-$i$ cluster $C$. Let $H^{\sigma}_s$ be $H^{\sigma}$ at step $s$. We shall maintain (\hyperlink{P3}{P3}) by maintaining the invariant that $\dm(H^{\sigma}_s[C]) \leq g L_{i-1}$; indeed,  adding more edges in later steps of the algorithm does not increase the diameter of the subgraph induced by $C$. 

(2) It is time-consuming to compute the diameter of a cluster {\em exactly}.  Thus, we explicitly associate with each cluster $C\in \mathcal{C}_i$ a proxy parameter of the diameter during the course of the construction. This proxy parameter has two properties: (a) it is at least the diameter of the cluster, and (b) it is lower-bounded by $\Omega(L_{i-1})$. Property (a) is crucial in arguing for the stretch of the spanner. Property (b) is crucial to have an upper bound on the number of level-$i$ clusters contained in a level-$(i+1)$ cluster, which speeds up its (the level-$(i+1)$ cluster's)  construction.
\end{remark}

 When $\eps$ is sufficiently small, specifically smaller than the constant hiding in the $\Omega$-notation in property (\hyperlink{P2}{P2}) by at least a factor of 2, it holds that $|\mc_{i+1}| \leq |\mc_i|/2$, yielding a geometric decay in the number of clusters at each level of the hierarchy. This geometric decay is crucial to our fast constructions. 

Our construction of the cluster hierarchy $\mathcal{H}$ will be carried out level by level, starting from level $1$. After we construct the set of level-$(i+1)$ clusters, we compute a subgraph $H^{\sigma}_{i}\subseteq G$ as stated in \Cref{lm:framework}. The final spanner $H^{\sigma}$ is obtained as the union of all subgraphs $\{H^{\sigma}_i\}_{i \in \mathbb{N}^+}$.  To bound the weight of $H^{\sigma}$, we rely on a potential function $\Phi$ that is formally defined as follows:

\begin{definition}[Potential Function $\Phi$]\label{def:Potential}  We use a potential function $\Phi: 2^{\tilde{V}}\rightarrow \mathbb{R}^+$ that maps each cluster $C$ in the hierarchy $\mathcal{H}$ to a potential value $\Phi(C)$, such that the total potential of clusters at level $1$ satisfies:
	\begin{equation}\label{eq:Phi1}
		 \sum_{C\in\mathcal{C}_1}\Phi(C) ~\leq~ w(\MST)~.
	\end{equation}
 Level-$i$ potential is defined as $\Phi_i = \sum_{C\in \mathcal{C}_i} \Phi(C)$ for any $i\geq 1$. The \emph{potential change} at level $i$, denoted by $\Delta_i$ for every $i \geq 2$, is defined as:
\begin{equation}\label{eq:PotentialReduction}
	 \Delta_i ~=~ \Phi_{i-1} - \Phi_{i}~. 
\end{equation}
\end{definition} 

The key to our framework is \Cref{lm:framework-technical} below. There are many parameters in \Cref{lm:framework-technical}, but the most important one is $\lambda$, which basically bounds the total weight of the edges that we add at every level and will play a major role in the final lightness bound.  Ideally, we want the inequality $w(H_i) \leq \lambda \Delta_{i+1}$ to hold for every $i$, but unfortunately, this is not the case. What we are able to show is that, $w(H_i) \leq \lambda \Delta_{i+1} + a_i$ for a small $a_i > 0$ in the sense that  $\sum_{i \in \mathbb{N}^+} a_i \leq A\cdot w(\mst)$ for some small parameter $A$. Therefore, we need the sequence $\{a_i\}_{i \in \mathbb{N}^+}$ in the statement of the lemma.  We note that parameter $A$ is the same for all $\sigma$.

\begin{restatable}{lemma}{FrameworkTechnical} \label{lm:framework-technical} Let $\rho\geq 1$ and $\eps_0 \in (0,1)$ be constants.  Let $\psi \in (0,1], t \geq 1, \eps \in (0,\eps_0)$ be  parameters, and $E^{\sigma}= \cup_{i\in \mathbb{N}^+} E^{\sigma}_i$ be the set of edges defined in Equation~\eqref{eq:Esigmaixdef}. Let $\{a_i\}_{i \in \mathbb{N}^+}$ be a sequence of positive real numbers such that $\sum_{i \in \mathbb{N}^+} a_i \leq A\cdot w(\mst)$ for some $A\in \mathbb{R}^+$. Let $H_0 = \mst$. For any level $i\geq 1$, assume that we can compute all subgraphs  $H_1,\ldots,H_i\subseteq G$  as well as the cluster sets $\{\mathcal{C}_{1},\ldots,\mathcal{C}_{i},\mathcal{C}_{i+1}\}$ in total   $O(\sum_{j=1}^i(|\mathcal{C}_j| + |E^{\sigma}_j|)f(n,m) + m)$ time for some function $f(\cdot,\cdot)$ such that:
	\begin{enumerate}[noitemsep]
		\item[(1)] $w(H_i) \leq  \lambda \Delta_{i+1} + a_i$ for some $\lambda \geq 0$,
		\item[(2)] for every $(u,v)\in E^{\sigma}_i$, $d_{H_{< L_i}}(u,v)\leq t(1+ \rho\cdot \epsilon)w(u,v)$ where $H_{< L_i}$ is the spanner constructed for edges of $G$ of weight less than $L_i$. 
	\end{enumerate}
	Then  we  can construct a $t(1+ \rho \eps)$-spanner for $G(V,E)$ with lightness  $O(\frac{\lambda + A + 1}{\psi}\log \frac{1}{\epsilon} + \frac{1}{\eps})$ in time $ O(\frac{mf(n,m)}{\psi}\log \frac{1}{\epsilon} + T_{\mst})$ when $\eps \in (1,\eps_0)$.
\end{restatable}

\begin{proof} Let $H^{\sigma} = \cup_{i\in \mathbb{N}} H_i$.  Note that $w(\msttilde) = w(\mst)$, since $\msttilde$ is simply a subdivision of $\mst$.  By condition (1) of \Cref{lm:framework}, 
	\begin{equation}\label{eq:LightnessHsigma}
		\begin{split}
			w(H^{\sigma}) &\leq \lambda \sum_{i \in \mathbb{N}^+}\Delta_i + \sum_{i \in \mathbb{N}^+} a_i + w(\mst) ~\leq~ 
			\lambda\cdot \Phi_1 + A\cdot w(\mst) + w(\mst)\\
			& \leq (\lambda + A+1)w(\mst) \quad \mbox{(by \Cref{eq:Phi1})}
		\end{split}
	\end{equation}
	
	\Cref{eq:LightnessHsigma} and \Cref{lm:ReductionToEsigma} imply the lightness upper bound; here $\lt_{H^{\sigma}} = (\lambda + A+1)$. The stretch bound $t(1+\rho\eps)$ follows directly from the fact that $E^{\sigma}= \cup_{i\in \mathbb{N}^+}E^{\sigma}_{i}$, Item (2), and \Cref{lm:ReductionToEsigma}. 
	
	To bound the running time, we note that $\sum_{i \in \mathbb{N}^+}|E^{\sigma}_i| \leq m$ and by property (\hyperlink{P2}{P2}), we have $\sum_{i \in \mathbb{N}^+} |\mathcal{C}_i| = |\mathcal{C}_1|\sum_{i \in \mathbb{N}^+}\frac{O(1)}{\epsilon^{i+1}} = O(|\mathcal{C}_1|) = O(m)$. Thus, by the assumption of  \Cref{lm:framework}, the total running time to construct $H^{\sigma}$ is:
	\begin{equation*}
		\tm_{H^{\sigma}}(m,n) = O\left((\sum_{i \in \mathbb{N}^+}(|\mathcal{C}_i|) + |E_i|)f(m,n) + m\right) = O\left( mf(m,n)\right).
	\end{equation*}
	Plugging this running time bound on top of  \Cref{lm:ReductionToEsigma} for all $\sigma \in [\mu_\psi]$ yields the required running time bound in~\Cref{lm:framework}. %
\end{proof}

\begin{remark}\label{remark:stretchLevel-i}In \Cref{lm:framework-technical}, we construct spanners for edges of $G$ level by level, starting from level $1$. By Item (2), when constructing spanners for edges in $E^{\sigma}_i$, we could assume by induction that all edges of weight less than $L_i/(1+\psi)$ already have stretch $t(1+\rho\eps)$ in the spanner constructed so far, denoted by $H_{< L_i/(1+\psi)}$. By defining $H_{<L_i} = H_{< L_i/(1+\psi)}\cup H_i$, we get a spanner for edges of length less than $L_i$. 
\end{remark}

In summary, two important components in our spanner construction are a hierarchy of clusters and a potential function as defined in \Cref{def:Potential}.  In \Cref{subsec:DesignPotential}, we present a construction of level-$1$ clusters and a general principle for assigning potential values to clusters.  In \Cref{subsec:LeveIplus1Construction}, we outline an efficient construction of clusters at any level $i+1$ for $i\geq 1$. The details of the construction are deferred to \Cref{sec:ClusteringDetails}. In \Cref{subsec:ConstructHi}, we present a general approach for constructing $H_i$. Our construction of $H_i$ assumes the existence of  \hyperlink{SPHigh}{$\ssa$} stated in \Cref{subsec:framework-intro}.

 \subsection{Designing A Potential Function}\label{subsec:DesignPotential}
 
In this section, we present in detail the underlying principle used to design the potential function $\Phi$ in \Cref{def:Potential}. We start by constructing and assigning potential values for level-$1$ clusters.

 \begin{lemma}\label{lm:level1Const} In time $O(m)$, we can construct a set of level-$1$ clusters $\mathcal{C}_1$ such that, for each cluster $C\in \mathcal{C}_1$, the subtree $\msttilde[C]$
 	of $\msttilde$ induced by $C$ is connected and satisfies $L_0 \leq \dm(\msttilde[C]) \leq 7L_0$. 
 \end{lemma}
 \begin{proof} We first break $\msttilde$ into a set $\mathcal{S}$ of subtrees of diameter at least $L_0$ and at most $4L_0$ as follows. We root $\msttilde$ at an arbitrary vertex $r$ and visit $\msttilde$ in post-order. At each vertex $v$, we keep track of the weight of the maximum-weight path ending at $v$ in the subtree rooted at $v$, denoted by $\mathsf{w}_v$. Whenever we finish visiting a child $u$ of $v$, we update $\mathsf{w}_v\leftarrow \max\{\mathsf{w}_v, \mathsf{w}_u + w(u,v)\}$. Once all children of $v$ are visited, if $\mathsf{w}_v \geq L_0$, we cut the subtree rooted at $v$ out of $\msttilde$ and add it to  $\mathcal{S}$. In such a case, when returning to the parent $x$ of $v$, since the subtree rooted at $v$ is removed from the tree, we do not update $\mathsf{w}_x\leftarrow \max\{\mathsf{w}_x, \mathsf{w}_v + w(x,v)\}$; the post-order traversal will continue to visit the next child of $x$, if any.

 Observe that (i) each subtree in $\mathcal{S}$ has diameter at least $L_0$ and at most $2(L_0+\bar{w}) \leq 4L_0$ and (ii) $\mathcal{S}$ can be constructed in $O(m)$ time, as $\msttilde$ has $O(m)$ vertices and edges. 

After removing all vertices in $\mathcal{S}$, there is at most one remaining subtree, say $T'$, of $\msttilde$ left, which has diameter at most $2L_0$. There must be an $\msttilde$ edge $e$ connecting $T'$ and  a subtree $T\in \mathcal{S}$. Then we add $T'$ and $e$ to $T$.   Since $T$ is augmented by subtrees of diameter at most $2L_0$ via an $\widetilde{\mst}$ edge, the diameter of $T$ after the augmentation is at most $4L_0 + 2L_0 + \bar{w} \leq 7L_0$.  Finally, we form $\mathcal{C}_1$ by taking the vertex set of each subtree in $\mathcal{S}$ to be a level-$1$ cluster.  The total running time is dominated by the running time to construct $\mathcal{S}$, which is $O(m)$. 
 \end{proof}

We note that a cluster $C\in \mc_1$ in \Cref{lm:level1Const} could contain only virtual vertices.  By choosing $g\geq 7$, clusters in $\mathcal{C}_1$ satisfy properties (\hyperlink{P1}{P1}) and (\hyperlink{P3}{P3}). Note that (\hyperlink{P2}{P2}) is not applicable to level-$1$ clusters by definition. As for (\hyperlink{P3}{P3}), $\dm(H^{\sigma}[C]) \leq 7 L_{0}$, for each $C \in \mathcal{C}_1$, since $H^{\sigma}$ includes all edges of $\msttilde$.

 Next, we assign a potential value for each level-$1$ cluster as follows:
 
 \begin{equation}\label{eq:Level1Poten}
 	\Phi(C) = \dm(\msttilde[C]) \qquad \forall C \in \mathcal{C}_1
 \end{equation}
 We now claim that the total potential of all clusters at level $1$ is at most $w(\mst)$ as stated in \Cref{def:Potential}.
 
 \begin{lemma}\label{lm:Level1Poten} $\Phi_1 \leq w(\mst)$. 
 \end{lemma}
 \begin{proof}
 	By definition of $\Phi_1$, we have: 
    \begin{align*}
        \Phi_1 ~&=~ \sum_{C\in \mathcal{C}_1} \Phi(C) ~=~ \sum_{C\in \mathcal{C}_1}\dm(\msttilde[C]) \\ &\leq~ \sum_{C\in \mathcal{C}_1} w(\msttilde[C]) ~\leq~ w(\msttilde) ~=~ w(\mst)~.
    \end{align*}
 	The penultimate inequality holds since level-$1$ clusters induce vertex-disjoint subtrees of $\msttilde$. 
 \end{proof}
 
 While the potential of a level-1 cluster is the diameter of the subtree induced by the cluster, the potential assigned to a cluster at level 2 or larger need not be the diameter of the cluster. Instead, it is an {\em overestimate} of the cluster's diameter, as imposed by the following {\em potential-diameter (PD) invariant}.
 
\hypertarget{PD}{}
 \begin{quote}
 	\textbf{PD Invariant:} For every $C \in \mathcal{C}_{i}$ and $i\geq 1$, $\dm(H_{< L_{i-1}}[C]) \leq \Phi(C)$. (Recall that $H_{< L_{i-1}}$ is the spanner constructed for edges of $G$ of weight less than $L_{i-1}$, as defined in \Cref{lm:framework-technical}.)
 \end{quote}
 
 \begin{remark}\label{remark:potential-diameter} As discussed in \Cref{remark:diameter}, it is time-expensive to compute the diameter of each cluster. By the \hyperlink{PD}{PD Invariant}, we can use the potential $\Phi(C)$ of a cluster $C \in \mathcal{C}_i$ as an upper bound on the diameter of $H_{< L_{i-1}}[C]$. As we will demonstrate in the sequel, $\Phi(C)$ can be computed efficiently.
 \end{remark}
 
 To define potential values for clusters at levels $2$ or larger, we introduce a {\em cluster graph}, in which the nodes correspond to clusters.
 We shall derive the potential values of clusters via their \emph{structure} in the cluster graph, as described next.

 \begin{definition}[Cluster Graph]\label{def:ClusterGraphNew} A cluster graph at level $i \geq 1$, denoted by $\mg_i = (\mv_i, \me'_i, \omega)$, is a \emph{simple graph} where each node corresponds to a cluster in $\mc_i$ and each inter-cluster edge $(\varphi_{C_u},\varphi_{C_v})$ is mapped to an edge $(u,v)\in \tilde{G} $ for some $u \in C_u$ and $v\in C_v$.  We assign  weights to both \emph{nodes and edges} as follows:  for each node $\varphi_C \in \mv_i$ corresponding to a cluster $C \in \mathcal{C}_i$, $\omega(\varphi_C) = \Phi(C)$, and for each edge $\mbe = (\varphi_{C_u},\varphi_{C_v}) \in \me'_i$ mapped to an edge $(u,v)$ of $\tilde{G}$, $\omega(\mbe) = w(u,v)$.   
 \end{definition} 

We remark that if there are multiple edges between the vertices of $C_u$ and $C_v$, it is often convenient to pick the edge with the smallest weight and assign this weight to $(\varphi_{C_u},\varphi_{C_v})$. However, doing so incurs additional time to keep track of the smallest weight edge in $\tilde{G}$ between every two clusters. Therefore, in our construction, the edge corresponding to $(\varphi_{C_u},\varphi_{C_v})$ might not have the smallest weight; see more details in \Cref{def:GiProp} below. 
 
 \begin{remark}The  notion of cluster graphs in \Cref{def:ClusterGraphNew} is slightly different from the notion of $(L,\eps,\beta)$-cluster graphs defined in \Cref{def:ClusterGraph-Param}. In particular, cluster graphs in \Cref{def:ClusterGraphNew} have weights on both edges and nodes, while $(L,\eps,\beta)$-cluster graphs in \Cref{def:ClusterGraph-Param} have weights on edges only. 
 \end{remark}
 
 In our framework, we want the cluster graph $\mg_i$ to have the following properties.
 
 \begin{definition}[Properties of $\mg_i$]\label{def:GiProp} \begin{enumerate}
 		\item[(1)] The edge set $\me'_i$ of $\mg_i$ is the union $\msttilde_{i}\cup \me_i$, where each edge $\msttilde_{i}$ corresponds to an edge in $\msttilde$ and $\me_i$ is the set of edges corresponding to \emph{a subset of} edges in $E^{\sigma}_i$.
 		\item[(2)] $\msttilde_{i}$ induces a spanning tree of $\mg_i$. We abuse notation by using $\msttilde_{i}$ to denote the induced spanning tree.
 		\item[(3)] $\mg_i$ has no \emph{removable edge}: an edge $(\varphi_{C_u},\varphi_{C_v}) \in \me_i$ is removable if (3a) the path $\msttilde_i[\varphi_{C_u},\varphi_{C_v}]$ between $\varphi_{C_u}$ and $\varphi_{C_v}$ only contains nodes in $\msttilde_{i}$ of degree at most $2$ and (3b) $\omega(\msttilde_i[\varphi_{C_u},\varphi_{C_v}]) \leq t(1 + 6g\eps)\omega(\varphi_{C_u},\varphi_{C_v})$.
  	\end{enumerate}
 \end{definition}

As we will show in the sequel, if an edge $(\varphi_{C_u},\varphi_{C_v})$ satisfies property (3b), there is a path of stretch at most $t(1+6g\eps)$ in $H_{< L_{i-1}}$ between $u$ and $v$ and hence, we do not need to consider edge $(u,v)$ in the construction of $H_i$.  To meet the required lightness bound, it turns out that it suffices to remove edges satisfying both properties (3a) and (3b), rather than removing all edges satisfying property (3b). More importantly, we can detect removable edges satisfying both  (3a) and (3b) faster than those that only satisfy (3b), since for (3b), we have to compute shortest distances in $\msttilde_i$ between $\varphi_{C_u}$ and $\varphi_{C_v}$, which is more time-consuming and complicated.

 At the outset of the construction of \emph{level-$(i+1)$ clusters},
 	we construct a cluster graph $\mg_i$. We assume that the spanning tree $\msttilde_{i}$ of $\mg_i$ is given, as we construct the tree by the end of the construction of level-$i$ clusters. After we complete the construction of level-$(i+1)$ clusters, we construct $\msttilde_{i+1}$ for the next level.

 \begin{observation}\label{obs:Level1MST} At level $1$, both $\mv_1$ and $\msttilde_1$ can be constructed in $O(m)$ time.
 \end{observation} 
 \begin{proof} Edges of $\msttilde_1$ correspond to the edges of $\msttilde$ that do not belong to
 	any level-1 cluster, i.e., to any $\msttilde[C]$, where $C \in \mathcal{C}_1$. Thus, the observation follows from \Cref{obs:Edge-Gtilde} and \Cref{lm:level1Const}. 
 \end{proof}

 \paragraph{The structure of level-$(i+1)$ clusters} 
 Next, we describe how to construct the level-$(i+1)$ clusters via the cluster graph $\mg_i$. 
 We shall construct a collection of subgraphs $\mathbb{X}$ of $\mg_i$, and then map each subgraph $\mx \in \mathbb{X}$ to a cluster $C_{\mx} \in \mathcal{C}_{i+1}$ as follows:
 \begin{equation}\label{eq:XtoCluster}
 	C_{\mathcal{X}}  = \cup_{\varphi_C\in \mv(\mx)} C~.
 \end{equation}

That is, $C_{\mathcal{X}}$ is the union of all level-$i$ clusters that correspond to nodes in $\mathcal{X}$. 
 
 For any subgraph $\mx$ in a cluster graph, we denote by $\mv(\mx)$ and $\me(\mx)$ the vertex and edge sets of $\mx$, respectively. To guarantee properties (\hyperlink{P1}{P1})-(\hyperlink{P3}{P3}) defined before \Cref{remark:diameter} for clusters in $\mathcal{C}_{i+1}$, we will make sure that subgraphs in $\mathbb{X}$  satisfy the following properties:
 
 \begin{itemize}[noitemsep]
 	\item \textbf{(P1').~} \hypertarget{P1'}{}  $\{\mv(\mx)\}_{\mx \in \mathbb{X}}$ is a partition of $\mv_i$.
 	\item \textbf{(P2').~} \hypertarget{P2'}{} $|\mv(\mx)| = \Omega(\frac{1}{\eps})$.
 	\item \textbf{(P3').~} \hypertarget{P3'}{} $L_i \leq \adm(\mx) \leq gL_{i}$.
 \end{itemize}
 
 Recall that $\adm(\mx)$ is the augmented diameter of $\mx$, a variant of diameter defined for graphs with weights on both nodes and edges; see \Cref{sec:prelim}. Recall that the augmented diameter of $\mx$ is at least the diameter of the corresponding cluster $C_{\mx}$. 
 
 We then set the potential of cluster $C_{\mx}$ corresponding to subgraph $\mx$ as:
 	\begin{equation}\label{eq:SetPotential-i}
 		\Phi(C_{\mx}) = \adm(\mx).
 	\end{equation}
 
  Thus, the augmented diameter of any such subgraph $\mx$ will be the weight of the corresponding node in the level-$(i+1)$ cluster graph $\mg_{i+1}$.   Our goal is to construct $H_i$ along with $\mathcal{C}_{i+1}$ as guaranteed by \Cref{lm:framework}.  $H_i$ consists of a subset of the edges in $E^{\sigma}_i$; we can assume that the vertex set of $H_i$ is just the entire set $V$. Up to this point, we have not explained yet how $H_i$ is constructed since the exact construction of $H_i$ depends on specific incarnations of our framework, which may change from one graph class to another. 
 
While properties (\hyperlink{P1'}{P1'}) and (\hyperlink{P2'}{P2'}) directly imply properties (\hyperlink{P1}{P1}) and (\hyperlink{P2}{P2}) of $C_{\mx}$, property (\hyperlink{P3'}{P3'}) does not directly imply property (\hyperlink{P3}{P3}); although the diameter of any weighted subgraph (with edge and vertex weights) is upper bounded by its augmented diameter, we need to guarantee that the (corresponding) edges of $\mx$ belong to $H_{< L_i}$. Indeed, without this condition, the diameter of $H_{< L_i}$ could be much larger than the augmented diameter of $\mx$.
 
 \begin{lemma}\label{lm:PropEquiv} Let $\mx \in \mathbb{X}$ be a subgraph of $\mg_i$ satisfying properties (\hyperlink{P1'}{P1'})-(\hyperlink{P3'}{P3'}). Suppose that for every edge $(\varphi_{C_u},\varphi_{C_v})\in \me(\mx)$, $(u,v) \in H_{< L_i}$.  By setting the potential value of $C_{\mx}$ to be $\Phi(C_{\mx}) = \adm(\mx)$ for every $\mx \in \mathbb{X}$, the \hyperlink{PD}{PD Invariant} is satisfied and  $C_{\mx}$ satisfies all properties (\hyperlink{P1}{P1})-(\hyperlink{P3}{P3}).
 \end{lemma}
 \begin{proof} It can be seen directly that properties (\hyperlink{P1'}{P1'}) and (\hyperlink{P2'}{P2'}) of $\mx$ directly imply properties (\hyperlink{P1}{P1}) and (\hyperlink{P2}{P2}) of $C_{\mx}$, respectively. We prove, by induction on $i$,
that property (\hyperlink{P3}{P3}) holds and that the  \hyperlink{PD}{PD Invariant} is satisfied. The basis $i=1$ is trivial.  For the induction step, we assume inductively that for each cluster $C\in \mathcal{C}_{i}$, $\dm(H_{< L_{i-1}})[C] \leq gL_{i-1}$ and that the \hyperlink{PD}{PD Invariant} is satisfied: $\Phi(C)\geq \dm(H_{< L_{i-1}})[C]$. Consider any level-$(i+1)$ cluster $C_{\mx}$ corresponding to a subgraph $\mx \in \mathbb{X}$. Let $H_{C_{\mx}}$ be the graph obtained by first taking the union $\cup_{\varphi_{C} \in \mv(\mx)}H_{< L_{i-1}}[C]$  and then adding in the edge set  $ \{(u,v)\}_{(\varphi_{C_u},\varphi_{C_v}) \in \me(\mx)}$.  Observe that $H_{C_{\mx}}$ is a subgraph of $H_{< L_i}$ by the assumption that  $(u,v) \in H_{< L_i}$ for every edge $(\varphi_{C_u},\varphi_{C_v})\in \me(\mx)$. We now show that $\dm(H_{C_{\mx}}) \leq \adm(\mx)$, which is at most $gL_i$ by property (\hyperlink{P3'}{P3'}). This would imply both property (\hyperlink{P3}{P3}) and the \hyperlink{PD}{PD Invariant} for $C_{\mx}$ since $\Phi(C_{\mx}) = \adm(\mx)$, which would complete the proof of the induction step.
 	
 	Let $u,v$ be any two vertices in $H_{C_{\mx}}$ whose shortest distance in $H_{C_{\mx}}$ realizes $\dm(H_{C_{\mx}})$. Let $\varphi_{C_u}, \varphi_{C_v}$ be the two nodes in $\mx$ that correspond to two clusters $C_u,C_v$ containing $u$ and $v$, respectively. Let $\mathcal{P}_{u,v}$ be a path in $\mg_i$ of minimum augmented weight between $\varphi_{C_u}$ and $\varphi_{C_v}$. Observe that $\omega(\mathcal{P}_{u,v}) \leq \adm(\mx)$. We now construct a path $P_{u,v}$ between $u$ and $v$ in $H_{C_{\mx}}$ as follows. We write $\mathcal{P}_{u,v} \stackrel{\mbox{\tiny{def.}}}{=}  (\varphi_{C_u} =\varphi_{C_1}, \mbe_1,\varphi_{C_2},\mbe_2, \ldots, \varphi_{C_\ell} = \varphi_{C_v})$ as an alternating sequence of nodes and edges. For every $1\leq p \leq \ell-1$, let $(u_p,v_p)$ be the edge in $E^{\sigma}_i$ that corresponds to $\mbe_p$. We then define
	$v_0 = u, u_\ell = v$ and
 		\begin{equation*}
 			P_{u,v} =Q_{H_{< L_{i-1}}[C_1]}(v_0,u_1) \circ (u_1,v_1)\circ Q_{H_{< L_{i-1}}[C_2]}(v_1,u_2) \circ \ldots \circ  Q_{H_{< L_{i-1}}[C_{\ell}]}(v_{\ell-1},u_\ell)~,
 		\end{equation*}
 	where $Q_{H_{< L_{i-1}}[C_p]}(v_{p-1},u_p)$ for $1\leq p \leq \ell$ 
	is the shortest path in the corresponding subgraph (between the endpoints of the respective edge, as specified in all the subscripts), and $\circ$ is the path concatenation operator. By the induction hypothesis for the \hyperlink{PD}{PD Invariant} and $i$, 
	$w(Q_{H_{< L_{i-1}}[C_p]}(v_{p-1},u_p)) \leq \omega(\varphi_{C_p})$ for each $1\leq p \leq \ell$. Thus, $w(P_{u,v}) \leq \omega(\mathcal{P}_{u,v}) \leq \adm(\mx)$. It follows that $\dm(H_{C_{\mx}}) ~\leq~ w(P_{u,v}) \leq  \adm(\mx)$ as desired.
 \end{proof}
 
 \paragraph{Local potential change}  For each subgraph $\mx \in \mathbb{X}$, we define the \emph{local potential change} of $\mx$, denoted by $\Delta_{i+1}(\mx)$ as follows:
 \begin{equation}\label{eq:LocalPotential}
 	\Delta_{i+1}(\mx) \stackrel{\mbox{\tiny{def.}}}{=}   \left(\sum_{\varphi_C\in \mv(\mx)} \Phi(C) \right) -  \Phi(C_{\mx}) = \left(\sum_{\varphi_C\in \mv(\mx)} \omega(\varphi_C) \right) - \adm(\mx). 
 \end{equation} 
 
 \begin{claim}\label{clm:localPotenDecomps}$\Delta_{i+1} = \sum_{\mx \in \mathbb{X}}\Delta_{i+1}(\mx)$.
 \end{claim}
 \begin{proof} By property (\hyperlink{P1}{P1}), subgraphs in $\mathbb{X}$ are vertex-disjoint and cover the vertex set $\mv_i$,  hence $\sum_{\mx \in \mathbb{X}}(\sum_{\varphi_C\in \mv(\mx)} \Phi(C)) = \sum_{C\in \mathcal{C}_i} \Phi(C) = \Phi_i$. Additionally, by the construction of level-$(i+1)$ clusters,  $\sum_{\mx \in \mathbb{X}}  \Phi(C_{\mx}) = \sum_{C'\in \mathcal{C}_{i+1}} \Phi(C') = \Phi_{i+1}$. Thus,  we have:
 	\begin{equation*}
 			\sum_{\mx \in \mathbb{X}}\Delta_{i+1}(\mx) = \sum_{\mx \in \mathbb{X}}\left(\left(\sum_{\varphi_C\in \mv(\mx)} \Phi(C) \right) -  \Phi(C_{\mx}) \right) = \Phi_i - \Phi_{i+1} = \Delta_{i+1}, 
 	\end{equation*}
	as claimed.
 \end{proof}

 The decomposition of the (global) potential change into local potential changes makes the task of analyzing the spanner weight (Item (1) in \Cref{lm:framework}) easier as we can do so locally. Specifically, we often construct $H_i$ by considering each node in $\mv_i$ and taking a subset of (the corresponding edges of) the edges incident to the node to $H_i$. We then calculate the number of edges taken to $H_i$ incident to all nodes in $\mx$, and bound their total weight by the local potential change of $\mx$. By summing up over all $\mx$, we obtain a bound on $w(H_i)$ in terms of the (global) potential change $\Delta_{i+1}$.

 \subsection{Constructing \texorpdfstring{Level-$(i+1)$}{Level-(i+1)} Clusters}\label{subsec:LeveIplus1Construction}
 
 To obtain a fast spanner construction, we will maintain for each cluster $C \in \mathcal{C}_i$ a {\em representative} vertex $r(C) \in C$. If $C$ contains at least one original vertex, then $r(C)$ is one original vertex in $C$; otherwise, $r(C)$ is a virtual vertex. (Recall that virtual vertices are those subdividing $\mst$ edges.)  For each vertex $v \in C$, we designate $r(C)$ as the {\em representative} of $v$, i.e., we set $r(v) = r(C)$ for each $v \in C$. We use the \textsc{Union-Find} data structure to maintain these representatives. Specifically, the representative of $v$ will be given as \textsc{Find}($v$). Whenever a level-$(i+1)$ cluster is formed from level-$i$ clusters, we call \textsc{Union} (sequentially on the level-$i$ clusters) to construct a new representative for the new cluster. 
 
 \paragraph{A careful usage of the Union-Find data structure} We will use the \textsc{Union-Find} data structure~\cite{Tarjan75}  for grouping subsets of clusters to larger clusters (via the \textsc{Union} operation) and checking whether two given vertices belong to the same cluster (via the \textsc{Find} operation). The amortized running time of each \textsc{Union} or \textsc{Find} operation is $O(\alpha(a,b))$, where $a$ is the total number of \textsc{Union} and \textsc{Find} operations and $b$ is the number of vertices in the data structure. Note, however, that our graph $\tilde{G}$ has $n$ original vertices but $O(m)$ virtual vertices, which subdivide $\mst$ edges. Thus, if we keep both original and virtual vertices in the \textsc{Union-Find} data structure, the amortized time of an operation will be $O(\alpha(m,m)) = O(\alpha(m))$ rather than $O(\alpha(m,n))$, as the total number of \textsc{Union} and \textsc{Find} operations is $O(m)$,  and will be super-constant for any super-constant value of $m$. 
 
 To reduce the amortized time to $O(\alpha(m,n))$, we only store original vertices in the \textsc{Union-Find} data structure. To this end, for each virtual vertex, say $x$, which subdivides an edge $(u,v) \in \mst$, we store a pointer, denoted by $p(x)$, which points to one of the endpoints, say $u$, in the same cluster with $x$, \emph{if there is at least one endpoint in the same cluster with $x$}. In particular, any virtual vertex has at most two possible clusters that it can belong to at each level of the hierarchy.  Hence, we can apply every \textsc{Union-Find} operation to $p(x)$ instead of $x$. For example,  to check whether two virtual vertices  $x$ and $y$ are in the same cluster,  we compare $r(p(x)) \stackrel{?}{=} r(p(y))$ via two \textsc{Find} operations. The total number of \textsc{Union} and \textsc{Find} operations in our construction remains $O(m)$ while the number of vertices that we store in the data structure is reduced to $n$. Thus, the amortized time of each operation reduces to $O(\alpha(m,n))$, and the total running time due to all these operations is $O(m \alpha(m,n))$.

If no endpoint of $(u,v)$ belongs to the same cluster with $x$, then the level-$i$ cluster containing $x$ is a path of virtual vertices subdivided from $(u,v)$. In this case, we simply let \textsc{Find}($x$) operation return $x$. That is, we will not maintain $x$ in the \textsc{Union-Find} data structure, but instead use a flag to mark if $x$ is in the same cluster with one of the endpoints $\{u,v\}$ or not. Also, we maintain virtual clusters, those that only have virtual vertices, in a regular list data structure, and \textsc{Union} operations can be implemented as the concatenation of two lists in $O(1)$ time. Once a virtual cluster is merged with a non-virtual cluster, all the virtual vertices need to update their flag and change their pointer $p(\cdot)$ accordingly.
 
 Following the approach in \Cref{subsec:DesignPotential},  we construct a graph $\mathcal{G}_i$ satisfying all properties in \Cref{def:GiProp}. Then we construct a set $\mathbb{X}$ of subgraphs of $\mg_i$ satisfying the three properties (\hyperlink{P1'}{P1'})-(\hyperlink{P3'}{P3'}) and a subgraph $H_i$ of $G$ (and of $\tilde{G}$ as well). Each subgraph $\mx\in \mathbb{X}$ is then converted to a level-$(i+1)$ cluster by \Cref{eq:XtoCluster}. 
 
 \paragraph{Constructing $\mathcal{G}_i$}  We shall assume inductively on $i, i \ge 1$ that:
 \begin{itemize}[noitemsep]
 	\item The set of edges $\widetilde{\mst}_i$ is given by the construction of the previous level $i$ in the hierarchy; for the base case $i = 1$ (see \Cref{subsec:DesignPotential}), $\widetilde{\mst}_1$ is simply a set of edges of $\widetilde{\mst}$ that are not in any level-$1$ cluster. 
 	\item The weight $\omega(\varphi_C )$ on each node $\varphi_C \in \mv_i$ is the potential value of cluster $C \in \mathcal{C}_i$; for the base case $i = 1$, the  potential values of level-$1$ clusters were computed in  $O(m)$ time, as discussed in \Cref{subsec:DesignPotential}.
 \end{itemize}
 
 By the end of this section, we will have constructed  
 the edge set $\widetilde{\mst}_{i+1}$ and the weight function on nodes  of $\mathcal{G}_{i+1}$, in  time $O(|\mathcal{V}_i|\alpha(m,n))$. Computing the weight function on nodes of $\mathcal{G}_{i+1}$ is equivalent to computing the augmented diameter of $\mx$, which in turn, is related to the potential function. The fact that we can compute all the weights efficiently in almost linear time is the crux of our framework.
 
 Note that we make no inductive assumption regarding the set of edges ${E^{\sigma}_i}$, which can be computed once in $O(m)$  overall time at the outset for all levels $i \ge 1$, since the edge sets $E^{\sigma}_{1}, E^{\sigma}_{2}, \ldots$ are pairwise disjoint and the number of levels is $O(\log m)$ by \Cref{clm:numLevels}.
 
 \begin{lemma}\label{lm:G_i-construction}We can construct $\mathcal{G}_i = (\mathcal{V}_i,\mathcal{E}_i\cup \widetilde{\mst}_i,\omega)$ 
 	in $O\left(\alpha(m,n)(|\mv_i|  + |E^{\sigma}_i|)\right)$ time, where $\alpha(\cdot,\cdot)$ is the inverse-Ackermann function.
 \end{lemma}
 \begin{proof}
 	Recall that any edge in $\widetilde{\mst}_i$ (of weight at most $\bar{w}$) is of strictly smaller weight than that of any edge in $E^{\sigma}_i$ (of weight at least $\frac{\bar{w}}{(1+\psi)\epsilon}$) for any $i\geq 1$ and $\epsilon \leq 1$. Note that  $\widetilde{\mst}_i$  and $E^{\sigma}_i$ are given at the  outset of the construction of $\mg_i$.  To construct the edge set $\mathcal{E}_i$, we do the following. 	For each edge $e = (u,v) \in E^{\sigma}_i$, we compute the representatives $r(u), r(v)$; this can be done in $O(\alpha(m,n))$ amortized time over all the levels up to $i$ using the \textsc{Union-Find} data structure. This is because the total number of \textsc{Union/Find} operations is bounded by $O(\sum_{1\leq j\leq i}|\mv_j|  + |E^{\sigma}_j|) = O(m)$. 
Equipped with the representatives, it takes $O(1)$ time to check whether $e$'s endpoints lie in the same level-$i$ cluster (equivalently, whether edge $e$ forms a self-loop in the cluster graph)---by checking whether $r(u) = r(v)$. In the same way, we can check in $O(1)$ time whether edges $e = (u,v)$ and $e' = (u',v')$ are parallel in the cluster graph---by comparing the representatives of their endpoints. Among parallel edges, we only keep the edge of minimum weight in $\mg_i$.  
 	
	Next, we remove all removable edges from $\mg_i$ as specified by properties (3a) and (3b) in \Cref{def:GiProp}. First we find in $O(|\mv_i|)$ time a collection $\mathbb{P}$ of \emph{maximal paths} in $\msttilde_{i}$ that only contain degree-$2$ vertices. By the maximality, paths in  $\mathbb{P}$  are node-disjoint. We then find for each path $\mathcal{P}\in \mathbb{P}$ a subset of edges $\me_{\mp} \subseteq \me_i$ whose both endpoints belong to $\mathcal{P}$; this can be done in $O(|\mv_i| + |E^{\sigma}_i|)$ total time for all paths in $\mathbb{P}$.  Finally, for each path $\mathcal{P} \in \mathbb{P}$ and each edge $(\varphi_{C_u},\varphi_{C_v}) \in \me_{\mp}$, we can compute $\omega(\mathcal{P}[\varphi_{C_u},\varphi_{C_v}])$ in $O(1)$ time, after an $O(|\mv(\mathcal{P})|)$ preprocessing time, as follows.  Fix an endpoint $\varphi_C \in \mp$ and for every node $\varphi_{C'} \in \mp$, we compute $\omega(\mathcal{P}[\varphi_{C},\varphi_{C'}])$ in total $O(|\mv(\mp)|)$ time. Then, we can compute in $O(1)$ time:
 	\begin{center}
 	    \scalebox{0.85}{$
				\omega(\mathcal{P}[\varphi_{C_u},\varphi_{C_v}]) = \begin{cases}
					\omega(\mathcal{P}[\varphi_{C},\varphi_{C_u}]) - \omega(\mathcal{P}[\varphi_{C},\varphi_{C_v}]) + \omega(\varphi_{C_v}) &\text{if $\omega(\mathcal{P}[\varphi_{C},\varphi_{C_u}]) \geq \omega(\mathcal{P}[\varphi_{C},\varphi_{C_v}])$}\\
					\omega(\mathcal{P}[\varphi_{C},\varphi_{C_v}]) - \omega(\mathcal{P}[\varphi_{C},\varphi_{C_u}]) + \omega(\varphi_{C_u}) &\text{otherwise}
				\end{cases}$
            } 
 	\end{center}

        Given $\omega(\mathcal{P}[\varphi_{C_u},\varphi_{C_v}])$, we can check in $O(1)$ time whether $(\varphi_{C_u},\varphi_{C_v})$ is removable and if so, we remove it from $\me_i$.	The total running time to remove all removable edges is $O(|\mv_{i}| + |E^{\sigma}_i|)$. 
 \end{proof}

One important concept in our algorithm for constructing clusters at level $i$ is the corrected potential change defined below.

 \begin{definition}[Corrected Potential Change]\label{def:corrected-potential} Let $\mx$ be a subgraph of $\mg_i$. The corrected potential change of $\mx$, denoted by  $\Delta_{i+1}^+(\mx)$, is defined as: 
 \begin{eqnarray*}
\Delta_{i+1}^+(\mx) =  \Delta_{i+1}(\mx) + \sum_{\mbe \in \msttilde_i\cap \me(\mx)}w(\mbe) 
 \end{eqnarray*}
 \end{definition}
We note that  $\Delta_{i+1}(\mx)$ could be negative. One instructive example, which will appear in our construction, is when $\mx$ is a subpath of $\msttilde_i$. In this case,  $\Delta_{i+1}(\mx) = - \sum_{\mbe \in \msttilde_i\cap \me(\mx)}w(\mbe) < 0$, while $\Delta_{i+1}^+(\mx) = 0$. Indeed, we can show that $\Delta_{i+1}^+(\mx)$ is always non-negative (see Item (2) in \Cref{lm:Clustering} below). Thefore, one could view $\sum_{\mbe \in \msttilde_i\cap \me(\mx)}w(\mbe)$ as a \emph{corrective term} to $\Delta_{i+1}(\mx)$ (to make it non-negative).

 The following key lemma states all the properties of clusters constructed in our framework; the details of the construction are deferred to \Cref{sec:ClusteringDetails}. Recall that  $\mv(\mx)$ and $\me(\mx)$ are the vertex set and edge set of $\mx$, respectively.

 \begin{restatable}{lemma}{Clustering}
 \label{lm:Clustering}  Given $\mg_i$, we can construct in time $O((|\mv_i| + |\me_i|)\eps^{-1})$ (i) a partition of $\mv_i$ into three sets $\{\mv_i^{\high}, \mv_i^{\lowp},\mv_i^{\lowm}\}$ and (ii) a collection $\mathbb{X}$ of subgraphs of $\mg_i$ and their augmented diameters, such that:
 \begin{enumerate}
 	\item[(1)]   For every node $\varphi_C \in \mv_i$: If $\varphi_C \in \mv_i^{\high}$, then $\varphi_C$ is incident to $\Omega(1/\eps)$ edges in $\me_i$; otherwise ($\varphi_C \in \mv_i^{\lowp} \cup \mv_i^{\lowm}$), the number of edges in $\me_i$ incident to $\varphi_C$ is $O(1/\eps)$.
 	
 	\item[(2)] If a subgraph $\mx$ contains at least one node in $\mv^{\lowm}_{i}$, then every node of $\mx$ is in $\mv^{\lowm}_i$. Let $\mathbb{X}^{\lowm} \subseteq \mathbb{X}$ be a set of sugraphs whose nodes are in $\mv_i^{\lowm}$ only.
 	\item[(3)] $\Delta_{i+1}^+(\mx) \geq 0$ for every $\mx \in \mathbb{X}$, and
 	\begin{equation}\label{eq:averagePotential}
 		\sum_{\mx \in \mathbb{X}\setminus \mathbb{X}^{\lowm}} \Delta_{i+1}^+(\mx) = \sum_{\mx \in \mathbb{X}\setminus \mathbb{X}^{\lowm}} \Omega(|\mv(\mx)|\eps^2 L_i). 
 	\end{equation}
 	\item[(4)] There is no edge in $\me_i$ between  $\mv^{\high}_i$ and $\mv^{\lowm}_i$. Furthermore, if there exists an edge   $(\varphi_{C_u},\varphi_{C_v}) \in \me_i$ such that both $\varphi_{C_u}$ and $\varphi_{C_v}$ are in 
 	$\mv_i^{\lowm}$, then  $\mv_i^{\lowm} = \mv_i$ and $|\me_i| = O(\frac{1}{\epsilon^2})$;  that is, the partition $\{\mv_i^{\high}, \mv_i^{\lowp},\mv_i^{\lowm}\}$ of $\mv_i$ degenerates.
 	\item[(5)] For every subgraph $\mx \in \mathbb{X}$, $\mx$ satisfies the three properties (\hyperlink{P1'}{P1'})-(\hyperlink{P3'}{P3'}) with constant $g=31$ and $\eps \leq \frac{1}{8(g+1)}$,  and $|\me(\mx)\cap \me_i| = O(|\mv(\mx)|)$.
 \end{enumerate}	
 Furthermore, $\mathbb{X}$ can be constructed in the \emph{pointer-machine model} with the same running time.
 \end{restatable}

 We note the following points regarding subgraphs in $\mathbb{X}$ constructed by \Cref{lm:Clustering}.
 
\begin{remark}\label{remark:Clustering} 
\begin{enumerate}
\item It is possible for a subgraph   $\mx \in \mathbb{X}$ to contain nodes in both $\mv_i^{\high}$ and $\mv_i^{\lowp}$. 
\item  \Cref{eq:averagePotential} implies that the \emph{average} amount of corrected potential change \emph{per subgraph} $\mx \in \mathbb{X}\setminus \mathbb{X}^{\lowm}$ is $\Omega(|\mv(\mx)|\eps^2 L_i)$.  On the other hand, there is no guarantee, other than non-negativity, on the corrected potential change of $\mx$ if $\mx \in \mathbb{X}^{\lowm}$.
\end{enumerate}
\end{remark}

We make the following observations on subgraphs of $\mathbb{X}$ that satisfy all the properties stated in \Cref{lm:Clustering}.

\begin{observation}\label{obs:XhighXlowp} If a subgraph $\mx \in \mathbb{X}$ has $\mv(\mx)\cap (\mv^{\high}_{i}\cup \mv^{\lowp}_i) \not= \emptyset$, then $\mv(\mx)\subseteq (\mv^{\high}_{i}\cup \mv^{\lowp}_i)$.
\end{observation}
\begin{proof}
	Follows from Item (2)  in \Cref{lm:Clustering} and the fact that $\{\mv_i^{\high}, \mv_i^{\lowp},\mv_i^{\lowm}\}$  is a partition of $\mv_i$.
\end{proof}

\begin{observation}\label{obs:LowmStructure} Unless the partition $\{\mv_i^{\high}, \mv_i^{\lowp},\mv_i^{\lowm}\}$ degenerates, for every edge $(\varphi_{C_u},\varphi_{C_v})$ with one endpoint in $\mv_i^{\lowm}$, w.l.o.g.\ $\varphi_{C_v}$, the other endpoint $\varphi_{C_u}$ must be in $\mv_i^{\lowp}$. 
As a result, $\me(\mx)\cap \me_i = \emptyset$ if $\mx \in \mathbb{X}^{\lowm}$.  
\end{observation}
\begin{proof} If the partition $\{\mv_i^{\high}, \mv_i^{\lowp},\mv_i^{\lowm}\}$ does not degenerate, by Item (4) in \Cref{lm:Clustering}, any edge incident to a node in $\mv_i^{\lowm}$ must be incident to a node in $\mv_i^{\lowp}$. By Item (2), if $\mx \in \mathbb{X}^{\lowm}$, then $\mv(\mx) \subseteq \mv_i^{\lowm}$ and hence, there is no edge between two nodes in $\mx$. Thus, $\me(\mx)\cap \me_i = \emptyset$. 
\end{proof}

Next, we show how to construct $\msttilde_{i+1}$ for the next level.

\begin{lemma}\label{lm:MSTiPlus1} Given the collection of subgraphs $\mathbb{X}$ of $\mg_i$ and their augmented diameters, we can construct the set of nodes $\mv_{i+1}$, and their weights, and the cluster tree $\msttilde_{i+1}$ of $\mg_{i+1}$ in $O(|\mv_{i}|\alpha(m,n))$ time.
\end{lemma}
\begin{proof} For each subgraph $\mx \in \mathbb{X}$, we call \textsc{Union} operations sequentially on the set of clusters corresponding to the nodes of $\mx$ to create a level-$(i+1)$ cluster $C_{\mx} \in \mathcal{C}_{i+1}$. Then we create a set of nodes $\mv_{i+1}$ for $\mg_{i+1}$: each node $\varphi_{C_{\mx}}$ corresponds to a cluster $C_{\mx} \in \mc_{i+1}$ (and also subgraph $\mx \in \mathbb{X}$). Next, we set the weight $\omega(\varphi_{C_{\mx}}) = \adm(\mx)$. The total running time of this step is $O(|\mv_i|\alpha(m,n))$.
	
	We now construct $\msttilde_{i+1}$. Let $\msttilde^{out}_{i} = \msttilde_{i}\setminus (\cup_{\mx \in \mathbb{X}}(\me(\mx)\cap \msttilde_{i}))$ be the set of $\msttilde_{i}$ edges that are not contained in any subgraph $\mx \in \mathbb{X}$. Let $\msttilde_{i+1}'$ be the graph with vertex set $\mv_{i+1}$ and there is an edge between two nodes $(\mx,\my)$ in $\mv_{i+1}$ of there is at least one edge in $\msttilde^{out}_{i}$ between two nodes in the two corresponding subgraphs $\mx$ and $\my$; $\msttilde_{i+1}'$ can be constructed in time $O(|\mv_{i}|)$. Note that $\msttilde_{i+1}'$ could have parallel edges (but no self-loop).  Since $\msttilde_{i}$ is  a spanning tree of $\mg_i$, $\msttilde_{i+1}'$ must be connected. $\msttilde_{i+1}$ is then a spanning tree of $\msttilde_{i+1}'$, which can be constructed in time $O(|\mv_{i}|)$ since  $\msttilde_{i+1}'$ has at most  $|\mv_i|$ edges. The lemma now follows. 
\end{proof}

\subsection{Constructing \texorpdfstring{$H_i$: Proof of \Cref{lm:framework}}{Hi: Proof of Framework Lemma}} \label{subsec:ConstructHi}

Recall that to obtain a fast algorithm for constructing a light spanner, \Cref{lm:framework-technical} requires a fast construction of clusters at every level and a fast construction of $H_i$, the spanner for level-$i$ edges $E^{\sigma}_i$.  In \Cref{subsec:LeveIplus1Construction}, we have designed an efficient construction of level-$i$ clusters (\Cref{lm:MSTiPlus1}). In this section, we show how to construct $H_i$ efficiently with stretch $t(1+\max\{s_{\ssa}(2g) + 4g, 10g\}\eps)$; that is parameter $\rho$ in \Cref{lm:framework-technical} is $\rho = \max\{s_{\ssa}(2g) + 4g, 10g\}$. By induction, we assume that the stretch of every edge of weight less than $L_i/(1+\psi)$ in $H_{<L_i/(1+\psi)}$ is $t(1+\max\{s_{\ssa}(2g) + 4g, 10g\}\eps)$.  Note that $H_{< L_i} = H_{<L_i/(1+\psi)} \cup H_i$; see \Cref{remark:stretchLevel-i}. 

Our construction of $H_i$ assumes the existence of \hyperlink{SPHigh}{$\ssa$}. Since edges of the input graph to \hyperlink{SPHigh}{$\ssa$} must have weights in $[L,(1+\eps)L)$ for some parameter $L$, we set parameter $\psi$ in \Cref{lm:framework-technical} to be $\eps$. Thus, level-$i$ edges $E^\sigma_{i}$ (and hence edges in $\me_i$ of $\mg_i$) have weights in $[L_i/(1+\eps),L_i)$.

We  now go into the details of the construction of $H_i$. We assume that we are given the collection $\mathbb{X}$ of subgraphs as described in \Cref{lm:Clustering}. Define:
\begin{equation}\label{eq:ACT-XLowHighdef}
	\begin{split}
		\mathbb{X}^{\high} &= \{\mx \in \mathbb{X}: \mv(\mx)\cap \mv^{\high}_{i} \not=\emptyset\}\\
		\mathbb{X}^{\lowp} &= \{\mx \in \mathbb{X}: \mv(\mx)\cap \mv^{\lowp}_{i} \not=\emptyset\}\\
	\end{split}
\end{equation}
It could be that $\mathbb{X}^{\high}\cap \mathbb{X}^{\lowp}\not= \emptyset$. By \Cref{obs:XhighXlowp}, $\{\mathbb{X}^{\high}\cup \mathbb{X}^{\lowp},\mathbb{X}^{\lowm}\}$ is a partition of $\mathbb{X}$.

\paragraph{Construction overview} Given a set of subgraphs $\mathbb{X}$ satisfying the properties stated in \Cref{lm:Clustering}, our general approach to construct $H_i$ is as follows. First, we add to $H_i$ (the corresponding edge of) every edge $\mbe$ contained in some subgraph $\mx$: $\mbe \in \me(\mx)\cap \me_i$. Edges added to $H_i$ in this step are incident to nodes in $\mv_i^{\lowp}\cup \mv_i^{\high}$. By Item (5)  of \Cref{lm:Clustering}, we only add $O(|\mv(\mx)|)$ edges per subgraph $\mx$, and hence, we can bound the total weight of these edges by ($O(\frac{1}{\eps^2})$ times) the corrected potential changes of subgraphs in $\mathbb{X}\setminus \mathbb{X}^{\lowm}$, due to Item (3) of \Cref{lm:Clustering}.  Next, we add to $H_i$ all edges incident to all nodes in $\mv_i^{\lowp} \cup \mv_i^{\lowm}$. Unless we are in the degenerate case, edges added to $H_i$ in the second step are incident to nodes in $\mv_i^{\lowp}$ (see \Cref{obs:LowmStructure}), and hence, their total weight can be bounded by  ($O(\frac{1}{\eps^3})$ times) the corrected potential changes of subgraphs in  $\mathbb{X}\setminus \mathbb{X}^{\lowm}$; to this end we apply both Item (3) of \Cref{lm:Clustering} and the fact that any node in $\mv_i^{\lowp}$ has at most $O(1/\eps)$ incident edges in $\me_i$. Now we are left with edges whose both endpoints are in $\mv_i^{\high}$, denoted by $\me^{\high}_i$. In the third step, we select a subset of (the corresponding edges of) these edges to add to $H_i$ by using \hyperlink{SPHigh}{$\ssa$}. The pseudocode is given in \Cref{fig:construct-Hi}.

Recall that each edge $(\varphi_{C_u},\varphi_{C_v}) \in \me_i$ has a corresponding edge $(u,v)\in E^{\sigma}_i$ where $u$ and $v$ are in two level-$i$ clusters $C_u$ and $C_v$, respectively.  Our goal in this section is to prove the following lemma.

\begin{restatable}{lemma}{HiConstruction}
	\label{lm:ConstructH_i} Given \hyperlink{SPHigh}{$\ssa$}, we can construct  $H_i$ in total time $O((|\mv_i| + |\me_i|)\tau(m,n))$	satisfying \Cref{lm:framework-technical}  with  $\lambda = O(\chi \eps^{-2} + \epsilon^{-3})$, and $A = O(\chi\eps^{-2} + \eps^{-3})$, when $\eps \leq 1/(2g)$. Furthermore, the stretch of every edge in $E_i^{\sigma}$ in $H_{<L_i}$ is $t(1+\max\{s_{\ssa}(2g) + 4g, 10g\}\eps)$. 
\end{restatable}

We apply $\ssa$ to $\mv^{\high}_i$ that has size at most $n$ as every level-$i$ cluster corresponding to a node in  $\mv^{\high}_i$ contains at least one original vertex in $G$. Furthermore, $|\me^{\high}_i|$ is bounded by $m$ and hence, $\tau(|\me^{\high}_i|, |\mv^{\high}_i|)\leq \tau(m,n)$.

\begin{remark}\label{remark:ACT} If $\ssa$ can be implemented in the ACT model in time $O((|\mv^{\high}_i| + |\me_i^{\high}|)\tau(m,n))$, then the construction of $H_i$ can be implemented in the ACT model  in time $O((|\mv_i| + |\me_i|)\tau(m,n))$.
\end{remark}

\begin{figure}[htb!]
\begin{tcolorbox}	
\paragraph{Constructing $H_i$} We construct $H_i$ in three steps, as briefly described in the construction overview above. Initially, $H_i$ contains no edges.

\begin{itemize}
	\item \textbf{(Step 1).~} For every sugraph $\mx \in\mathbb{X}$ and every edge $\mbe = (\varphi_{C_u},\varphi_{C_v}) \in \me(\mx)$ such that $\mbe \in \me_i$, we add the corresponding edge $(u,v)$ to $H_i$. (Note that if $\mbe \not\in \me_i$, it is in $\msttilde_i$ and hence $(u,v)$ belongs to $H_0$).
	
	\item \textbf{(Step 2).~} For each node $\varphi_{C_u} \in \mv_i^{\lowp} \cup \mv_{i}^{\lowm}$, and for each edge $(\varphi_{C_u},\varphi_{C_v})$ in $\me_i$ incident to $\varphi_{C_u}$, we add the corresponding edge $(u,v)$ to $H_i$,
	
	\item \textbf{(Step 3).~} Let $\me_i^{\high}\subseteq \me_i$ be the set of edges whose both endpoints are in $\mv_i^{\high}$, and $\mathcal{K}_i = (\mv^{\high}_i, \me_i^{\high},\omega)$ be a subgraph of $\mg_i$. We run \hyperlink{SPHigh}{$\ssa$} on $\mathcal{K}$ to obtain $\me^{\prune}_i$. For every edge $(\varphi_{C_u},\varphi_{C_v})\in \me_i^{\prune}$, we add the corresponding edge $(u,v)$ to $H_i$.
\end{itemize}
\end{tcolorbox}
  \caption{The algorithm for constructing $H_i$.}
    \label{fig:construct-Hi}
\end{figure}

\paragraph{Analysis} In \Cref{clm:Hi-Time}, \Cref{clm:Hi-Stretch}, and \Cref{clm:Hi-Weight} below,  we bound the running time to construct $H_i$, the stretch of edges in $E^{\sigma}_i$, and the weight of $H_i$, respectively.

\begin{claim}\label{clm:Hi-Time} $H_i$ can be constructed in time $O((|\mv_i| + |\me_i|)\tau(m,n))$.
\end{claim}
\begin{proof}
	We observe that Steps 1 and 2 can be straightforwardly implemented in $O(|\mv_i| + |\me_i|)$ time. Here we assume that we have a constant time translation from the cluster graph edges to the original edges. This can be done by  storing for each edge $(\varphi_{C_u},\varphi_{C_v})\in \me_i$ a pointer to the original edge $(u,v)$ when $(\varphi_{C_u},\varphi_{C_v})$ was created.  The running time of Step 3 is dominated by the running time of \hyperlink{SPHigh}{$\ssa$}. Note that we assume that \hyperlink{SPHigh}{$\ssa$} has access to a function $\source(\cdot)$ that maps each node $\varphi_C \in\mv_i^{\high}$  to a representative of $C$ and each edge $(\varphi_{C_u},\varphi_{C_v})\in \me^{\high}$ to the corresponding edge $(u,v) \in E^{\sigma}_i$.  We can construct function $\source(.)$ by simply storing the pointer to the corresponding vertex in $C$ or the pointer to the corresponding edge. Thus, the running time of Step 3 is $O((|\mv_i| + |\me_i|)\tau(m,n))$. 	This implies the claimed running time.	 
	 
\end{proof}

Next, we bound the stretch of edges in $E^\sigma_{i}$. We first observe that the input to \hyperlink{SPHigh}{$\ssa$} satisfies its requirement.

\begin{claim}\label{clm:Ki-clustergraph}$\mathcal{K}_i = (\mv^{\high}_i, \me_i^{\high},\omega)$ is a $(L, \eps, \beta)$-cluster graph with $L = L_i/(1+\eps)$, $\beta = 2g$,  and  $H_{< L} = H_{< L_i/(1+\eps)}$, where $H_{< L_i/(1+\eps)}$ is the spanner constructed for edges of weight less than $L_i/(1+\eps)$ (see \Cref{remark:stretchLevel-i} with $\psi = \eps$). Furthermore, the stretch of $H_{< L}$ for edges of weight less than $L$ is $t(1+ \max\{s_{\ssa}(2g)+4g,10g\}\eps)$.  
\end{claim} 
\begin{proof}
We verify all properties in \Cref{def:ClusterGraph-Param}. Properties (1) and (2) follow directly from the definition of $\mathcal{K}_i$. Since we set $\psi = \eps$, every edge $(u,v) \in E^{\sigma}_i$  has $L_i/(1+\eps) \leq w(u,v)\leq L_i$. Since $L = L_i/(1+\eps)$, we have that $L \leq w(u,v) \leq (1+\eps)L$; this implies property (3). By property \hyperlink{P3}{(P3)}, we have $\dm(H_{<L_i/(1+\eps)}[C]) \leq gL_{i-1} = g(1+\eps) \eps L \leq 2g \eps L = \beta \eps L$ when $\eps < 1$. Thus, $\mk_i$ is a   $(L, \eps, \beta)$-cluster graph. By induction, the stretch of $H_{< L}$ is $t(1+ \max\{s_{\ssa}(2g)+4g,10g\}\eps)$.   
 
\end{proof}

\begin{claim}\label{clm:Hi-Stretch} For every edge $(u,v) \in E^\sigma_{i}$, $d_{H_{< L_i}}(u,v) \leq t(1+\max\{s_{\ssa}(2g) + 4g, 10g\}\eps)w(u,v)$ when $\eps \leq 1/(2g)$.
\end{claim}
\begin{proof}
	Let $F^{\sigma}_i = \{(u,v) \in E^{\sigma}_i: \exists (\varphi_{C_u},\varphi_{C_v}) \in \me_i\}$ be the set of edges in $E^{\sigma}_i$ that correspond to the edges in $\me_i$. We first show that:
	\begin{equation}\label{eq:stertchF}
		d_{H_{< L_i}}(u,v) \leq t(1+ s_{\ssa}(2g)\eps)w(u,v) \qquad \forall (u,v) \in F^{\sigma}_i.
	\end{equation}
	
	To that end, let $(\varphi_{C_u},\varphi_{C_v}) \in \me_i$ be the edge corresponding to $(u,v)$ where $(u,v) \in F^{\sigma}_i$. If at least one of the endpoints of $(\varphi_{C_u},\varphi_{C_v})$ is in $\mv^{\lowp}_i \cup \mv^{\lowm}_i$, then $(u,v) \in H_i$ by the construction in Step 2, hence \Cref{eq:stertchF} holds. Otherwise, $\{\varphi_{C_u},\varphi_{C_v}\}\subseteq \mv_i^{\high}$, which implies that $(\varphi_{C_u},\varphi_{C_v}) \in \me^{\high}_i$. Since we add all edges of $\me_i^{\prune}$ to $H_i$, by property (2) of \hyperlink{SPHigh}{$\ssa$} and \Cref{clm:Ki-clustergraph}, the stretch of $(u,v)$ is $t(1+s_{\ssa}(2g)\eps)$. 
	
	It remains to bound the stretch of any edge $(u',v') \in E^{\sigma}_i\setminus F^{\sigma}_i$. Recall that  $(u',v')$ is not added to $\me_i$ because (a) both $u'$ and $v'$ are in the same level-$i$ cluster in the construction of the cluster graph in \Cref{lm:G_i-construction} , or (b) $(u',v')$ is parallel with another edge $(u,v)$ also in \Cref{lm:G_i-construction}, or (c) the edge $(\varphi_{C_{u'}},\varphi_{C_{v'}})$ corresponding to $(u',v')$ is a removable edge (see \Cref{def:GiProp}).

	In case (a), since the level-$i$ cluster containing both $u'$  and $v'$ has diameter at most $gL_{i-1}$ by property (\hyperlink{P3}{P3}), we have a path from $u'$ to $v'$ in $H_{< L_{i-1}}$ of diameter at most $gL_{i-1} ~=~ g\eps L_i ~\leq \frac{L_i}{1+\psi}~\leq~w(u',v')$ when $\eps \leq  1/(2g)$. Thus, the stretch of edge $(u',v')$ is $1$. For case (c), the stretch of $(u',v')$ in $H_{< L_{i-1}}$ is $t(1+6g\eps)$ since $\eps \leq 1$.  Thus, in both cases, we have:
		\begin{equation}\label{eq:stertch-uvprime}
		d_{H_{< L_i}}(u',v') \leq t(1+ 6g\eps)w(u',v') 
	\end{equation}

	We now consider case (b). 	Let $C_u$ and $C_v$ be two level-$i$ clusters containing $u$ and $v$, respectively. W.l.o.g, we assume that $u' \in C_u$  and $v' \in C_v$. Since we only keep an edge of minimum weight among all parallel edges, $w(u,v) \leq w(u',v')$. Since the level-$i$ clusters that contain $u$  and $v$ have diameters at most $gL_{i-1} = g\eps L_i$ by property (\hyperlink{P3}{P3}), it follows that $\dm(H_{< L_i}[C_u]),\dm(H_{< L_i}[C_v]) \le g\epsi L_i$. 	We have:
	\begin{equation*}
		\begin{split}
			d_{H_{< L_i}}(u',v') &\leq 	d_{H_{< L_i}}(u,v) + \dm(H_{< L_i}[C_u]) + \dm(H_{< L_i}[C_v])\\ &\leq t(1+ \max\{s_{\ssa}(2g),6g\}\eps)w(u,v) +   2g\epsi L_i \\
			&\leq  t(1+ \max\{s_{\ssa}(2g),6g\}\eps)w(u',v')  +  2g\epsi L_i\\
			&\leq t(1+ \max\{s_{\ssa}(2g),6g\}\eps)w(u',v') + 4g\eps w(u',v') \\
			&=  t(1+ \max\{s_{\ssa}(2g)+4g,10g\}\eps)w(u',v') \qquad \mbox{(since $t\geq 1$)}.
		\end{split}
	\end{equation*}
The second inequality is due to \Cref{eq:stertchF} and \Cref{eq:stertch-uvprime}, and the forth inequality is due to  $w(u',v') \geq L_i/(1+\eps) \geq L_i/2$). 
\end{proof}

\begin{claim}\label{clm:Hi-Weight} Let $\msttilde^{in}_i = \cup_{\mx \in \mathbb{X}}(\me(\mx)\cap \msttilde_i)$ be the set of $\msttilde_i$ edges that are contained in subgraphs in $\mathbb{X}$. Then, $w(H_i) \leq \lambda \Delta_{i+1} + a_i$ for $\lambda = O(\chi \eps^{-2} + \eps^{-3})$  and $a_i = (\chi\eps^{-2} + \eps^{-3}) \cdot w(\msttilde^{in}_i) + O(L_i/\eps^2)$. 
\end{claim}
\begin{proof} Let $\msttilde^{in}_i(\mx) = \me(\mx)\cap \msttilde_i$ for each subgraph $\mx \in \mathbb{X}$.  By the definition of $ \mathbb{X}^{\lowp}$ and $ \mathbb{X}^{\high}$ (see~\Cref{eq:ACT-XLowHighdef}), it holds that:
	\begin{equation} \label{eq:mvilowhigh}
		\begin{split}
			|\mv_{i}^{\high}| \leq \sum_{\mx \in \mathbb{X}^{\high}}|\mv(\mx)|  \quad &\mbox{and} \quad	|\mv_{i}^{\lowp}| \leq \sum_{\mx \in \mathbb{X}^{\lowp}}|\mv(\mx)| 
		\end{split}
	\end{equation}
	
	First, we consider the non-degenerate case where $\mv_i^{\lowm} \not = \mv_i$. By \Cref{obs:LowmStructure}, any edge in $\me_i$ incident to a node in $\mv_i^{\lowm}$ is also incident to a node in $\mv_i^{\lowp}$. We bound the total weight of the edges added to $H_i$ by considering each step in the construction of $H_i$ separately. Let $F^{(a)}_i\subseteq E^{\sigma}_i$ be the set of edges added to $H_i$ in the construction in Step $a$, $a\in \{1,2,3\}$.
	
	By \Cref{obs:LowmStructure}, $\me(\mx)\cap \me_i = \emptyset$ if $\mx \in \mathbb{X}^{\lowm}$. Recall that $\mathbb{X}^{\high}\cup \mathbb{X}^{\lowp} = \mathbb{X}\setminus \mathbb{X}^{\lowm}$. By Item (5) in \Cref{lm:Clustering}, the total weight of the edges added to $H_i$ in Step 1 is:
	\begin{equation}\label{eq:Fi1}
		\begin{split}
			w(F^{(1)}_i)  &=  \sum_{\mx \in \mathbb{X}^{\high} \cup \mathbb{X}^{\lowp}} O(|\mv(\mx)|) L_i \stackrel{{\scriptstyle {\mbox{Eq.~}}(\ref{eq:averagePotential})}}{=}  O(\frac{1}{\eps^2})\sum_{\mx \in \mathbb{X}^{\high} \cup \mathbb{X}^{\lowp}} \Delta^+_{i+1}(\mx)\\
			&= O(\frac{1}{\eps^2})\sum_{\mx \in\mathbb{X}} \Delta^+_{i+1}(\mx)  \qquad \mbox{(since $\Delta^+_{i+1}(\mx)\geq 0$  by \Cref{lm:Clustering})} \\
			&=  O(\frac{1}{\eps^2})\sum_{\mx \in \mathbb{X}} \left(\Delta_{i+1}(\mx) + w(\msttilde^{in}_i(\mx))\right)\\
			&= O(\frac{1}{\eps^2})(\Delta_{i+1} + w(\msttilde^{in}_i)) \qquad \mbox{(by \Cref{clm:localPotenDecomps})}~.
		\end{split}
	\end{equation}   
	
	Next, we bound $w(F^{(2)}_i)$. Let $(u,v)$ be an edge added to $H_i$ in Step 2  and let $(\varphi_{C_u},\varphi_{C_v})$ be the corresponding edge of $(u,v)$. Since $\mv^{\lowm}_i \not= \mv_i$, at least one of the endpoints of $(\varphi_{C_u},\varphi_{C_v})$, w.l.o.g.\ $\varphi_{C_u}$,  	is in $\mv^{\lowp}_{i}$ by \Cref{obs:LowmStructure}. Recall by Item (1) of \Cref{lm:Clustering} that all nodes in $\mv_i^{\lowp}$ have low degree, i.e., incident to  $O(1/\eps)$ edges in $\me_i$.	Thus, $|F^{(2)}_i| = O(\frac{1}{\eps}) |\mv_i^{\lowp}|$. We have: \begin{equation}\label{eq:Fi2}
		\begin{split}
			w(F^{(2)}_i)  &=  O(\frac{1}{\eps}) |\mv_i^{\lowp}|  L_i \stackrel{{\scriptstyle \mbox{Eq.~}(\ref{eq:mvilowhigh})}}{=} O(\frac{1}{\eps})\sum_{\mx \in \mathbb{X}^{\lowp}}|\mv(\mx)| L_i \\&=~ O(\frac{1}{\eps})\sum_{\mx \in \mathbb{X}^{\high}\cup \mathbb{X}^{\lowp}}|\mv(\mx)| L_i\\
			&\stackrel{{\scriptstyle {\mbox{Eq.~}}(\ref{eq:averagePotential})}}{=}  O(\frac{1}{\eps^3})\sum_{\mx \in \mathbb{X}^{\high}\cup \mathbb{X}^{\lowp}} \Delta^+_{i+1}(\mx) = O(\frac{1}{\eps^3})\sum_{\mx \in\mathbb{X}} \Delta^+_{i+1}(\mx)  \\
			& =    O(\frac{1}{\eps^3})\sum_{\mx \in \mathbb{X}} \left(\Delta_{i+1}(\mx) + w(\msttilde^{in}_i(\mx))\right)\\
			&= O(\frac{1}{\eps^3})(\Delta_{i+1} + w(\msttilde^{in}_i)) \qquad \mbox{(by \Cref{clm:localPotenDecomps})}~.
		\end{split}
	\end{equation}  
	
	By property (1) of \hyperlink{SPHigh}{$\ssa$}, the number of edges added to $H_i$ in Step 3 is at most $\chi|\mv^{\high}_i|$. Thus: 
		\begin{equation}\label{eq:Fi3}
		\begin{split}
			w(F^{(3)}_i)  &~\leq~  \chi |\mv_i^{\high}|  L_i \stackrel{{\scriptstyle \mbox{Eq.~}(\ref{eq:mvilowhigh})}}{\leq}  \chi\sum_{\mx \in \mathbb{X}^{\high}}|\mv(\mx)| L_i ~\leq~ \chi\sum_{\mx \in \mathbb{X}^{\high}\cup \mathbb{X}^{\lowp}}|\mv(\mx)| L_i\\
			&\stackrel{{\scriptstyle {\mbox{Eq.~}}(\ref{eq:averagePotential})}}{=}  O(\chi \eps^{-2})\sum_{\mx \in \mathbb{X}^{\high}\cup \mathbb{X}^{\lowp}} \Delta^+_{i+1}(\mx)  = O(\chi \eps^{-2})\sum_{\mx \in\mathbb{X}} \Delta^+_{i+1}(\mx)\\   
			& =     O(\chi \eps^{-2})\sum_{\mx \in \mathbb{X}} \left(\Delta_{i+1}(\mx) + w(\msttilde^{in}_i(\mx))\right)\\
			&= O(\chi \eps^{-2})(\Delta_{i+1} + w(\msttilde^{in}_i)) \qquad \mbox{(by \Cref{clm:localPotenDecomps})}~.
		\end{split}
	\end{equation} 
	
	By \Cref{eq:Fi1,eq:Fi2,eq:Fi3}, we conclude that:
	\begin{equation}\label{eq:Hi-nondegen}
		\begin{split}
			w(H_i) &=  O( \chi \eps^{-2}  + \eps^{-3}) (\Delta_{i+1} + w(\msttilde^{in}_i)) \leq \lambda(\Delta_{i+1} + w(\msttilde^{in}_i))
		\end{split}
	\end{equation} 
	for some $\lambda =  O(\chi \eps^{-2} + \eps^{-3})$.

	It remains to consider the degenerate case where $\mv_i^{\lowm} = \mv_i$. Even if we add every edge that corresponds to an edge in $\me_i$ to $H_i$, Item (3) in \Cref{lm:Clustering} implies that the number of such edges is at most $O(\frac{1}{\eps^2})$. Thus, we have:
	\begin{equation}\label{eq:Hi-degen}
		w(H_i) = O(\frac{L_i}{\eps^2}) \leq  \lambda\cdot (\Delta_{i+1} + w(\msttilde^{in}_i)) + O(\frac{L_i}{\eps^2}) 
	\end{equation} 
	where in the last equation, we use the fact that:
	\begin{equation*}
		\Delta_{i+1} + w(\msttilde^{in}_i) \stackrel{\text{\footnotesize{\Cref{clm:localPotenDecomps}}}}{=} \sum_{\mx \in \mathbb{X}} (\Delta_{i+1}(\mx) + \msttilde^{in}_i(\mx))  = \sum_{\mx \in \mathbb{X}} \Delta^+_{i+1}(\mx)\geq 0~
	\end{equation*}
	by Item (3) of \Cref{lm:Clustering}. Thus, the claim follows from \Cref{eq:Hi-degen,eq:Hi-nondegen}. %
\end{proof}

\begin{proof}[Proof of \Cref{lm:ConstructH_i}] The running time follows from \Cref{clm:Hi-Time}.  By \Cref{clm:Hi-Stretch}, the stretch is $t(1+ \max\{s_{\ssa}(2g) + 4g, 10g\}\eps)$. By \Cref{clm:Hi-Weight}, we have
	$\sum_{i \in \mathbb{N}^+}a_i = \sum_{i\in \mathbb{N}^+}(\lambda  \msttilde^{in}_i + O(L_i/\eps^2))$. Observe by the definition that the sets of corresponding edges of $\msttilde^{in}_i$ and $\msttilde^{in}_j$ are disjoint for any $i\not=j \geq 1$. Thus, $\sum_{i\in \mathbb{N}^+} \msttilde^{in}_i\leq w(\mst)$. Observe that:
	
	\begin{equation*}
		\sum_{i\in \mathbb{N}^+}O(\frac{L_i}{\eps^2}) ~=~  O(\frac{1}{\epsilon^2}) \sum_{i=1}^{i_{\max}} \frac{L_{i_{\max}}}{\epsilon^{i_{\max}-i}} ~=~ O(\frac{L_{i_{\max}}}{\epsilon^2(1-\epsilon)}) ~=~ O(\frac{1}{\epsilon^2}) w(\mst)~;
	\end{equation*}
	here $i_{\max}$ is the maximum level. The last equation is due to that $\eps \leq 1/2$  and every edge has weight at most $w(\mst)$ (by the removal step in the construction of $\tilde{G}$). Thus, $A = \lambda + O(\eps^{-2}) = O(\chi \eps^{-2} + \eps^{-3}) +  O(\eps^{-2}) =  O(\chi \eps^{-2} + \eps^{-3})$ as claimed. %
\end{proof}

We are now ready to prove \Cref{lm:framework}.

\begin{proof}[Proof of \Cref{lm:framework}]
By \Cref{lm:Clustering} and \Cref{lm:MSTiPlus1}, level-$(i+1)$ clusters can be constructed in time $O((|\mv_i| + |\me_i|)\eps^{-1} + |\mv_i|\alpha(m,n)) = O((|\mathcal{C}_i| + |E^{\sigma}_i|)(\alpha(m,n) + \eps^{-1})$ when $\eps \ll 1$. By \Cref{lm:ConstructH_i}, $H_i$ can be constructed in time $O(|\mv_i| + |\me_i|)\tau(m,n)= O((|\mathcal{C}_i| + |E^{\sigma}_i|)\tau(m,n))$.

We can construct a minimum spanning tree in time  $T_{\mst} = O((n+m)\alpha(m,n))$ by using Chazelle's algorithm~\cite{Chazelle00}. Thus, by \Cref{lm:framework-technical}, the construction time of the light spanner is 
\begin{align*}
    &O(m\eps^{-1}(\tau(m,n) + \alpha(m,n) + \eps^{-1})\log(1/\eps) +  T_{\mst}) \\&= O(m\eps^{-1}(\tau(m,n) + \alpha(m,n) + \eps^{-1})\log(1/\eps)~.
\end{align*}

By \Cref{lm:ConstructH_i} and \Cref{lm:framework-technical}, the lightness of the spanner is $$O(\frac{\lambda + A + 1}{\eps}\log \frac{1}{\epsilon} + \frac{1}{\eps}) = O((\chi \eps^{-3} + \eps^{-4})\log(1/\eps)).$$
Note that we set $\psi = \eps$ in this case. Since $g = 31$, by \Cref{lm:ConstructH_i} and \Cref{lm:framework-technical},  the stretch of the spanner is $$t(1+\max\{s_{\ssa}(2g) + 4g, 10g\}\eps) \leq t(1+(s_{\ssa}O(1)) + O(1))\eps)~.$$ This completes the proof of the theorem. 
\end{proof}

\begin{table}[!ht]
\caption{Notation introduced in \Cref{sec:framework}.}
	\label{table:notation}
\renewcommand{\arraystretch}{1.3}
\begin{tabular}{|l|l|}
  	\hline
	\textbf{Notation} & \textbf{Meaning} \\ \hline
	$E^{light}$ &$ \{e \in E(G) : w(e)\le w/\varepsilon\}$\\ \hline 
	$E^{heavy}$ & $E \setminus E^{light}$ \\\hline
	$E^{\sigma} $ & $\bigcup_{i \in \mathbb{N}^{+}}E_{i}^{\sigma}$\\\hline
	$E_{i}^{\sigma} $ & $\{e \in E(G) : \frac{L_i}{1+\psi} \leq w(e) < L_i\}$\\\hline
	$g$ & constant in \hyperlink{P3}{property (P3)}, $g = 31$ \\\hline
	$\mathcal{G}_i = (V_i, \msttilde_{i} \cup \mathcal{E}_i, \omega)$ & cluster graph; see \Cref{def:ClusterGraphNew}. \\\hline
	$\me_i$ & corresponds to a subset of edges of $E^{\sigma}_i$\\\hline
	$\mathbb{X}$ & a collection of subgraphs of $\mathcal{G}_i$\\\hline
	$\mx, \mv(\mx), \me(\mx)$ & a subgraph in $\mathbb{X}$, its vertex set, and its edge set\\\hline
	$\Phi_i$ & $\sum_{c \in C_i}\Phi(c)$ \\\hline
	$\Delta_{i+1} $&$ \Phi_i - \Phi_{i+1}$\\\hline
	$\Delta_{i+1}(\mx)$ & $(\sum_{\phi_C\in \mx }\Phi(C) ) - \Phi(C_{\mx})$\\\hline
	$\Delta_{i+1}^+(\mx)$ & $\Delta_{i+1}(\mx) + \bigcup_{e \in \me(\mx)\cap \msttilde_i}w(e)$\\\hline
	$C_\mx$ & $\bigcup_{\phi_C \in \mx}C$ \\\hline
	$\{\mv^{\high}_i,\mv^{\lowp},\mv^{\lowm}_{i}\}$ & a partition of $\mv_i$ in \Cref{lm:Clustering}\\\hline
	$\mathbb{X}^{\lowm}$ & $\mx \in \mathbb{X}^{\lowm}$ has $\mv(\mx)\subseteq \mv^{\lowm}_i$ \\\hline
	$s_{\ssa}$ & the stretch constant of \hyperlink{SPHigh}{$\ssa$}\\\hline
\end{tabular}
\renewcommand{\arraystretch}{1}
\end{table}

\section{Light Spanners for Minor-free Graphs in Linear Time}\label{sec:minor-linear}

In \Cref{sec:applications}, we showed a construction of a light spanner for $K_r$-minor-free graphs with running time $O(nr\sqrt{r}\alpha(nr\sqrt{r},n))$. The extra factor $\alpha(nr\sqrt{r},n)$ is due to  \textsc{Union-Find} data structure in the proof of \Cref{lm:framework}. To remove this factor, we do not use  \textsc{Union-Find}. Instead, we follow the idea of Mare{\v{s}}~\cite{Mare04} that was applied to construct a minimum spanning tree for $K_r$-minor-free graphs. Specifically, after the construction of level-$(i+1)$ clusters, we prune the set of edges that are involved in the construction of levels at least $i+1$, which is $\cup_{j\geq i+1}E^{\sigma}_{j}$, as follows.

\paragraph{The algorithm} Let $E^{\sigma}_{\geq i} = \cup_{j\geq i}E^{\sigma}_{j}$. We inductively maintain a set of edges $\me_{\geq i}$, where each edge in $\me_{\geq i}$ is associated with an edge in $E^{\sigma}_{\geq i}$. (Note that only those in $\me_{i}$ are involved in the construction of spanner at level $i$.) Furthermore, we inductively guarantee that:
\begin{quote}
    \textbf{Size invariant:~} $|\me_{\geq i}| = O(r\sqrt{\log r})|\mv_i|$.
\end{quote}
Upon completing the construction of level-$(i+1)$ clusters, we construct the set of nodes $\mv_{i+1}$.  We now consider the set of edges $\me'_{\geq i+1} = \me_{\geq i}\setminus \me$. Let $\tilde{\me}_{\geq i+1}$ be  obtained from  $\me'_{\geq i+1}$ by removing parallel edges: two edges $(\varphi_{1},\varphi_{2})$ and  $(\varphi'_{1},\varphi'_{2})$ are \emph{parallel} if there exist two subgraphs $\mx,\my \in \mathbb{X}$ such that,~w.l.o.g, $\varphi_1, \varphi_1' \in \mv(\mx)$ and  $\varphi_2, \varphi_2' \in \mv(\my)$. (Among all parallel edges, we keep an  edge with minimum weight in $\tilde{\me}_{i+1}$.)  We construct the edge set  $\me_{\geq i+1}$ (between vertices in $\mv_{i+1}$) at level $(i+1)$ from $\tilde{\me}_{\geq i+1}$ by creating one edge $(\mx,\my)\in \me_{\geq i+1}$  for each associated edge $(\varphi_x,\varphi_y) \in \tilde{\me}_{\geq i+1}$ where $\varphi_x \in \mv(\mx)$ and $\varphi_y \in \mv(\my)$; $\omega(\mx,\my) = \omega(\varphi_x,\varphi_y)$.

\paragraph{Analysis} Observe that $\me_{i+1}$ corresponds to a subset of edges of  $E^{\sigma}_{\geq i+1}$ since $\me'_{\geq i+1}$, by definition, corresponds to a subset of edges of  $E^{\sigma}_{\geq i+1}$.  The stretch is in check (at most $(1+O(\eps))$), since we only remove parallel edges and since level-$(i+1)$ clusters have diameter  $O(\eps)$ times the weight of level-$(i+1)$ edges by \hyperlink{P3}{property (P3)}.  Furthermore,  since $\me_{\geq i} = O(r\sqrt{\log r} |\mv_i|)$  by the size invariant, $\me_{i+1}$ can be constructed in $O(|\mv_i|)$ time. Since the graph $(\mv_{i+1}, \me_{\geq i+1})$ is a minor of $G$ and hence, is $K_r$-minor-free, we conclude that $| \me_{\geq i+1}| = O(r\sqrt{\log r})|\mv_{i+1}|$ by \Cref{lm:minor-sparsity}, which implies the size invariant for level $i+1$. 

By the size invariant, we do not need \textsc{Union-Find} data structure, as $\me_{\geq i}$ now has  $O(r\sqrt{\log r}|\mv_i|) = O(r\sqrt{\log r}|\mc_i|)$ edges. Thus, the running time to construct $\mg_i$ in \Cref{lm:G_i-construction} becomes $O_{\eps}(|\mc_i| + |\me_{i}|) = O_{\eps}(r\sqrt{\log r}|\mc_i|)$, and the running time to  construct $\msttilde_{i+1}$  in \Cref{lm:MSTiPlus1} also becomes $O(r\sqrt{\log r}|\mc_i|)$.

 We are now ready to prove \Cref{thm:minor-free-fast} for minor-free graphs; we rely on \Cref{lm:framework-technical}. 

\begin{proof}[Proof of \Cref{thm:minor-free-fast}] Note that $t = 1+\eps$ in this case. By \Cref{lm:ConstructH_i} and \Cref{lm:App-Minor}, the stretch of $H_{<L_i}$ is $(1+\eps)(1+\max\{4g,10g\}\eps) = 1 + O(\eps) $. We can get back stretch $1+\eps$ by scaling $\eps$. 
	
	By \Cref{lm:App-Minor}, $\chi = O(r\sqrt{\log r})$ where $\chi$ is the parameter defined in \hyperlink{SPHigh}{Algorithm $\ssa$}. Thus, by \Cref{lm:ConstructH_i}, $\lambda = O(r\sqrt{\log r}\eps^{-2} + \eps^{-3})$, and $A = O(r\sqrt{\log r}\eps^{-2} + \eps^{-3})$. Thus, the lightness of the spanner is $O((r\sqrt{\log r}\eps^{-3} + \eps^{-4})\log(1/\eps)) = O(r\sqrt{\log r})$ for a constant $\eps$.
	
	It remains to bound the running time of the algorithm. Observe that $|\mv_i| = |\mathcal{C}_i|$ and $|\me_i| = O(r\sqrt{\log r})|\mv_i|$. Thus, the running time to (1) construct $\mg_i$ and  $\msttilde_{i+1}$ is  $O_{\eps}(r\sqrt{\log r}|\mc_i|)$ as discussed above, (2) construct  $\mathbb{X}$ is $O_{\eps}(|\mv_i| + |\me_i|) = O_{\eps}(r\sqrt{\log r}|\mathcal{C}_i|)$ by \Cref{lm:Clustering}, and (3) construct  $H_i$ is  $O(|\mv_i| + |\me_i|) = O(r\sqrt{\log r} |\mathcal{C}_i|)$ by \Cref{lm:ConstructH_i} and \Cref{lm:App-Gen}, here $\tau(m,n) = O(1)$. Thus, the total running time to construct level-$(i+1)$ clusters and $H_i$ is  $O_{\eps}(r\sqrt{\log r}|\mc_i|)$. We can construct a minimum spanning tree in time  $T_{\mst} = O(n r\sqrt{\log r})$ by using the algorithm of Mare{\v{s}}~\cite{Mare04}. Thus, by \Cref{lm:framework-technical}, the running time of the light spanner is $O(n r\sqrt{\log r})$ for a constant $\eps$.  
\end{proof}

\section{Clustering: Proof of \texorpdfstring{\Cref{lm:Clustering}}{Clustering Lemma}}\label{sec:ClusteringDetails}

In this section, we construct the set of subgraphs $\mathbb{X}$ of the cluster graph $\mg_i= (\mv_i, \msttilde_{i}\cup \me_i,\omega)$ as claimed in \Cref{lm:Clustering}, by giving a fast implementation of the construction of Borradaile, Le and Wulff-Nilsen (BLW)~\cite{BLW17} using augmented diameters. The pseudocode is given in \Cref{fig:clustering-alg}. Basically, the algorithm has five major steps, each constructing a certain type of cluster, except for Step 3, whose goal is to clean up long paths of $\msttilde_i$. In \Cref{subsec:detailed-implement}, we expand every step in the pseudocode. 

\begin{figure}[ph!]
\begin{tcolorbox}	
$\textsc{ConstructCluster}(\mg_i = (\mv_i,\msttilde_i \cup \me_i,\omega))$: \begin{itemize}
    \item \emph{Step 1: Group nodes in $\mv^{\high}_{i}$ and its neighbors connected via $\me_i$ into subgraphs.} 
    
    This step constructs a set of subgraphs $\mathbb{X}_1$ such that every node in $\mv^{\high}_{i}$ and its neighbors connected via edges in $\me_i$ are grouped to some subgraph in  $\mathbb{X}_1$. One key property is that every subgraph $\mx \in \mathbb{X}_1$ has many nodes (at least $2g/\eps$). See \Cref{lm:Clustering-Step1} for details.  
  
    \item \emph{Step 2: Group branching nodes of $\msttilde_i$.}
    
    A node is \emph{branching}\footnote{We actually work with a more refined notion of branching in the detailed implementation.} in $\msttilde_i$ if its degree in  $\msttilde_i$ is at least 3.  In this step, we form $\mathbb{X}_2$ such that every subgraph in $\mathbb{X}_2$ (i) is a subtree of $\msttilde_i$ and (ii) contains a branching node. After this step, every remaining subtree of  $\msttilde_i$ either has a small augmented diameter (at most $6L_i$) or is a path. See \Cref{lm:Clustering-Step2} for details.

    \item \emph{Step 3: Augment $\mathbb{X}_1 \cup \mathbb{X}_2$.}
    
    A path, say $\tilde{P}$, of $\msttilde_i$ after Step 2 could still contain branching nodes of $\msttilde_i$, which we want to avoid in subsequent steps. Therefore, if $\tilde{P}$ contains any branching node, say $\varphi$, we will add $\varphi$ to a subgraph in $\mathbb{X}_1 \cup \mathbb{X}_2$ that has an $\msttilde_i$ edge to $\varphi$. This augmentation does not change the structure of the subgraph in $\mathbb{X}_1 \cup \mathbb{X}_2$ by much, and more importantly, every node in the remaining long paths has a degree at most $2$ in $\msttilde_i$; these are called \emph{suspended paths} of $\msttilde_i$. See \Cref{lm:Clustering-Step3} for details. 

    \item  \emph{Step 4: Group suspended subpaths connected by an edge in $\me_i$.}

    This step constructs $\mathbb{X}_4$ such that every subgraph $\mx \in \mathbb{X}_4$ contains a single edge $\mbe \in \me_i$ whose endpoints are in (long) suspended paths of  $\msttilde_i$ after Step 3.  The goal is to ensure that, after this step, only short prefix subpaths of $\varphi$ can contain nodes that are incident to an edge in $\me$. The details of this step are given in \Cref{lm:Clustering-Step4}.     

   \item  \emph{Step 5:  Break long suspended paths and final augmentation.}

   This step has two mini steps. (Step 5A) we merge subtrees of $\msttilde_i$ of augmented diameter at most $6L_i$  to subgraphs in $\mathbb{X}_1\cup \mathbb{X}_2\cup \mathbb{X}_4$ via $\msttilde_i$ edges.  (Step 5B) we break every remaining (long) suspended path, say $\Tilde{P}$, into subpaths of diameter $\Theta(L_i)$. There are two types of subpaths broken from  $\Tilde{P}$: prefix subpaths and internal subpaths. Internal subpaths are added to a new set $\mathbb{X}_5^{\internal}$. For a prefix subpath, if it has an $\msttilde_i$ edge to a subgraph $\mx \in \mathbb{X}_1\cup \mathbb{X}_2\cup \mathbb{X}_4$, it will be merged to $\mx$; otherwise, it will be added to a new set $\mathbb{X}_5^{\prefix}$. 

We show (in \Cref{lm:Step1Potential,lm:Step2Potential,lm:Step4Potential}) that the corrected potential changes of subgraphs in $\mathbb{X}_1\cup \mathbb{X}_2\cup \mathbb{X}_4$ remain the same after the augmentation.

    \item Return $\mathbb{X} = \mathbb{X}_1 \cup \mathbb{X}_2\cup \mathbb{X}_4 \cup \mathbb{X}_5^{\internal} \cup \mathbb{X}_5^{\prefix}$ as the set of clusters.  
\end{itemize}

\end{tcolorbox}
  \caption{The algorithm for constructing $\mathbb{X}$.}
    \label{fig:clustering-alg}
\end{figure}

\subsection{The detailed implementation}\label{subsec:detailed-implement}

Recall that $g$ is a constant defined in \hyperlink{P3}{property (P3)}, and that $\msttilde_{i}$ is a spanning tree of $\mg_i$ by Item (2) in \Cref{def:GiProp}. We refer readers to \Cref{table:notation} for a summary of the notation introduced in \Cref{sec:framework}.

\paragraph{Step 1} In this step, we group every node of high degree to a subgraph in the following lemma. 

\begin{lemma}[Step 1]\label{lm:Clustering-Step1} Let $\mv^{\high}_i = \{\varphi_{C} \in \mv: \varphi_{C} \mbox{ is incident to $\geq\frac{2g}{\eps}$ edges in } \me_i\}$. Let $\mv_i^{\high+}$ be obtained from $\mv^{\high}_i$  by adding all neighbors that are connected to nodes in $\mv^{\high}_i$ via edges in $\me_i$. We can construct in $O(|\mv_i| + |\me_i|)$ time a collection of node-disjoint subgraphs $\mathbb{X}_1$ of $\mg_i$ such that:
		\begin{enumerate}[noitemsep]
			\item[(1)] Each subgraph $\mx \in \mathbb{X}_1$ is a tree.
			\item[(2)] $\cup_{\mx \in \mathbb{X}_1}\mv(\mx) = \mv^{\high+}_i$.
			\item[(3)] $L_i \leq \adm(\mx) \leq 13L_i$, assuming that $\eps \leq 1/g$ for every $\mx \in \mathbb{X}_1$.
			\item[(4)] $|\mv(\mx)|\geq \frac{2g}{\eps}$ for every $\mx \in \mathbb{X}_1$.
		\end{enumerate}
\end{lemma}
\begin{proof} Let $\mathcal{J} = (\mv_i,\me_i)$ be the subgraph of $\mg_i$ with the same vertex set and with edge set $\me_i$. Let $\mathcal{N}_{\mathcal{J}}(\varphi)$ be the set of neighbors of a node $\varphi$ in $\mathcal{J}$, and  $\mathcal{N}_{\mathcal{J}}[\varphi] = \mathcal{N}_{\mathcal{J}}(\varphi)\cup \{\varphi\}$.   We construct $\mathbb{X}_1$ in three steps; initially, $\mathbb{X}_1 = \emptyset$.
	\begin{enumerate}
		\item[(1)] Let $\mathcal{I}$ be  a \emph{maximal} set of nodes in $\mathcal{V}^{\high}$ such that for any two nodes $\varphi_{1},\varphi_{2} \in \mathcal{I}$,  $\mathcal{N}_{\mathcal{J}}[\varphi_{1}] \cap \mathcal{N}_{\mathcal{J}}[\varphi_{2}] = \emptyset$. (We can construct $\mathcal{I}$ greedily by adding one node from $\mathcal{V}^{\high}$ at a time to $\mathcal{I}$, and deleting all the nodes in the second neighborhood of the added node.) For each node $\varphi \in \mathcal{I}$, we form a subgraph $\mx$ that consists of $\varphi$, its neighbors $\mathcal{N}_{\mathcal{J}}[\varphi]$, and all incident edges in $\me_i$ of $\varphi$. We then add $\mx$ to $\mathbb{X}_1$.
		
		\item[(2)]  For every node $\varphi \in \mv^{\high}_i\setminus \mathcal{I}$, we do the following. Observe that $\varphi$ must have a neighbor $\varphi'$ that is already grouped to a subgraph, say $\mx \in \mathbb{X}_1$; if there are multiple such neighbors, we pick one of them arbitrarily. We add $\varphi$ and the edge $(\varphi,\varphi')$ to $\mathcal{X}$. Observe that every node in $\mv^{\high}_i$ is grouped to some subgraph at the end of this step.
		\item[(3)]  For each node $\varphi$ in $\mv_i^{\high+}$ that has not been grouped to a subgraph in steps (1) and (2), there must be at least one neighbor, say $\varphi'$, of $\varphi$ that is grouped in step (1) or step (2) to a subgraph $\mx \in \mathbb{X}_1$; if there are multiple such nodes, we pick one of them arbitrarily. We then add $\varphi$ and the edge $(\varphi,\varphi')$ to $\mx$. 
	\end{enumerate}
This completes the construction of $\mathbb{X}_1$. We now show that subgraphs in $\mathbb{X}_1$ have all desired properties.

Observe that Items (1) and (2) follow directly from the construction. For Item (4), we observe that every subgraph $\mx \in \mathbb{X}_1$ is created in step (1) and hence, contains a node $\varphi \in \mv^{\high}_i$ and all of its neighbors (in $\mathcal{J}$) by the definition of $\mi$. Thus, $|\mv(\mx)|\geq 2g/\eps$ since $\varphi$ has at least $2g/\eps$ neighbors.  

For Item (3), we observe that each subgraph $\mx\in \mathbb{X}_1$ after step (3) has hop-diameter\footnote{The {\em hop-diameter} of a graph is the maximum hop-distance over all pairs of vertices, where the {\em hop-distance} between a pair of vertices is the minimum number of edges over all paths between them.} at least $2$ and at most $6$. Recall that every edge $\mbe\in \me_{i}$ has weight of at most $L_i$, and every node has weight of at most $g\eps L_i$, which is at most $L_i$ since $\eps \leq 1/g$. Thus, $\adm(\mx) \leq 7 g\epsilon L_i + 6L_i ~\leq~ 13L_i$. Furthermore, since every edge $\mbe\in \me_{i}$ has weight of at least $L_i/(1+\psi)\geq L_i/2$ and $\mx$ has at least two edges in $\me_i$,  $\adm(\mx)\geq 2(L_i/2) = L_i$; this implies Item (3). 
  
For the construction time, first note that $\mathcal{I}$ can be constructed via a greedy linear-time algorithm; hence step (1) can be carried out in $O(|\mathcal{V}_i| + |\mathcal{E}_i|)$ time. Steps (2) and (3)  can be implemented within this time in a straightforward way; this implies the claimed running time.
\end{proof}


Given a forest $F$, we say that $x$ is \emph{$F$-branching} if it has degree at least $3$ in $F$. For brevity, we shall omit the prefix $F$ in ``$F$-branching'' whenever this does not lead to confusion. The construction of Step 2 is described in the following lemma.

\paragraph{Step 2} In this step, we form subtrees of $\msttilde_{i}$, each of which contains at least one branching node. 

\begin{lemma}[Step 2]\label{lm:Clustering-Step2} Let $\Ftilde^{(2)}_i$ be the forest obtained from $\msttilde_{i}$ by removing every node in $\mv^{\high+}_i$ (defined in \Cref{lm:Clustering-Step1}). We can construct in $O(|\mv_i|)$ time a collection $\mathbb{X}_2$ of subtrees of $\Ftilde^{(2)}_i$ such that for every $\mx\in \mathbb{X}_2$:
	\begin{enumerate}[noitemsep]
		\item[(1)] $\mx$ is a tree and has an $\mx$-branching node.
		\item[(2)] $L_i \leq \adm(\mx)\leq 2L_i$.
		\item[(3)] $|\mv(\mx)| = \Omega(\epsilon^{-1})$  when $\epsilon \leq 1/g$. 
		\item[(4)] Let $\Ftilde^{(3)}_i$ be obtained from $\Ftilde^{(2)}_i$ by removing every node contained in subgraphs of $\mathbb{X}_2$. Then, for every tree $\Ttilde \subseteq \Ftilde^{(3)}_i$, (4a) $\adm(\Ttilde)\leq 6L_i$ or (4b) $\Ttilde$ is a path.
	\end{enumerate}
\end{lemma}
\begin{proof}
	 We say that a tree $\Ttilde \in \Ftilde^{(2)}_i$ is  \emph{long} if $\adm(\Ttilde)\geq 6L_i$ and \emph{short} otherwise. We construct $\mathbb{X}_2$, initially empty, as follows:
	 \begin{itemize}
	 	\item While  there  exists a a long tree  $\Ttilde$ of $\Ftilde^{(2)}_i$ that has at least one $\Ttilde$-branching node, say $\varphi$, we traverse $\Ttilde$ (by increasing distances) starting from $\varphi$ and {\em truncate} the traversal at nodes whose augmented distance from $\varphi$ is at least $L_i$, which will be the leaves of the subtree. (The exact implementation details are delayed until the end of this proof.) 
	 	As a result, the augmented radius (with respect to the center $\varphi$) of the subtree induced by the visited (non-truncated) nodes is at least $L_i$ and at most $L_i + \bar{w} + g\epsilon L_i$. (Here, $\bar{w}$ is an upper bound on the weights of $\msttilde_{i}$ edges, and $ g\epsilon L_i$ is an upper bound on node weights.)  We then form a subgraph, say $\mx$, from the subtree of $\tilde{F}^{(2)}_i$ induced by the visited nodes, add $\mx$ to $\mathbb{X}_2$, remove every node of $\mx$ from $\Ttilde$, and update  $\Ftilde^{(2)}_i$.
	 \end{itemize}
 	We observe that Item (1) follows directly from the construction. Since the algorithm only stops when every long tree has no branching node, meaning that it is a path, Item (4) is satisfied.  We now show Items (2) and (3).
 	
	By construction, $\mx$ is a tree of augmented radius at least $L_i$ and at most $L_i + g\epsilon L_i +\bar{w}$, hence $L_i ~\le~ \adm(\mx) ~\leq~ 2(L_i + g\epsilon L_i +\bar{w}) ~\leq~ 6L_i$ since $\bar{w} < L_i$ and $\epsilon \leq 1/g$; this implies Item (2). 
	
	 Let $\md$ be an augmented diameter path of $\mx$; $\adm(\md)\geq L_i$ by construction. Note that every edge has weight of at most $\bar{w} \leq L_{i-1}$ and every node has weight in $[L_{i-1},gL_{i-1}]$ by \hyperlink{P3'}{property (P3')}. Thus,  $\md$ has at least $\frac{\adm(\md)}{2gL_{i-1}} ~\geq~ \frac{L_i}{2g\eps L_i} = \Omega(\epsilon^{-1})$ nodes; this implies Item (3).
	
	 It remains to show that the construction of $\mathbb{X}_2$ can be implemented efficiently. First, we construct $\Ftilde^{(2)}_i$ by simply going through every node in $\mv_i$ and remove nodes that are grouped in $\mv^{\high+}_{i}$. 	We maintain a list $\mathcal{B}$ of branching nodes of $\Ftilde^{(2)}_i$; all branching nodes can be found in $O(|\mv(\Ftilde^{(2)}_i)|) = O(|\mv_i|)$ time. Note that $\Ftilde^{(2)}_i$ changes during the course of the construction. Initially, nodes in $\mathcal{B}$ are \emph{unmarked}. 
\begin{tcolorbox}	
    While $|\mathcal{B}| \not= \emptyset$:
        \begin{itemize}[noitemsep]
            \item Let $\varphi$ be a node in $\mathcal{B}$.
            \item  If $\varphi$ is marked or no longer is a branching node (of some tree in current $\Ftilde^{(2)}_i$).
            \begin{itemize}
                \item Remove $\varphi$ from $\mathcal{B}$.
            \end{itemize}
            \item  Else ($\varphi$ is an unmarked, branching node)
            \begin{itemize}
                \item  Let $\Ttilde$ be the tree containing $\varphi$.
                \item  Traverse $\Ttilde$ starting from $\varphi$ until the augmented radius of the subtree induced by visited nodes of $\Ftilde^{(2)}_i$, denoted by $\Ttilde_{\varphi}$, is at least $L_i$. It is possible that all nodes of the tree $\Ttilde$ are visited before the radius gets to be $L_i$, in which case we have $\Ttilde_{\varphi} = \Ttilde$ and $\Ttilde_{\varphi}$ will not be added as a subgraph of $\mathbb{X}_2$.
                \item  Mark every node of $\Ttilde_\varphi$ and remove every node in $\Ttilde_\varphi$ from $\Ftilde^{(2)}_i$.
            \end{itemize}
        \end{itemize}
\end{tcolorbox}

	Clearly, maintaining the list $\mathcal{B}$ throughout this process can be carried out in $O(|\mv(\Ftilde^{(2)}_i)|)$ time.
	Other than that, each iteration of these three steps can be implemented in time linear in the number of nodes visited during that iteration plus the number of edges in $\Ftilde^{(2)}_i$ incident to those nodes; also note that once a node is visited, it will no longer be considered in subsequent iterations. It follows that the total running time is $O(|\mv_i|)$.
\end{proof}

 The goal of constructing a subgraph from a branching node $\varphi$ is to guarantee that the $\tilde{T}$-branching node $\varphi$ is also a $\mathcal{X}$-branching as in item (1) of \Cref{lm:Clustering-Step2}. Thus, there must be at least one neighbor, say $\varphi'$, of $\varphi$ that does not belong to the augmented diameter path of $\mx$. Then we could show that the amount of corrected potential change $\Delta^{+}_{i+1}(\mx)$ is at least $ \omega(\varphi')\geq L_{i-1} = \eps L_i$. This will ultimately help us show that the corrected potential change $\Delta^{+}_{i+1}(\mx)$  is $\Omega(\eps^{2} |\mv(\mx)|L_i)$.

\paragraph{Step 3: Augmenting $\mathbb{X}_1\cup \mathbb{X}_2$}   Let $\Ftilde^{(3)}_i$ be the forest obtained in Item (4b) in \Cref{lm:Clustering-Step2}. Let $\mathcal{A}$ be the set of all nodes $\varphi$ in $\Ftilde^{(3)}_i$ such that $\varphi$ is in a tree $\Ttilde \in \Ftilde^{(3)}_3$ of augmented diameter at least $6L_i$ and $\varphi$ is a branching node in $\msttilde_i$.  For each node $\varphi \in \mathcal{A}$ such that $\varphi$ is connected to a node, say $\varphi'$,  in a subgraph $\mx \in \mathbb{X}_1\cup \mathbb{X}_2$ via an $\msttilde_{i}$ edge $\mathbf{e}$, we add $\varphi$ and $\mathbf{e}$ to $\mx$.  We note that $\varphi'$ exists since $\varphi$ has degree at least $3$ in  $\msttilde_i$. (If there are many such nodes $\varphi'$, we choose an arbitrary one.)

\begin{lemma}\label{lm:Clustering-Step3} The augmentation in Step 3 can be implemented in $O(|\mv_i|)$ time and increases the augmented diameter of each subgraph in  $\mathbb{X}_1\cup \mathbb{X}_2$ by at most $4L_i$ when $\eps \leq 1/g$. \\
Furthermore, let $\Ftilde^{(4)}_i$ be the forest obtained from $\Ftilde^{(3)}_i$ by removing every node in $\mathcal{A}$. Then, for every tree $\Ttilde \subseteq \Ftilde^{(4)}_i$, either:
	\begin{enumerate}[noitemsep]
		\item[(1)]$\adm(\Ttilde)\leq 6L_i$ or
		\item[(2)] $\Ttilde$ is a path such that (2a)  every node in $\Ttilde$ has \emph{degree at most $2$} in $\msttilde_{i}$ and (2b) at least one endpoint $\varphi$ of  $\Ttilde$ is connected via an $\msttilde_{i}$ edge to a node $\varphi'$ in a subgraph of $\mathbb{X}_1\cup \mathbb{X}_2$, unless $\mathbb{X}_1\cup \mathbb{X}_2 = \emptyset$. We say that  $\Ttilde$ is a \emph{suspended path} of $\msttilde_i$.
	\end{enumerate}
\end{lemma}
\begin{proof}
	Since every $\msttilde_{i}$ edge has weight of at most $\bar{w}\leq L_i$ and every node has weight of at most $g\eps L_i \leq L_i$ when $\eps \leq 1/g$, the augmentation in Step 3 increases the augmented diameter of each subgraph in  $\mathbb{X}_1\cup \mathbb{X}_2$ by at most $2(\bar{w} + 2g\eps L_i) ~\leq~ 4L_i$. 
	
	For the implementation, we first find the set $\mathcal{A}$ in $O(|\mv_i|)$ time in a straightforward way. Then for each node $\varphi \in \mathcal{A}$, we can check its neighbors in $\msttilde_{i}$ to find a node $\varphi'$ as described in Step 3; indeed, we only need to check at most three neighbors of $\varphi$. Thus, the running time of Step 3 is $O(|\mv_i|)$.
	
	Items (1) and (2a) follow directly from the construction. For Item (2b), we note that since $\msttilde_{i}$ is a spanning tree (and hence connected), $\Ttilde$ must be connected via an $\msttilde_{i}$ edge, say $\mathbf{e}$, to another node not in $\Ttilde$, assuming that $\mathbb{X}_1\cup \mathbb{X}_2 \not= \emptyset$.  Since every node in $\Ttilde$ has degree at most $2$, the endpoint of $\mathbf{e}$ in $\Ttilde$ must be one of the two endpoints of $\Ttilde$, as claimed. 
\end{proof}

The main intuition behind Step 3 is to guarantee properties (2a) and (2b) for every long path $\Ttilde \in \Ftilde^{(4)}_i$. Recall that in Item (3) of \Cref{def:GiProp}, we guarantee that $\mg_i$ has no removable edge. Thus, any edge between two nodes in $\Ttilde$ is not removable. Later, we use this property to argue that the corrected potential change $\Delta^+_{i+1}(\mx)$ is non-trivial for every subgraph $\mx$ formed in the construction of Step 4 below.

\paragraph{Step 4: Grouping suspended subpaths} Let $\Ftilde^{(4)}_i$ be the forest obtained from $\Ftilde^{(3)}_i$ as described in \Cref{lm:Clustering-Step3}. By Item (2b) in  \Cref{lm:Clustering-Step3}, every tree of augmented diameter at least $6L_i$ of $\Ftilde^{(4)}_i$ is a simple path, which we call a \emph{long path}.

\begin{quote}
	\textbf{Red/Blue Coloring.~}\hypertarget{RBColoring}{}  Given a suspended path $\Ptilde\subseteq \Ftilde^{(4)}_i$, we color their nodes red or blue. If a node has augmented distance at most $L_i$ from at least one of the path's endpoints, we color it red; otherwise, we color it blue. Observe that each red node belongs to the suffix or prefix of $\mathcal{P}$; the other nodes are colored blue. 
\end{quote}

The construction of Step 4 is described by the following lemma. We include the proof of all claimed properties except Item (4), which will be delayed to \Cref{subsec:ProofStep3}, as its proof is more complicated.

\begin{lemma}[Step 4]\label{lm:Clustering-Step4} Let $\Ftilde^{(4)}_i$ be the forest obtained from $\Ftilde^{(3)}_i$ as described in \Cref{lm:Clustering-Step3}. We can construct in $O((|\mv_i| + |\me_i|)\epsilon^{-1})$ time a collection $\mathbb{X}_4$ of subgraphs of $\mg_i$ such that for every $\mx\in \mathbb{X}_4$:
	\begin{enumerate}[noitemsep]
		\item[(1)] $\mx$ contains a single edge in $\me_i$.
		\item[(2)] $L_i \leq \adm(\mx)\leq 5L_i$.
		\item[(3)]  $|\mv(\mx)| = \Theta(\epsilon^{-1})$ when $\epsilon \leq 1/(8(g+1))$. 
		\item[(4)] $\Delta_{i+1}^{+}(\mx) = \Omega(\eps^2 |\mv(\mx)| L_i)$.
		\item[(5)] Let $\Ftilde^{(5)}_i$ be obtained from $\Ftilde^{(4)}_i$ by removing every node contained in subgraphs of $\mathbb{X}_4$. If we apply \hyperlink{RBColoring}{Red/Blue Coloring} to each suspended path of augmented diameter at least $6L_i$ in $\Ftilde^{(5)}_i$, then there is no edge in $\me_i$ that connects two blue nodes in $\Ftilde^{(5)}_i$.
	\end{enumerate}
\end{lemma}
\begin{proof} We only apply the construction to paths of augmented diameter at least $6L_i$ in  $\Ftilde^{(4)}_i$, called \emph{long paths}.  
	
	Let $\Ptilde$ be a long path. For each blue node $\varphi \in \Ptilde$, we assign a subpath $\mathcal{I}(\varphi)$ of $\Ptilde$,
	called the \emph{interval of $\varphi$}, which contains every node within an augmented distance (in $\Ptilde$) at most $L_i$ from $\varphi$. By definition, we have:
	
	\begin{claim}\label{clm:Interval-node}
		For any blue node $\nu$, it holds that
		\begin{itemize}[noitemsep]
			\item[(a)] $ (2 - (3g+2)\epsilon)L_i \leq  \adm(\mathcal{I}(\nu))\leq 2L_i $.
			\item[(b)]   	Denote by  $\mi_1$ and $\mi_2$  the two subpaths obtained by removing $\nu$ from the path $\mi(\nu)$. 
			Each of these subpaths has $\Theta(\epsilon^{-1})$ nodes and augmented diameter at least $(1-2(g+1)\epsilon)L_i$.
		\end{itemize}
	\end{claim}
	\begin{proof}  
		(a) The upper bound on the augmented diameter of $\mi(\nu)$ follows directly from the construction.  Thus, it remains to prove the lower bound on $\adm(\mathcal{I}(\nu))$. Let $\Ptilde$ be the path containing $\mathcal{I}(\nu)$. Let $\mu$ be an endpoint of $\mathcal{I}(\nu)$. Let $\mu'$ be the neighbor of $\mu$  in $\Ptilde\setminus \mathcal{I}(\nu)$; $\mu'$ exists since $\nu$ is a blue node (see Figure~\ref{fig:Iv}).  Observe that $\adm(\Ptilde[\nu,\mu'])\geq L_i$. Thus, we have:
		\begin{equation}\label{eq:Step3-Puv}
			\adm( \Ptilde[\nu,\mu])\geq L_i -   \bar{w}- \omega(\mu') \geq  (1- (g+1)\epsilon) L_i
		\end{equation}
		since  $\omega(\mu')\leq g\epsilon L_{i}$ by property (P5)  and $L_i ~ =~  L_{i-1}/\eps ~\geq~ \bar{w}/\eps$ when $i \geq 1$. Thus,
		\begin{equation*}
			\adm(\mathcal{I}(\nu)) \geq 2(1- (g+1)\epsilon) L_i  - \omega(\nu) \geq (2-(3g+2)\epsilon)L_i.
		\end{equation*}
		The first inequality in the above equation is because we count  $\omega(\nu)$ twice in the sum of the augmented diameters of two paths from $\nu$ to each endpoint of $\mi(\nu)$.
		\begin{figure}[!htb]
			\begin{center}
				\includegraphics[width=0.4\textwidth]{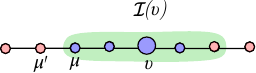}
			\end{center}
			\caption{Nodes in the green shaded region belong to $\mi(\nu)$.}
			\label{fig:Iv}
		\end{figure}
		
		(b) We focus on bounding $\adm(\mi_1)$; the same bound applies to $\adm(\mi_2)$. We assume w.l.o.g.~that $\mi_1 \subseteq \Ptilde[\nu,\mu]$ and hence $\adm(\mi_1)~\geq~ \adm(\Ptilde[\nu,\mu]) -  \bar{w}- \omega(\nu) ~\geq~ (1- 2(g+1)\epsilon) L_i$.
		
		We now bound $|\mv(\mi_1)|$. The upper bound on the number of nodes of $\mi_{1}$ follows from the fact that $\mi_{1}$ has augmented diameter at most $2L_i$ (see Item (a))  	and each node has weight of at least $L_{i-1} = L_i \epsilon $ by \hyperlink{P3'}{property (P3')}. To show the lower bound on the number of nodes of $\mi_{1}$, we observe that  $\mi_{1}$ has an augmented diameter at least 
		$(1 - (g+1)\epsilon)L_i$, which is at least $L_i/2$ when $\epsilon \leq \frac{1}{2(g+1)}$,
		while each edge in $\mi_1$ has weight of at most $L_{i-1}$ and each node has weight of at most $g L_{i-1}$. It follows that $|\mv(\mi_1)|\geq \frac{\adm(\mi_1)}{(1+g)L_{i-1}} = \Omega(\epsilon^{-1})$. The same bound holds for $|\mv(\mi_2)|$.		
	\end{proof}
	
	We keep track of a list $\mathcal{B}$ of edges in $\me_i$ with both \emph{blue endpoints}. We then construct $\mathbb{X}_4$, initially empty, as follows:
	
	\begin{itemize}
		\item While $\mathcal{B}\not=\emptyset$, we pick an edge $(\nu,\mu)$ with both endpoints blue,  form a subgraph $\mx = \{(\nu,\mu)\cup \mathcal{I}(\nu) \cup \mathcal{I}(\mu)\}$, and add $\mx$ to $\mathbb{X}_4$. We then remove all nodes in  $\mathcal{I}_{\nu} \cup \mathcal{I}_{\mu} $ from the path or two paths containing $\nu$ and $\mu$, update the color of nodes in the new paths to satisfy \hyperlink{RBColoring}{Red/Blue Coloring} and the edge set $\mathcal{B}$.
	\end{itemize}
	
	We observe that Items (1) and (5) follow directly from the construction. For Item (2), we observe by Claim~\ref{clm:Interval-node}  that $\mathcal{I}(v)$ has augmented diameter at most $2L_i$ and at least $L_i$ when $\epsilon \leq \frac{1}{8(g+1)}$, and  the weight of the edge $(\mu,\nu)$ is at most $L_i$. Thus, $L_i \leq \adm(\mx)\leq L_i + 2\cdot 2L_i = 5L_i$, as claimed. Item (3) follows directly from  Claim~\ref{clm:Interval-node} since $|\mathcal{I}(v)| = \Theta(\epsilon^{-1})$ and $|\mathcal{I}(\mu)| = \Theta(\epsilon^{-1})$. The proof of Item (4) is delayed to \Cref{subsec:ProofStep3}. In a nutshell, the proof is divided into two cases: (a) $\mathcal{I}(\nu) \cap  \mathcal{I}(\mu) = \emptyset$ and (b) $\mathcal{I}(\nu) \cap  \mathcal{I}(\mu) \not= \emptyset$. In the former case, we show that $\Delta^{+}_{i+1}(\mx)= \Omega(|\mx_i|\eps L_i)$; the proof is by a straightforward calculation. In the latter case, we show that  $\Delta^{+}_{i+1}(\mx)= \Omega(|\mx_i|\eps^2 L_i)$; the proof crucially uses the fact that $\mg_i$ has no removable edge (see Item (3) in \Cref{def:GiProp}) and that $\eps \leq \frac{1}{8(g+1)}$.
		
	
	Finally, we show that the construction of $\mathbb{X}_4$ can be implemented efficiently.  Observe that for each long path $\Ptilde$, coloring all nodes of $\Ptilde$  can be done in $O(|\mv(\Ptilde)|) = (|\mv_i|)$ time.   Since the interval $\mi(\nu)$ assigned to each blue node $\nu$ consists of $O(\epsilon^{-1})$ nodes by Claim~\ref{clm:Interval-node}(b), listing intervals for all blue nodes can be carried out within time $O(|\mv(\Ptilde)|\eps^{-1}) = O(|\mv_i|\eps^{-1})$. For each edge $(\nu,\mu)\in \mathcal{E}_i$, we can check whether both endpoints are blue in  $O(1)$ time. Thus, 
	it takes $O(|\mathcal{E}_i|\epsilon^{-1})$ time to construct  $\mathcal{B}$.

	For each edge $(\nu,\mu) \in \mathcal{B}$ picked in the construction of $\mathbb{X}_4$, forming $\mx = \{(\nu,\mu)\cup \mathcal{I}(\nu) \cup \mathcal{I}(\mu)\}$ takes $O(1)$ time. When removing any such interval $\mathcal{I}(\nu)$ from a path $\Ptilde$, we may create two new sub-paths $\Ptilde_1,\Ptilde_2$, and then need to recolor the nodes following  \hyperlink{RBColoring}{Red/Blue Coloring}. Specifically, some blue nodes in the prefix and/or suffix of $\Ptilde_1,\Ptilde_2$ are colored red; importantly, a node's color may only change from blue to red, but it cannot change in the other direction.

	Since the total number of nodes to be recolored as a result of removing such an interval $\mathcal{I}(\nu)$ is $ O(\epsilon^{-1})$,  the total recoloring running time is $O(|\mv(\Ftilde^{(4)}_i)|\epsilon^{-1}) = O(|\mv_i|\eps^{-1})$.
	To bound the time required for updating the edge set $\mathcal{B}$ throughout this process, we
	note that edges are never added to $\mathcal{B}$ after its initiation. Specifically, when a blue node $\nu$ is 
	recolored as red, we remove all incident edges of $\nu$ from  $\mathcal{B}$, and none of these edges will be considered again; this can be done in $O(\epsilon^{-1})$ time per node $\nu$, since $\nu$ is incident to at most $\frac{2g}{\eps} = O(\frac{1}{\eps})$ edges in $\me_i$ due to the construction of Step 1 (\Cref{lm:Clustering-Step1}). Once a node is added to $\mx$, it will never be considered again. It follows that the total running time required for implementing Step 4 is  $O((|\mv_i| + |\me_i|)\epsilon^{-1})$, as claimed. 
\end{proof}

\begin{remark}\label{remark:Clustering-Step4} Item (5) of \Cref{lm:Clustering-Step4} implies that for every edge $(\varphi_{C},\varphi_{C'})\in \me_i$ with both endpoints in $\mv(\Ftilde^{(5)}_i)$, at least one of the endpoints must belong to a low-diameter tree of $\Ftilde^{(5)}_i$ or a (red) suffix of a long path in $\Ftilde^{(5)}_i$.
\end{remark}

\begin{observation}\label{obs:Clustering-F5} Every tree $\Ttilde \subseteq \Ftilde^{(5)}_i$ such that $\adm(\Ttilde) \leq 6L_i$ is connected via an $\msttilde_{i}$ edge to a node in some subgraph $\mx \in \mathbb{X}_1 \cup \mathbb{X}_2\cup \mathbb{X}_4$, unless there is no subgraph formed in Steps 1-4, i.e, $ \mathbb{X}_1 \cup \mathbb{X}_2\cup \mathbb{X}_4 = \emptyset$. 
\end{observation}

We call the case where $\mathbb{X}_1 \cup \mathbb{X}_2\cup \mathbb{X}_4 = \emptyset$ the \emph{degenerate case}. In the degenerate case,  $\mg_i$ has a very special structure, which will be described later (in \Cref{lm:exception}); for now, we focus on the construction of the last step. 

\paragraph{Step 5} Let $\Ttilde$ be  a path in  $\Ftilde^{(5)}_i$ obtained by Item (5) of \Cref{lm:Clustering-Step4}. We construct two sets of subgraphs, denoted by $\mathbb{X}^{\internal}_5$ and $\mathbb{X}^{\prefix}_5$, of $\mg_i$, and also modify subgraphs in $\mathbb{X}_1,\mathbb{X}_2$ and  $\mathbb{X}_4$. The construction is broken into two steps. Step 5A is only applicable when we are not in the degenerate case; Step 5B is applicable regardless of the degenerate case.  
\begin{itemize}
	\item (Step 5A)\hypertarget{5A}{}  If $\Ttilde$ has augmented diameter at most $6L_i$, let $\mathbf{e}$ be an $\widetilde{\mst}_i$ edge connecting $\Ttilde$  and a node in some subgraph $\mx \in \mathbb{X}_1\cup \mathbb{X}_2 \cup \mathbb{X}_4$; $\mathbf{e}$ exists by \Cref{obs:Clustering-F5}. We add both $\mathbf{e}$ and $\Ttilde$ to $\mx$.
	\item (Step 5B)\hypertarget{5B}{} 	Otherwise,  the augmented diameter of $\Ttilde$ is at least $6L_i$ and hence, it must be a path by Item (4) in \Cref{lm:Clustering-Step2}.  In this case, we greedily break $\Ttilde$ into subpaths of augmented diameter at least $L_i$ and at most $2L_i$. (This is possible because both edge and node weights are much smaller than $L_i$ for a sufficiently small constant $\eps$.) Let $\Ptilde$ be a subpath broken from $\Ttilde$. If $\Ptilde$ is connected to a node in a subgraph $\mx \in \mathbb{X}_1\cup \mathbb{X}_2 \cup \mathbb{X}_4$ via an  edge $\mathbf{e}\in \msttilde_{i}$, we add $\Ptilde$ and $\mathbf{e}$ to $\mx$.	Else, if $\Ptilde$ contains an endpoint of $\Ttilde$, we add $\Ptilde$ to $\mathbb{X}^{\prefix}_5$; otherwise, we add $\Ptilde$ to $\mathbb{X}^{\internal}_5$. 
\end{itemize}

\begin{lemma}\label{lm:Clustering-Step5} We can implement the construction of $ \mathbb{X}_5^{\internal}$ and $\mathbb{X}_5^{\prefix}$ in $O(|\mv_i|)$ time.	Furthermore, every subgraph $\mx \in \mathbb{X}_5^{\internal} \cup \mathbb{X}_5^{\prefix}$ satisfies:
	\begin{enumerate}[noitemsep]
		\item[(1)] $\mx$ is a subpath of $\msttilde_{i}$.
		\item[(2)] $L_i \leq \adm(\mx)\leq 2 L_i$ when $\eps \leq 1/g$.
		\item[(3)] $|\mv(\mx)| = \Theta(\epsilon^{-1})$.
	\end{enumerate}
\end{lemma}
\begin{proof} 
	Items (1) and (2) follow directly from the construction. For Item (3), we observe the following facts:   $\adm(\mx)\geq L_i$,   each edge has weight of at most $L_{i-1}$, and each node has weight of at most $g L_{i-1}$. Thus,  $|\mv(\mx)|\geq \frac{L_i}{(1+g)L_{i-1}} = \Omega(\epsilon^{-1})$. By the same argument, since each node has weight at least $L_{i-1}$ by \hyperlink{P3'}{property (P3')},  $|\mv(\mx)|\leq \frac{2L_i}{L_{i-1}} = O(1/\eps)$; this implies Item (3). 
	
	We now focus on the construction time. We observe that for every tree $\Ttilde\in \Ftilde^{(5)}_i$, computing its augmented diameter can be done in $O(|\mv(\Ttilde)|)$ time. Thus, we can identify all trees of $\Ftilde^{(5)}_i$ of augmented diameter at least $6L_i$ to process in Step 5B in $O(|\mv(\Ftilde^{(5)}_i)|) = O(|\mv_i|)$ time.   Breaking each path $\Ttilde$  in Step 5B into a collection of subpaths $\{\Ptilde_1,\ldots, \Ptilde_k\}$ greedily can be done in $O(|\mv(\Ttilde)|)$ time. For each $j\in [k]$, to check whether $\Ptilde_j$ is connected by an $\widetilde{\mst}_i$ edge to subgraph in $\mathbb{X}_1\cup \mathbb{X}_2\cup \mathbb{X}_4$, we examine each node $\varphi \in \Ptilde_j$ and all $\widetilde{\mst}_i$ edges incident to $\varphi$. In total, there are at most $|\mv(\Ftilde^{(5)}_i)|$ nodes and $|\widetilde{\mst}_i| = |\mv_i|-1$ edges to examine; this implies the claimed time bound. 
\end{proof}

Finally, we construct the collection $\mathbb{X}$ of subgraphs of $\mg_i$ as follows:
\begin{equation}\label{eq:MathbbX}
	\mathbb{X} = \mathbb{X}_1 \cup \mathbb{X}_2\cup \mathbb{X}_4 \cup \mathbb{X}_5^{\internal} \cup \mathbb{X}_5^{\prefix}.
\end{equation}
We note that in the above equation, $ \mathbb{X}_1, \mathbb{X}_2$, and $\mathbb{X}_4$ are the set of subgraphs after being modified in Steps 3 and 5. To complete the proof of \Cref{lm:Clustering}, we need to:
\begin{enumerate}
	\item show that subgraphs in  $\mathbb{X}$ satisfies three properties: \hyperlink{P1'}{(P1')}, \hyperlink{P2'}{(P2')}, and \hyperlink{P3'}{(P3')}, and that $|\me_i\cap \me(\mx)| = O(|\mv(\mx)|)$. This implies Item (5) of \Cref{lm:Clustering}. We present the proof in \Cref{subsec:PropX}.
	\item  construct a partition $\{\mv_i^{\high},\mv_i^{\lowp}, \mv_i^{\lowm}\}$  of $\mv_i$, show Items (1)-(4) and the running time bound as claimed by \Cref{lm:Clustering}. We present the proof in \Cref{subsec:PartitionMvi}
\end{enumerate}

\subsection{Properties of \texorpdfstring{$\mathbb{X}$}{X}}\label{subsec:PropX}

In this section, we prove the following lemma.

\begin{lemma}\label{lm:XProp} Let $\mathbb{X}$ be the set of subgraphs as defined in \Cref{eq:MathbbX}. For every subgraph $\mx \in \mathbb{X}$, $\mx$ satisfies the three properties (\hyperlink{P1'}{P1'})-(\hyperlink{P3'}{P3'}) with $g = 31$ and $\eps \leq \frac{1}{8(g+1)}$, and $|\me(\mx)\cap \me_i| = O(|\mv(\mx)|)$. Furthermore, $\mathbb{X}$ can be constructed in $O((|\mv_i| + |\me_i|)\eps^{-1})$ time.
\end{lemma}
\begin{proof} We observe that \hyperlink{P1'}{property (P1')} follows directly from the construction. Additionally,  \hyperlink{P2'}{property (P2')} follows from Item (4) of \Cref{lm:Clustering-Step1}, Items (3) of \Cref{lm:Clustering-Step2}, \Cref{lm:Clustering-Step4}, and \Cref{lm:Clustering-Step5}. The lower bound $L_i$ on the augmented diameter  of a subgraph $\mx \in \mathbb{X}$ follows from Item (3) of \Cref{lm:Clustering-Step1}, Items (2) of \Cref{lm:Clustering-Step2}, \Cref{lm:Clustering-Step4}, and \Cref{lm:Clustering-Step5}. Thus, to complete the proof of \hyperlink{P3'}{property (P3')}, it remains to show that $\adm(\mx) \leq gL_i$ with $g = 31$ and $\eps \leq \frac{1}{8(g+1)}$. Observe that  the condition that $\eps \leq \frac{1}{8(g+1)}$ follows by considering all constraints on $\eps$ in \Cref{lm:Clustering-Step1,lm:Clustering-Step2,lm:Clustering-Step3,lm:Clustering-Step4,lm:Clustering-Step5}.

	If $\mx$ is formed in Step 5B, that is $\mx \in \mathbb{X}^{\internal}_5\cup \mathbb{X}^{\prefix}_5$,  then $\adm(\mx)~\leq~ 2L_i$ by \Cref{lm:Clustering-Step5}. Otherwise, excluding any augmentation to $\mx$ due to Step 5, \Cref{lm:Clustering-Step1},  \Cref{lm:Clustering-Step2} and \Cref{lm:Clustering-Step3} yield $\adm(\mx) \leq 13L_i + 4L_i \leq 17 L_i$ where the $4L_i$ term is due to the augmentation in Step 3 (see  \Cref{lm:Clustering-Step3}). By \Cref{lm:Clustering-Step4}, $\adm(\mx) \leq \max(17L_i,5L_i) = 17L_i$.

	 We then may augment $\mx$ with trees of diameter at most $6L_i$ (Step 5A) and/or with subpaths of diameter at most $2L_i$ (Step 5B). A crucial observation is that any augmented tree or subpath is connected by an $\widetilde{\mst}_i$ edge to a node that was grouped to $\mx \in \mathbb{X}_1\cup \mathbb{X}_2\cup \mathbb{X}_4$. If we denote the resulting subgraph by $\mx^+$, then 
	\begin{equation*}
		\adm(\mx^+) \leq \adm(\mx) + 2\bar{w} + 12 L_i \leq \adm(\mx) + 14L_i  \leq 31L_i. \end{equation*}
	In the above equation, term $2\bar{w}$ is from the two $\widetilde{\mst}_i$ edges connecting two augmented trees (or paths), and $12L_i$ is the upper bound on the sum of the augmented diameters of two augmented trees (or paths). \hyperlink{P3'}{Property (P3')} now follows.
	
	The fact that $|\me(\mx)\cap \me_i| = O(|\mv(\mx)|)$ and the running time bound follow directly from  \Cref{lm:Clustering-Step1}, \Cref{lm:Clustering-Step2}, \Cref{lm:Clustering-Step3}, \Cref{lm:Clustering-Step4} and \Cref{lm:Clustering-Step5}. Recall that the augmentation in Step 3 is in a star-like way and hence, no cycle is formed in subgraphs of $\mathbb{X}_1\cup \mathbb{X}_2$ after the augmentation.  
\end{proof}

\subsection{Constructing  a Partition of \texorpdfstring{$\mv_i$}{Vi}}\label{subsec:PartitionMvi}

We first consider the degenerate case where $\mathbb{X}_1\cup \mathbb{X}_2\cup \mathbb{X}_4 = \emptyset$. 
\begin{figure}[!htb]
    \begin{center}
			\includegraphics[width=0.4\textwidth]{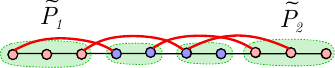}
		\end{center}
		\caption{Red edges are edges in $\me_i$; every edge is incident to at least one red node.}
		\label{fig:exception}
\end{figure}

\begin{lemma}[Structure of Degenerate Case]\label{lm:exception}
	If  $\mathbb{X}_1\cup \mathbb{X}_2\cup \mathbb{X}_4 = \emptyset$, then
	$\Ftilde^{(5)}_i =  \msttilde_{i}$, and $\msttilde_{i}$  is a single (long) path.   	Moreover, every edge $\mbe \in \me_i$ must be incident to a  node in $\Ptilde_1\cup \Ptilde_2$,
	where $\Ptilde_1$ and $\Ptilde_2$ are the prefix and suffix subpaths of $\msttilde_{i}$ of augmented diameter at most $L_i$. Consequently, we have that $|\me_i| = O(1/\epsilon^2)$.
\end{lemma}
\begin{proof}	By the assumption of the lemma, no subgraph is formed in Steps 1-4.

	Since no subgraph is formed in Step 1, $\Ftilde^{(2)}_i = \msttilde_i$. Since no subgraph is formed in Step 2, there is no branching node in $\Ftilde^{(2)}_i$; thus $\Ftilde^{(3)}_i = \Ftilde^{(2)}_i$ and it is a  single (long) path. 	Since $\mathbb{X}_1 \cup \mathbb{X}_2  = \emptyset$, there  is no augmentation in Step 3.  Since no subgraph is formed in Step 4, $\Ftilde^{(5)}_i = \Ftilde^{(4)}_i$ and both are equal to $\msttilde_i$, which is a long path (see Figure~\ref{fig:exception}).

	By Item (5) in \Cref{lm:Clustering-Step4}, any edge $\mbe \in \me_i$ must be incident to a red node. The augmented distance from any red node to at least one endpoint of $\msttilde_{i}$ is at most $L_i$ by the definition of \hyperlink{RBColoring}{Red/Blue Coloring}, and hence every red node belongs to $\Ptilde_1\cup \Ptilde_2$. Since each node has weight of at least $L_{i-1}$ by \hyperlink{P3'}{property (P3')}, we have:
	\begin{equation*}
		|\mv(\Ptilde_1\cup \Ptilde_2)|\leq  \frac{2L_i}{L_{i-1}} = \frac{2}{\epsilon}
	\end{equation*}
	Since each node of $\Ptilde_1\cup \Ptilde_2$ is incident to at most $\frac{2g}{\epsilon}$ edges in $\me_i$ (as there is no subgraph formed in Step 1; $\mv^{\high}_i = \emptyset$), it holds that $|\me_i| = O(1/\epsilon^2)$, as desired.	
\end{proof}

We are now ready to describe the construction of the partition  $\{\mv^{\high}_i, \mv^{\lowp}_i, \mv^{\lowm}_i\}$ of $\mv_i$

\begin{tcolorbox}
	\hypertarget{PartitionV}{}
	\textbf{Construct Partition $\{\mv^{\high}_i, \mv^{\lowp}_i, \mv^{\lowm}_i\}$:} In the degenerate case, we define $\mv^{\lowm}_i = \mv_i$ and $\mv^{\high}_i = \mv^{\lowp}_i = \emptyset$. Otherwise, we define $\mv^{\high}_i$ to be the set of all nodes that are incident to at least $2g/\eps$ edges in $\me_i$,  $\mv^{\lowm}_i = \cup_{\mx \in \mathbb{X}^{\internal}_5}\mv(\mx)$ and $\mv^{\lowp}_i = \mv_i\setminus (\mv^{\high}_i \cup \mv^{\lowm}_i )$. 
\end{tcolorbox}

We show the following properties of $\{\mv^{\high}_i, \mv^{\lowp}_i, \mv^{\lowm}_i\}$, which is equivalent to Item (4) in \Cref{lm:Clustering}.

\begin{lemma}\label{lm:Xlowm} 
	\begin{enumerate}
		\item[(1)] If $\mx$ contains a node in $\mv^{\lowm}$, then $\mv(\mx)\subseteq \mv^{\lowm}$.
		\item[(2)]	There is no edge in $\me_i$ between a node in $\mv^{\high}_i$ and a node in $\mv^{\lowm}_i$.
		\item[(3)]  If there exists an edge   $(\varphi_{C_u},\varphi_{C_v}) \in \me_i$ such that both $\varphi_{C_u}$ and $\varphi_{C_v}$ are in 
		$\mv_i^{\lowm}$, then we are in the degenerate case.
	\end{enumerate}
\end{lemma}
\begin{proof}
	Item (1) follows directly from the construction. We now show Item (2).  By the construction of Step 1 (\Cref{lm:Clustering-Step1}), any neighbor, say $\varphi$, of a node in $\mv^{\high}_i$ is in $\mv^{\high+}_i$. Thus, $\varphi$ will not be considered after Step 1. It follows that there is no edge between a node in $\mv^{\high}_i$ and a node in $\mv^{\lowm}_i$ since nodes in $\mv^{\lowm}_i$ are in Step 5. 
	
	To show Item (3), we observe by the construction that every node, say $\varphi_{C_u}$, in $\mv^{\lowm}$ is a blue node of some long path $\Ptilde$ in $\Ftilde^{(5)}_i$. In a non-degenerate case, then by Item (5) of \Cref{lm:Clustering-Step4}, every edge  $(\varphi_{C_u},\varphi_{C_v})$ must have the node $\varphi_{C_v}$ being a red node of $\Ptilde$. But then by Step \hyperlink{5B}{5B} of the algorithm, $\varphi_{C_v}$ belongs to some subgraph of $\mathbb{X}_5^{\prefix}$ and hence is not in $\mv^{\lowm}$. 
\end{proof}

Next, we focus on bounding the corrected potential change $\Delta^+_i(\mx)$ of every cluster $\mx\in \mathbb{X}$. Specifically, we show that: 
	\begin{itemize}[noitemsep]
		\item if $\mx \in \mathbb{X}_1$, then $\Delta^+_{i+1}(\mx) = \Omega(|\mv(\mx)|L_i\epsilon)$; the proof is in \Cref{lm:Step1Potential}.
		\item if $\mx \in \mathbb{X}_2$, then $\Delta^+_{i+1}(\mx) = \Omega(|\mv(\mx)|L_i\epsilon^2)$; the proof is in \Cref{lm:Step2Potential}. 
		\item  if $\mx \in \mathbb{X}_4$, then $\Delta^+_{i+1}(\mx) = \Omega(|\mv(\mx)|L_i\epsilon^2)$; the proof is in \Cref{lm:Step4Potential}. 
		\item the corrected potential change is non-negative, and we provide a lower bound of the average corrected potential change for subgraphs in $\mathbb{X}\setminus \mathbb{X}^{\lowm}$ in \Cref{lm:Step5PotentialChange}. 
	\end{itemize}

\begin{lemma}\label{lm:Step1Potential} For every subgraph $\mx\in \mathbb{X}_1$, it holds that $ \Delta^+_{i+1}(\mx)\geq   \frac{|\mv(\mx)|L_i\epsilon}{2}$. 
\end{lemma}
\begin{proof}
	 Let $\mx \in \mathbb{X}_1$ be a subgraph formed in Step 1, which could possibly be augmented in Steps  3 and 5. By Item (4) of \Cref{lm:Clustering-Step1}, $|\mv(\mx)| \geq \frac{2g}{\epsilon}$.  Observe by the definition of the corrected potential change (\Cref{def:corrected-potential}),  $\Delta^+_{i+1}(\mx)  \geq \Delta_{i+1}(\mx)$ and hence:
	\begin{equation}\label{eq:Step1-potential}
		\begin{split}
			\Delta^+_{i+1}(\mx) &\geq \sum_{\varphi \in \mv(\mx)}\omega(\varphi)  - \adm(\mx) \stackrel{\mbox{\hyperlink{P3'}{\tiny{(P3')}}}}{\geq} \sum_{\varphi \in \mv(\mx)}L_{i-1} - gL_i \\ &= \frac{|\mv(\mx)| L_{i-1}}{2} + \underbrace{(\frac{|\mv(\mx)| L_{i-1}}{2} - gL_i)}_{\geq 0 \mbox{ since }|\mv(\mx)|\geq (2g)/\epsilon}\\
			&\geq  \frac{|\mv(\mx)| L_{i-1}}{2}  =  \frac{|\mv(\mx)|\epsilon L_{i}}{2},
		\end{split}
	\end{equation} 
as claimed. 
\end{proof}

 When analyzing the corrected potential change, it is instructive to keep in mind the worst-case example, where the subgraph is a path of $\widetilde{\mst}_i$; in this case, it is not hard to verify
(see \Cref{lm:Step5PotentialChange}) that the corrected potential change is $0$. 
However, the key observation is that the worst-case example cannot happen for subgraphs formed in Step 2, 
as any such subgraph (a subtree of $\widetilde{\mst}_i$) is a \emph{$\mx$-branching node}; such a node has at least three neighbors. Consequently, we can show that any subgraph formed in Step 2 has a sufficiently large corrected potential change, as formally argued next.  

\begin{lemma}\label{lm:Step2Potential}For every subgraph $\mx\in \mathbb{X}_2$,  
		$\Delta^+_{i+1}(\mx)\ = \Omega\left( |\mv(\mx)|L_i\epsilon^2\right)$.
\end{lemma}
\begin{proof}
	Let $\mx$ be a subgraph that is initially formed in Step 2 and could possibly be augmented in Steps  3 and 5.   Recall that in the augmentation in Step 3, we add to $\mx$ nodes of $\mv_i$ via  $\widetilde{\mst}_i$ edges, and in the augmentation done in Step 5, we add to $\mx$ subtrees of $\widetilde{\mst}_i$ via $\widetilde{\mst}_i$ edges. Thus, the resulting subgraph after the augmentation remains, as prior to the augmentation, a subtree of $\widetilde{\mst}_i$. That is, $\me(\mx)\subseteq \msttilde_{i}$. Letting $\md$ be an augmented diameter path of $\mx$, we have by the definition of augmented diameter that
	\begin{equation*}
		\adm(\mx)  = \sum_{\varphi \in \md}\omega(\varphi) + \sum_{\mbe \in \me(\md)} \omega(e)
	\end{equation*}
	Let $\my = \mv(\mx)\setminus \mv(\md)$. Then $|\my| \geq 1$  since $\mx$ has a $\mx$-branching node by Item (1) of \Cref{lm:Clustering-Step2} and that 
	\begin{equation}\label{eq:Step2-Potential}
		\begin{split}
			\Delta^+_{i+1}(\mx) =  \left(\sum_{\varphi \in \mx}\omega(\varphi) + \sum_{ e\in \me(\mx)} \omega(e)\right) - \adm(\mx) \geq \sum_{\varphi \in \my}\omega(\varphi) \stackrel{\mbox{\hyperlink{P3'}{\tiny{(P3')}}}}{\geq} |\my|L_{i-1}
		\end{split}
	\end{equation} 
	Note that $\me(\mx)\subseteq \widetilde{\mst}_i$. 	By  \hyperlink{P3'}{property (P3')},  $\adm(\md) \leq gL_i$ while each node has weight of at least $L_{i-1}$. Thus, we have:
	\begin{equation}\label{eq:Step2-D-size}
		|\mv(\md)| \leq \frac{g L_i}{L_{i-1}} = O(\epsilon^{-1})  = O(\frac{|\my|}{\epsilon}),
	\end{equation}
	since $ |\my|\geq 1$. By combining \Cref{eq:Step2-Potential} and ~\Cref{eq:Step2-D-size}, we have
	\begin{equation*}
		\Delta^+_{i+1}(\mx) \geq \frac{|\my| L_{i-1}}{2} +  \Omega(\epsilon|\mv(\md)|L_{i-1}) = \Omega((|\my| + \mv(\md))\epsilon L_{i-1}) = \Omega(|\mv(\mx)|\epsilon^2 L_i),
	\end{equation*}
as claimed.	
\end{proof}

\begin{lemma}\label{lm:Step4Potential} For every subgraph $\mx\in \mathbb{X}_4$,  it holds that 
	$\Delta^+_{i+1}(\mx) = \Omega\left( |\mv(\mx)|L_i\epsilon^2\right)$.
\end{lemma}
\begin{proof}  Let $\mx\in \mathbb{X}_4$ be a subgraph initially formed in Step 4; $\mx$ is possibly augmented in Step 5. Let $\mx^+$ be $\mx$ after the augmentation (if any).  Let $\md^+$ be the augmented diameter path of $\mx^+$ and $\md = \md^+\cap \mx$. Since the augmentation in Step 5 is by attaching trees to $\mx$ via edges, $\md$ is a path in $\mx$. (Note that $\mx$ might contain a cycle, and if there is a cycle, the cycle must contain the single edge of $\mx$ in $\me_i$.) First, we observe that $|\mv(\md^+)| = O(\frac{1}{\eps})$ by the same argument as in \Cref{eq:Step2-D-size}. Furthermore, $|\mv(\mx)| = \Omega(\frac{1}{\eps})$ by Item (3) in \Cref{lm:Clustering-Step4}. Thus, $|\md^+| = O(|\mv(\mx)|)$.

 Let $\my = \mv(\mx^+)\setminus \mv(\md^+)$. Since $|\md^+| = O(|\mv(\mx)|)$, by Item (4) in \Cref{lm:Clustering-Step4}, it holds that:
 
 \begin{equation}\label{eq:Step4PotentialXplus}
 	\Delta_{i+1}^{+}(\mx)  = \Omega(|\mv(\mx)|\eps^2 L_i) =  \Omega(|\mv(\mx) \cup \mv(\md^+)|\eps^2 L_i).
 \end{equation}
Furthermore,
 \begin{equation*}
 	\begin{split}
 		\Delta_{i+1}^{+}(\mx^+) & = \sum_{\varphi \in \mx^+}\omega(\varphi) + \sum_{\mathbf{e}\in \me(\mx^+)\cap \msttilde_i} \omega(\mathbf{e}) - \omega(\md^+)\\
 		&\geq  \sum_{\varphi \in \my}\omega(\varphi) + \sum_{\varphi \in \mx}\omega(\varphi) + \sum_{\mathbf{e}\in \me(\mx)\cap \msttilde_i} \omega(\mathbf{e}) - \omega(\md) \\
 		& \geq \Omega(L_{i}\eps |\my|) +  \Delta_{i+1}^{+}(\mx) \stackrel{\mbox{\tiny{~\cref{eq:Step4PotentialXplus}}}}{=} \Omega(|\my|\eps L_i)  + \Omega(|\mv(\mx) \cup \mv(\md^+)|\eps^2 L_i)\\
 		&= \Omega(|\mv(\mx) \cup \mv(\md^+) \cup \my|\eps^2 L_i) = \Omega(|\mv(\mx^+)|\eps^2 L_i),\\
 	\end{split}
 \end{equation*} 
 as claimed. 
\end{proof}

Next, we show Item (3)  of \Cref{lm:Clustering} regarding the corrected potential changes of subgraphs in $\mathbb{X}$.

\begin{lemma}\label{lm:Step5PotentialChange}  $\Delta_{i+1}^+(\mx) \geq 0$ for every $\mx \in \mathbb{X}$, and
	\begin{equation*}
		\sum_{\mx \in \mathbb{X}\setminus \mathbb{X}^{\lowm}} \Delta_{i+1}^+(\mx) = \sum_{\mx \in \mathbb{X}\setminus \mathbb{X}^{\lowm}} \Omega(|\mv(\mx)|\eps^2 L_i). 
	\end{equation*}
\end{lemma}
\begin{proof}
	If $\mx \in \mathbb{X}_1\cup \mathbb{X}_2\cup \mathbb{X}_4$, then $\Delta^+_{i+1}(\mx)\geq 0$ by \Cref{lm:Step1Potential,lm:Step2Potential,lm:Step4Potential}. Otherwise, $\mx \in \mathbb{X}^{\prefix}_5\cup \mathbb{X}^{\internal}_5$, and hence $\mx$ is a subpath of $\msttilde_{i}$. Thus, by definition, $\Delta^+_{i+1}(\mx) = \sum_{\varphi\in \mx}\omega(\varphi) + \sum_{\mathbf{e}\in \me(\mx)\cap \msttilde_{i}}\omega(\mathbf{e}) - \adm(\mx) = 0$. That is, $\Delta^+_{i+1}(\mx)\geq 0$ in every case.
	
	We now show a lower bound on the average potential change of subgraphs in $\mathbb{X}\setminus \mathbb{X}^{\lowm}$. We assume that we are not in the degenerate case; otherwise, $\mathbb{X}\setminus \mathbb{X}^{\lowm} = \emptyset$ and there is nothing to prove. 	By Item (1) of \Cref{lm:Xlowm}, $\mathbb{X}^{\lowm} =\mathbb{X}^{\internal}_5$ and only subgraphs in $\mathbb{X}^{\prefix}_5$ may have negative potential change. By \Cref{lm:Step1Potential,lm:Step2Potential,lm:Step4Potential}, on average, each node $\varphi$ in any subgraph $\mx \in \mathbb{X}_1\cup \mathbb{X}_2\cup \mathbb{X}_4$ has $\Omega(\eps^2 L_i)$ corrected potential change, denoted by $\overline{\Delta}(\varphi)$.
	
	By construction, a subgraph in $\mathbb{X}^{\prefix}_{5}$ is a prefix (or suffix), say $\Ptilde_1$, of a long path $\Ptilde$. The other suffix, say $\Ptilde_2$, of $\Ptilde$ is attached to a subgraph, say $\mx \in \mathbb{X}_1\cup \mathbb{X}_2\cup \mathbb{X}_4$ by the construction of Step \hyperlink{5B}{5B} and Item (2) \Cref{lm:Clustering-Step3}. Since $|\mv(\Ptilde_2)| = \Omega(1/\eps)$ by Item (3)  of \Cref{lm:Clustering-Step5}, $\sum_{\varphi \in \Ptilde_2}\overline{\Delta}(\varphi) = \Omega(1/\eps)(\eps^2 L_i) = \Omega(\eps L_i)$. We distribute half this corrected potential change to all the nodes in $\Ptilde_1$, by Item (3)  of \Cref{lm:Clustering-Step5}, each gets $\Omega(\frac{\eps L_i}{1/\eps}) = \Omega(\eps^2 L_i)$. This implies:
		\begin{equation*}
			\sum_{\mx \in \mathbb{X}\setminus \mathbb{X}^{\lowm}} \Delta_{i+1}^+(\mx) = \sum_{\varphi \in \mv_i\setminus \mv^{\lowm}_i} \Omega(\eps^2 L_i) =  \sum_{\mx \in \mathbb{X}\setminus \mathbb{X}^{\lowm}} \Omega(|\mv(\mx)|\eps^2 L_i),
		\end{equation*}
	as desired.
\end{proof}

We are now ready to prove \Cref{lm:Clustering} that we restate below.

\Clustering*
\begin{proof}
	We observe that Items (1), (2) and (4) follow directly \Cref{lm:exception} and  \Cref{lm:Xlowm}. Item (5) follows from \Cref{lm:XProp}. Item (3) follows from \Cref{lm:Step5PotentialChange}. The construction time is asymptotically the same as the construction time of $\mathbb{X}$, which is $O((|\mv_i| + |\me_i|)\eps^{-1})$ by \Cref{lm:XProp}. 
	
	Finally, we compute the augmented diameter of each subgraph $\mx \in \mathbb{X}$. We observe that the augmentations in Step 3 and Step 5 do not create any cycle. Thus, if $\mx$ is initially formed in Steps 1, 2 or 5B, then $\mx$ is eventually a tree. It follows that the augmented diameter of $\mx$ can be computed in $O(|\mv(\mx)|)$ time by a simple tree traversal\footnote{The same algorithm as in \Cref{lm:level1Const} applies: root the tree at an arbitrary node and visit it in post-order. For each node $\varphi$ in the tree, keep track of (the weight of) the path in the subtree rooted at $\varphi$ that has maximum diameter and ends at $\varphi$.}. If $\mx$ is formed in Step 4, then it has exactly one edge $\mbe$ not in $\msttilde_{i}$ by Item (1) in \Cref{lm:Clustering-Step4}  and that  $\mx$ contains at most one cycle. Let $\mz$ be such a cycle (if any); $\mz$ has $O(1/\eps)$ edges by Item (3) in \Cref{lm:Clustering-Step4}. Thus, we can reduce computing the diameter of $\mx$ to computing the diameter of trees by guessing an edge of $\mz$ that does not belong to the diameter path of $\mx$ and remove this edge from $\mx$; the resulting graph is a tree. There are $O(\frac{1}{\eps})$ guesses, and for each guess, computing the diameter takes $O(|\mv(\mx)|)$ time, which implies $O(|\mv(\mx)|\eps^{-1})$ time\footnote{It is possible to compute the augmented diameter of $\mx$ in $O(|\mv(\mx)|)$ time using a more involved approach.} to compute $\adm(\mx)$. Thus, the total running time to compute the augmented diameter is $\sum_{\mx\in \mathbb{X}} O(|\mv(\mx)|\eps^{-1}) = O(|\mv_i|\eps^{-1})$.  
\end{proof}

\subsection{Completing the Proof of \texorpdfstring{\Cref{lm:Clustering-Step4}}{Step 4}} \label{subsec:ProofStep3}

In this section, we complete the proof of Item (4) in \Cref{lm:Clustering-Step4}. We consider two cases: (Case 1) $\mi(\nu) \cap \mi(\mu) =\emptyset$ and (Case 2) $\mi(\nu) \cap \mi(\mu) \not= \emptyset$. We reuse the notation in \Cref{lm:Clustering-Step4} here.

\begin{wrapfigure}{r}{0.5\textwidth}
	\begin{center}
		\includegraphics[width=0.35\textwidth]{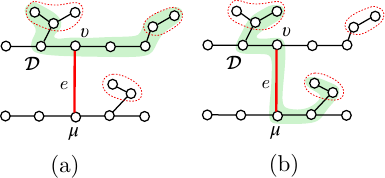}
	\end{center}
	\caption{\footnotesize{$\mathcal{D}$ is the diameter path and enclosed trees are augmented to a Step-4 subgraph in Step 5A. The green shaded regions contain nodes in $\mathcal{D}$.  (a) $\mathcal{D}$ does not contain $\mathbf{e}$. (b) $\mathcal{D}$ contains $\mathbf{e}$.}}
	\vspace{-15pt}
	\label{fig:Step3-diam}
\end{wrapfigure}

\paragraph{Case 1: $\mi(\nu) \cap \mi(\mu) =\emptyset$} Let $\mx = (\nu,\mu)\cup \mi(\nu) \cup \mi(\mu)$ where $\mbe = (\nu,\mu)$ is the only edge in $\me_i$ contained in $\mx$.  For any subgraph $\mz$ of $\mx$, we define: 
\begin{equation}\label{eq:Potential-Y-D}
	\Phi^+(\mz) = \sum_{\alpha \in \mz}\omega(\alpha) + \sum_{\mbe' \in \widetilde{\mst}_i\cap \me(\mz)} \omega(\mbe)
\end{equation}
to be the total weight of nodes and $\widetilde{\mst}_i$ edges in $\mz$.  Let $\md$ be an augmented diameter path of $\mx$, and $\my = \mx \setminus \mv(\md)$ be the subgraph obtained from $\mx$ by removing nodes on $\md$. Let $\mi(\nu)$ and $\mi(\mu)$ be two intervals in the construction in Step 4 that are connected by an edge $\mbe = (\nu,\mu)$.

\begin{claim}\label{clm:PotentialY-bound} $\Phi^+(\my) =  \frac{5 L_i}{4} +\Omega(|\mv(\my)|\epsilon L_i)$.
\end{claim}
\begin{proof}
	Let $\ma = \my \setminus (\mi(\nu)\cup \mi(\mu))$ be the subgraph of $\my$ obtained by removing every node in $\mi(\nu)\cup \mi(\mu)$ from $\my$, and $\mb = \my \cap (\mi(\nu)\cup \mi(\mu))$ be the subgraph of $\my$ induced by nodes of $\my$ in $(\mi(\nu)\cup \mi(\mu))$. Since every node has weight of at least $L_{i-1}$ by \hyperlink{P3'}{property (P3')}, we have 		\begin{equation}
		\Phi^+(\ma) \geq |\mv(\ma)| L_{i-1} = |\mv(\ma)| \epsilon L_i
	\end{equation}
	
	We consider two cases:
	\begin{itemize}
		\item \emph{Case 1: $\md$ does not contain the edge $(\nu,\mu)$.} See Figure~\ref{fig:Step3-diam}(a). In this case, $\md \subseteq  \widetilde{\mst}_i$, and  that $\mathcal{I}(\nu)\cap\md =\emptyset$ or $\mathcal{I}(\mu)\cap \md = \emptyset$ since $\mi(\nu)$ and $\mi(\mu)$ are connected only by $\mbe$. Focusing on $\mathcal{I}(\nu)$ (w.l.o.g), since $\mi(\nu)\subseteq \widetilde{\mst}_i$,  $\Phi^+(\mb)\geq \adm(\mathcal{I}(\nu)) \geq (2-(3g+2)\epsilon) L_i$ by Claim~\ref{clm:Interval-node}.
		
		\item  \emph{Case 2: $D$ contains the edge $(\nu,\mu)$.} See Figure~\ref{fig:Step3-diam}(b). In this case at least two sub-intervals, say $\mi_1,\mi_2$,  of four intervals $\{\mathcal{I}(\nu)\setminus \nu, \mathcal{I}(\mu)\setminus \mu\}$ are disjoint from $\md$.  
		By Claim~\ref{clm:Interval-node}, then $\Phi^+(\mb)\geq \adm(\mi_1) + \adm(\mi_2) \geq (2- 4(g+1)\epsilon) L_i$ by Claim~\ref{clm:Interval-node}.
	\end{itemize}	
	In both cases, $\Phi^+(\mb) \geq (2- 4(g+1)\epsilon) L_i \geq \frac{3 L_i}{2} $ when $\epsilon  \leq \frac{1}{8(g+1)}$.

	By Claim~\ref{clm:Interval-node}, $|\mv(\mb)|  = O(\epsilon^{-1})$. 	 This implies that:
	\begin{equation*}
		\begin{split}
			\Phi^+(\my) &= \Phi^+(\ma) + \Phi^+(\mb) \geq \Phi^+(\ma) + \frac{3 L_i}{2}~=~  \frac{5 L_i}{4} +  |\mv(\ma)|(\epsilon L_i) + \frac{L_i}{4}\\
			&= \frac{5 L_i}{4}+  |\mv(\ma)|(\epsilon L_i) + \Omega(|\mv(\mb)| \epsilon L_i)\\
			&=  \frac{5 L_i}{4}+ \Omega((|\mv(\ma)| + |\mv(\mb)|)\epsilon L_i) =\frac{5 L_i}{4}+ \Omega(|\mv(\my)|\epsilon L_i),
		\end{split}
	\end{equation*}
	which concludes the proof of Claim~\ref{clm:PotentialY-bound}.
\end{proof}

Note that $\mv(\md) \leq \frac{gL_i}{L_{i-1}}  = O(\epsilon^{-1})$ since every node has weight at least $L_{i-1}$ by \hyperlink{P3'}{property (P3')}. Thus, we have:
\begin{equation*}
	\begin{split}
		\Delta_{i}^+(\mx)  &= \Phi^+(\md) + \Phi^+(\my) - \adm(\mx) = \Phi(\my) - \omega(\mbe)\\ & \geq L_i/4 +  \Omega(|\mv(\my)|\epsilon L_i) \qquad \mbox{(by Claim~\ref{clm:PotentialY-bound})}\\
		&= \Omega( |\mv(\md)| \epsilon L_i) + \Omega(|\mv(\my)|\epsilon L_i)= \Omega(|\mv(\mx)|\epsilon L_i)~.
	\end{split}
\end{equation*}

Thus, Item (4) of \Cref{lm:Clustering-Step4} follows.

	\begin{figure}[!ht]
		\begin{center}
			\includegraphics[width=0.8\textwidth]{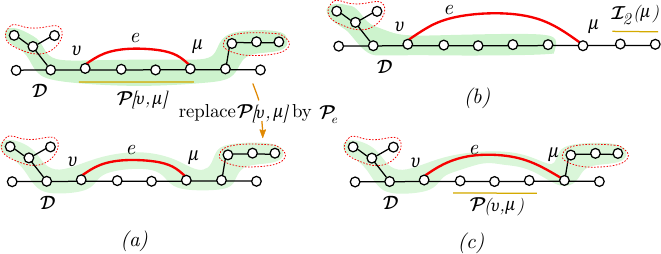}
		\end{center}
		\caption{Nodes enclosed in dashed red curves are attached to $\mx$ in Step 4.  }
		\label{fig:replace-e-vsPe}
	\end{figure}

\paragraph{Case 2: $\mi(\nu) \cap \mi(\mu) \not=\emptyset$}
	Let $\md$ be a diameter path of $\mx$, and $\my = \mx \setminus \mv(\md)$. Recall that $\mx$ contains only one edge $\mbe = (\nu,\mu) \in \me_i$ by item (1) of \Cref{lm:Clustering-Step4}. Let $\mathcal{P}_{\mbe} = (\nu,\mbe,\mu)$ be the path that consists of only edge $\mbe$ and its endpoints. Let  $\mathcal{P}[\nu,\mu]$ be the subpath of $\msttilde_{i}$ between $\nu$ and $\mu$.

	We observe that $\mbe$ is not removable by Item (3) of \Cref{def:GiProp}, and by the fact that the path  $\mathcal{P}[\nu,\mu]$ is a path in $\msttilde_i$ in which every node has degree at most $2$ in $\msttilde_{i}$ (see (2a) in \Cref{lm:Clustering-Step3}), $\omega(\mathcal{P}[\nu,\mu]) \geq t(1+6g\eps)\omega(\mbe)  \geq (1+6g\eps)\omega(\mbe) $ since $t\geq 1$. Then it follows that:

\begin{equation}\label{eq:p-vs-pe}
		\begin{split}
				\omega(\mathcal{P}[\nu,\mu]) - 	\omega(\mathcal{P}_{\mbe}) ) &> 6g\epsilon  \cdot \omega(\mbe)    - w(\nu) - w(\mu) \\
				&> 6g\epsilon L_i/2 - 2g\epsilon L_i= g\epsilon L_i
		\end{split}
	\end{equation}
	In particular, this means that $	\omega(\mathcal{P}(\nu,\mu)) \geq 	\omega(\mbe)$.

	Thus, if $\md$ contains both $\nu$ and $\mu$, then it must contain $\mbe$, since otherwise, $\md$ must contain $\mathcal{P}[\nu,\mu]$ and by replacing $\mathcal{P}[\nu,\mu]$ with $\mathcal{P}_{\mbe}$ we obtain a shorter path by Equation~\eqref{eq:p-vs-pe} (see Figure~\ref{fig:replace-e-vsPe}(a)).

	\begin{claim}\label{clm:D-P-size}
		$|\mv(\mp[\nu,\mu])|\leq \frac{4}{\epsilon}$ and $|\mv(\md)|\leq \frac{g}{\epsilon}$.
	\end{claim}
	\begin{proof}
		Observe that $\adm(\mp[\nu,\mu])  \leq 4L_i$ since $\mp[\nu,\mu]\subseteq \mi(\nu)\cup \mi(\mu)$. Thus,  $|\mv(\mp[\nu,\mu])| \leq \frac{4L_i}{L_{i-1}} = \frac{4}{\epsilon}$ since each node of $\mp[\nu,\mu]$ has weight of at least $L_{i-1}$ by \hyperlink{P3'}{property (P3')}.  Similarly, $\adm(\md)\leq gL_i$ by \hyperlink{P3'}{property (P3')} while each node has weight at least $L_{i-1}$. Thus,  $|\mv(\md)|\leq \frac{gL_I}{L_{i-1}} = \frac{g}{\epsilon}$. 
	\end{proof}

	\noindent We consider two cases:
	\begin{itemize}
		\item \emph{Case 1} If $\md$ does not contain edge $\mbe$ (see Figure~\ref{fig:replace-e-vsPe}(b)), then (a) $\md\subseteq \widetilde{\mst}_i$  and (b) $|\{\nu,\mu\} \cap \mathcal{D}| \leq 1$. From (a), we have: 
        \begin{equation}\label{eq:Dmx-vsmy}
			\begin{split}
				\Delta_{i+1}^+(\mx) &\geq \adm(\md) + \Phi^+(\my) - \adm(\mx) =  \Phi^+(\my) 
				\end{split}
		\end{equation}
    Suppose w.l.o.g.\ that $\nu \in \mathcal{D}$. Let $\mi_2(\mu)$ be a subpath of  $\mi(\mu)\setminus \{\mu\}$ such that  $\mi_2(\mu)\cap \md = \emptyset$; $\mi_2(\mu)$ exists since $\mu \not\in \mathcal{D}$.  Also recall by item (b) of \Cref{clm:Interval-node} that $\adm(\mi_2(\mu)) \geq (1-2(g+1)\epsilon)L_i \geq L_i/2$ when $\eps \leq \frac{1}{4(g+1)}$. Also, by item (b) of \Cref{clm:Interval-node} we have $|\mv (\mi_2(\mu))|  = O(1/\eps)$ and hence $|\mv(\mi_2(\mu))| \eps L_i = O(L_i)$. Continuing \Cref{eq:Dmx-vsmy}, we have:
    
    \begin{equation}
    \begin{split}
         \Delta_{i+1}^+(\mx) &\geq  \Phi^+(\my) \geq \adm(\mi_2(\mu))+ \Phi^+(\my\setminus \mi_2(\mu))\\
          &\geq L_i/2 + |\mv(\my\setminus \mi_2(\mu))|\epsilon L_{i}\\
          &= L_i/4 +  \Omega(|\mv(\mi_2(\mu))| + |\mv(\my\setminus \mi_2(\mu))|) \eps L_i\\
          &=   L_i/4 +  \Omega(|\mv(\my)|)\eps L_{i}\\
          &= \Omega(|\mv(\md)|)\eps L_i +   \Omega(|\mv(\my)|)\eps L_{i}  \qquad \mbox{(by Claim~\ref{clm:D-P-size})} \\
          &= \Omega(|\mv(\mx)|\epsilon L_i)
    \end{split}
    \end{equation}
		
		\item \emph{Case 2} If $\md$ contains $\mbe$  (see Figure~\ref{fig:replace-e-vsPe}(c)), then $\md \cap \mp(\nu,\mu) = \emptyset$; here $\mp(\nu,\mu)$ is the path obtained from $\mp[\nu,\mu]$ by removing its endpoints. It follows that
		\begin{equation}
			\begin{split}
				\Delta_{i+1}^+(\mx) &\geq \adm(\md) + \Phi^+(\my) - \adm(\mx) =  \Phi(\my) - w(\mbe) \\
				&\geq \adm(\mp[\mu,\nu]) + \Phi^+(\my\setminus \mp[\mu,\nu]) - w(\mbe)\\
				&\geq g\epsilon L_i 		+  |\mv(\my\setminus \mp[\mu,\nu])|L_{i-1} \qquad\mbox{(by Equation~\eqref{eq:p-vs-pe})} \\
				&\geq \Omega((|\mv(\mp[\mu,\nu])| + |\mv(\md)|)\epsilon^2 L_i) +   |\mv(\my\setminus \mp[\mu,\nu])|\epsilon L_{i} \\
				&= \Omega(|\mv(\mx)|\epsilon^2 L_i)
			\end{split}
		\end{equation}
	\end{itemize}
    where the penultimate inequality is due to Claim~\ref{clm:D-P-size}. In both cases, we have $\Delta^+_{i+1}(\mx) = \Omega(|\mv(\mx)|\epsilon^2 L_i)$ as claimed in Item (4) of \Cref{lm:Clustering-Step4}.

\paragraph{Acknowledgement.~}  Hung Le is supported by the NSF CAREER award CCF-2237288, the NSF grants CCF-2121952 and CCF-2517033, and a Google Research Scholar Award. Shay Solomon is funded by the European Union (ERC, DynOpt, 101043159). Views and opinions expressed are however those of the author(s) only and do not necessarily reflect those of the European Union or the European Research Council. Neither the European Union nor the granting authority can be held responsible for them. Shay Solomon is also funded by a grant from the United States-Israel Binational Science Foundation (BSF), Jerusalem, Israel, and the United States National Science Foundation (NSF). Shay Solomon was also funded by the Israel Science Foundation grant No.1991/19 when this work was done. We thank Oded Goldreich for his suggestions concerning the presentation of this work, and we thank Lazar Milenković for his support. We thank anonymous reviewers for their exceptionally thorough comments on the presentation of this paper. 

\bibliographystyle{plain}
\bibliography{spanner}
\appendix

\section{The Algebraic Computation Tree Model} \label{algapp}
 In this appendix, we give a brief description of the {\em algebraic computation tree (ACT)} model. 
 (Refer to \cite{BenOr83} and Chapter 3 in the book \cite{NS07} for a more detailed description.)

An ACT is a binary computation tree where each leaf is associated with an output and each internal node is either (\emph{i}) labeled with a variable $f_x$ determined by $f_x = a_1\circ a_2$ or $f_x = \sqrt{a_1}$ where $\circ \in \{+,-,\times,\div\}$ and each $a_i$, $i \in \{1,2\}$, is either a value of a proper ancestor of $x$, an input element or a constant in $\mathbb{R}$, or (\emph{ii}) labeled with a comparison $a \bowtie 0$, where $a$ is either a value of a proper ancestor of $x$, or an input element, and the left (resp. right) child is labeled with ``$\le$''  (resp. ``$>$''). An ACT tree $T$ corresponds to an algorithm $\mathcal{A}_T$, which traverses a path down the tree starting at the root and either (\emph{i}) evaluates $f_x$ if the node has one child, or (\emph{ii}) selects left or right child depending on the outcome of the comparison. When $\mathcal{A}_T$ reaches the leaf, it evaluates the expression by replacing all the variables with the input values and terminates. It is required that no input lead to undefined behavior, such as division by 0, or taking a square root of a negative number; furthermore, for each leaf $w$ there should be an input on which $\mathcal{A}_T$ terminates in $w$. 
A problem $\mathcal{P}$ is solvable in the ACT model if there exists an ACT tree $T$ such that for any valid input for $\mathcal{P}$, $\mathcal{A}_T$ returns the value of $\mathcal{P}$ on that input. 

The ACT model is particularly convenient for proving lower bounds, as it represents explicitly all possible execution paths of an algorithm. For most algorithmic purposes, one can alternatively consider the real RAM model; one can prove an equivalence between the two models, which holds up to some subtle issues of uniformity that lie outside the scope of the current paper.

\end{document}